\documentclass[10pt,reqno]{amsart}
\usepackage{amsaddr}
\usepackage{etoolbox}
\patchcmd{\section}{\scshape}{\bfseries\scshape}{}{}
\makeatletter
\renewcommand{\@secnumfont}{\bfseries}
\makeatother

\pdfoutput=1 
\usepackage[a4paper,margin=3cm]{geometry}

\usepackage[colorinlistoftodos]{todonotes}

\usepackage[T1]{fontenc}
\usepackage[utf8]{inputenc}
\usepackage[pagebackref]{hyperref}\hypersetup{pdfborder=0 0 0}
\usepackage[francais,english]{babel}
\usepackage{latexsym,amssymb,graphicx,caption,listings}

\usepackage{amsmath} 
\usepackage{enumitem} 
\usepackage{dsfont} 
\usepackage{graphicx} 
\usepackage{amsthm} 
\usepackage{subcaption}
\usepackage{amssymb} 

\newtheorem{theorem}{Theorem}[section]%
\newtheorem{corollary}[theorem]{Corollary}%
\newtheorem{proposition}[theorem]{Proposition}%
\newtheorem{lemma}[theorem]{Lemma}%

\newtheorem{definition}[theorem]{Definition}%
\newtheorem{example}[theorem]{Example}%
\newtheorem{remark}[theorem]{Remark}%

\numberwithin{equation}{section}

\newcommand{\N}{\mathbb{N}}
\newcommand{\Z}{\mathbb{Z}}
\newcommand{\E}{\mathbb{E}}
\newcommand{\R}{\mathbb{R}}
\newcommand{\T}{\mathbb{T}}

\renewcommand{\P}{\mathds{P}}

\newcommand{\cN}{\mathcal{N}}
\newcommand{\cE}{\mathcal{E}}

\def\bfone{{\boldsymbol 1}}

\newcommand{\var}{{\rm var\,}}
\newcommand{\diff}{\mathrm{d}}
\renewcommand{\leq}{\leqslant}
\renewcommand{\geq}{\geqslant}

\usepackage{xcolor}

\title[Accelerated~gossip~in~networks~of~given~dimension]{Accelerated~gossip~in~networks~of~given~dimension using~Jacobi~polynomial~iterations}

\author{Raphaël Berthier, Francis Bach, Pierre Gaillard}

\address{INRIA - Département d’informatique de l’ENS \\
	Ecole normale supérieure, CNRS, INRIA \\ 
	PSL Research University, 75005 Paris, France}

\email{raphael.berthier@inria.fr}
\email{francis.bach@inria.fr}
\email{pierre.gaillard@inria.fr}

\begin{document}

\begin{abstract}
	Consider a network of agents connected by communication links, where each agent holds a real value. The gossip problem consists in estimating the average of the values diffused in the network in a distributed manner. We develop a method solving the gossip problem that depends only on the spectral dimension of the network, that is, in the communication network set-up, the dimension of the space in which the agents live. This contrasts with previous work that required the spectral gap of the network as a parameter, or suffered from slow mixing. Our method shows an important improvement over existing algorithms in the non-asymptotic regime, i.e., when the values are far from being fully mixed in the network. Our approach stems from a polynomial-based point of view on gossip algorithms, as well as an approximation of the spectral measure of the graphs with a Jacobi measure. We show the power of the approach with simulations on various graphs, and with performance guarantees on graphs of known spectral dimension, such as grids and random percolation bonds. An extension of this work to distributed Laplacian solvers is discussed. As a side result, we also use the polynomial-based point of view to show the convergence of the message passing algorithm for gossip of \cite{moallemi2005consensus} on regular graphs. The explicit computation of the rate of the convergence shows that message passing has a slow rate of convergence on graphs with small spectral gap. 
\end{abstract}

\maketitle

\section{Introduction}

The averaging problem, or gossip problem, is a fundamental primitive of distributed algorithms. Given a network composed of agents and undirected communication links between them, we assign to each agent $v$ a real value $\xi_v$, called an observation. The goal is to design an iterative communication procedure allowing each agent to know the average of the initial observations in the network, as quickly as possible.

The landmark paper \cite{boyd2006randomized} suggests the natural following protocol to solve the averaging problem: at each iteration, each agent replaces his current observation by some average of the observations of its neighbors in the network. We will refer to this method in the following by the term \emph{simple gossip}. More precisely, we are given a weight matrix $W = (W_{v,w})_{v,w \in V}$, called the gossip matrix, indexed by the vertices $v \in V$ of the network graph, satisfying the property that $W_{v,w}$ is non-zero only if $v \sim w$, that is $v$ and $w$ are connected in the graph. Then the simple gossip iteration writes
\begin{align}
\label{eq:iteration-simple-gossip}
&x^{t=0}_v = \xi_v \, , &&x^{t+1}_v = \sum_{w: w \sim v} W_{v,w} x_w^t \, , \qquad t \geq 0 \, .
\end{align}
The paper \cite{boyd2006randomized} proves the linear convergence of the observations to their average.

However, the rate of the linear convergence was shown to worsen significantly in many networks of interest as the size of the network increases. More precisely, define the diameter $D$ of the network as the largest number of communication links needed to connect any two agents. While obviously, $D$ steps of averaging are needed for any gossip method to spread information in the network, the simple gossip method may require up to $\Theta(D^2)$ communication steps to estimate the average, as for instance on the line graph or the two-dimensional grid \cite{rebeschini2017accelerated}. To reach the $O(D)$ bound, a diverse set of ideas were proposed, including second-order recursions \cite{cao2006accelerated,rebeschini2017accelerated}, message passing algorithms \cite{moallemi2005consensus}, lifted Markov chain techniques \cite{shah2009gossip}, methods using Chebychev polynomial iterations \cite{Arioli:2014:CAI:2638909.2639021,scaman2017optimal} or inspiration arising from advection-diffusion processes \cite{sardellitti2010fast}. To our knowledge, all of these accelerated methods assume that the agents hold additional information about the network graph, such as its spectral gap. For instance, second-order methods typically take the form (see \cite{cao2006accelerated})
\begin{equation}
\label{eq:iteration-second-order-gossip}
\begin{aligned}
&x^{t=0}_v = \xi_v \, , \qquad x^{t=1}_v = \sum_{w: w \sim v} W_{v,w} x_w^0 \, , \\
&x^{t+1}_v = \omega\sum_{w: w \sim v} W_{v,w} x_w^{t} + (1-\omega)x^{t-1}_v \, , \qquad t \geq 1 \, ,
\end{aligned}
\end{equation}
where $\omega$ is some simple function of the spectral gap $\gamma$, that is the distance between the largest and the second largest eigenvalue of $W$. This iteration obtains optimal asymptotic convergence on many graphs, with a relaxation time of the linear convergence that scales like $1/\sqrt{\gamma}$ as $\gamma \to 0$.

In this paper, we develop a gossip method based not on the spectral gap $\gamma$, but on the density of eigenvalues of $W$ near the upper edge of the spectrum. Looking at the upper part of the spectrum at a broader scale allows us to improve the local averaging of the gossip algorithm in the regime $t < 1/\sqrt{\gamma}$. This improvement is worthy as the spectral gap $\gamma$ can get arbitrarily small in large graphs. For instance, in the case of the line graph or the two-dimensional grid, the relaxation time $1/\sqrt{\gamma}$ scales like the diameter $D$ of the graph. Thus the regime $t < 1/\sqrt{\gamma}$ can be relevant for applications.

Remarkably, the spectral density of $W$ near the upper edge can be described by a very natural parameter: the spectral dimension $d$. The network is of spectral dimension $d$ if the number of eigenvalues of $W$ in $[1-E,1]$ decreases like $E^{d/2}$ for small $E$ ($\gamma \ll E \ll 1$), see \ref{sec:spectral-dimension-infinite-graph} for rigorous definitions. We will see with examples that this definition coincides with our intuition of the dimension of the graph, which is the dimension of the manifold on which the agents live. For instance, the grid with nodes $\Z^d$ where the nodes at distance $1$ are connected, is a graph of dimension $d$. Thus the parameter $d$ is much easier to know than the spectral gap $\gamma$.

In real-world situations, the practitioner reasonably knows if the network on which she implements the gossip method is of finite dimension, and if so, she also knows the dimension $d$. In this paper, we argue that she should run a second-order iteration with time-dependent weights
\begin{equation}
\label{eq:iteration}
\begin{aligned}
&x_v^{t=0} = \xi_v \, , \qquad x^{t=1}_v = a_0\sum_{w: w \sim v} W_{v,w} x_w^0 + b_0x_v^0 \, , \\
&x^{t+1} = a_t\sum_{w: w \sim v} W_{v,w} x_w^{t}
+ b_t x^t_v - c_t x^{t-1}_v \, , \quad t \geq 1 \, ,
\end{aligned}
\end{equation}
where the recurrence weights $a_t, b_t, c_t$ are given by the formulas
\begin{equation}
\label{eq:coeff}
\begin{aligned}
&a_0 = \frac{d+4}{2(2+d)} \, , &&b_0 = \frac{d}{2(2+d)} \, , \\
&a_t = \frac{(2t+d/2+1)(2t+d/2+2)}{2(t+1+d/2)^2} \, , &&b_t = \frac{d^2(2t+d/2+1)}{8(t+1+d/2)^2(2t+d/2)} \, , \\
&c_t = \frac{t^2(2t+d/2+2)}{(t+1+d/2)^2(2t+d/2)} \, , \qquad t \geq 1 \, . 
\end{aligned}
\end{equation}
The motivation for these choice of weights $a_t, b_t, c_t$ should not be obvious at first sight. It follows from a \emph{polynomial-based} point of view on gossip algorithms: it consists in seeing the iterations \eqref{eq:iteration-simple-gossip}, \eqref{eq:iteration-second-order-gossip} and \eqref{eq:iteration} as sequences $P_0,P_1,P_2,\dots$ of polynomials in the gossip matrix $W$. The correspondence is given by the relation $x^t = P_t(W)\xi$ where $x^t = (x^t_v)_{v\in V}$ and $\xi = (\xi_v)_{v\in V}$. This approach is inspired from similar work done in the resolution of linear systems \cite{fischer1996polynomial} and on the load balancing problem \cite{diekmann1999efficient}. The choice of an iteration is reframed as the choice of a sequence of polynomials, and the performance of the resulting gossip method depends on the spectrum of $W$. As the dimension of the graph gives the rate of decrease of the spectral density near the edge of the spectrum, it suggests the sequence of polynomials one should take: we choose a parametrized sequence of polynomials called \emph{Jacobi polynomials}, that is well-known in the literature on orthogonal polynomials (see the Definition \ref{def:Jacobi-polynomials} of the Jacobi polynomials). This actually leads to the iteration \eqref{eq:iteration}, that we call the \emph{Jacobi polynomial iteration}.

The Jacobi polynomial iteration \eqref{eq:iteration} improves the convergence of the gossip method in the transitive phase $t < 1/\sqrt{\gamma}$, but looses the optimal rate of convergence of second-order gossip, because it does not use the spectral gap $\gamma$. We argue that in most applications of gossip methods, the asymptotic rate of convergence is not relevant as there is noise in the initial data $\xi$, thus a high precision on the result would be useless. However, we also build a gossip iteration that uses both parameters $d$ and $\gamma$ and achieves both the efficiency in the non-transitive regime and the fast rate of convergence.

This resolution of the gossip problem with inner-product free polynomial-based iterations is new, and could lead to other interesting algorithms on other types of graphs. Here, the phrase ``inner-product free'' comes from the literature on polynomial-based iterations for linear systems \cite{fischer1996polynomial}, and refers to the fact that recurrence coefficients $a_t, b_t, c_t$ are computed without using the gossip matrix $W$ (but parametrized using the knowledge of $d$). Indeed, as the knowledge of the gossip matrix $W$ is distributed across the graph, it would be a challenging distributed problem to compute the recurrence coefficients if they depended on $W$. 

Although our work is inspired by iterative methods for linear systems, the Jacobi iteration that we developed for gossip can be transposed into a new idea to this literature, which can be useful for the distributed resolution of Laplacian systems over multi-agent networks.
 
Finally, in Section \ref{sec:MP}, we show that the message passing gossip iteration of \cite{moallemi2005consensus} can be interpreted as an inner-product free polynomial iteration. This point of view allows to derive convergence rates of the message passing gossip on regular graphs. 

\medskip
\textbf{Outline of the paper.} Section \ref{sec:problem_setting} sets some notation used in the remainder of this paper.
In Section \ref{sec:simulations}, we give simulations in different types of networks of dimension $2$ and $3$. We show that the recursion \eqref{eq:iteration} brings important benefits over existing methods in the non-asymptotic regime, i.e., when the observations are far from being fully mixed in the graph. 

In Sections \ref{sec:second_order_gossip}-\ref{sec:Jacobi-polynomial-gossip}, we develop the derivation of the Jacobi polynomial iteration. Section \ref{sec:second_order_gossip} describes an optimal way to design polynomial-based gossip algorithms, following the lines of \cite{fischer1996polynomial} and  \cite{diekmann1999efficient}, and discusses its feasibility. Section \ref{sec:Jacobi-polynomial-gossip} uses the notion of spectral dimension of a graph to inspire the practical Jacobi polynomial iteration \eqref{eq:iteration}.

In Section \ref{sec:performance-guarantees}, we prove some performance guarantees of the Jacobi polynomial iteration \eqref{eq:iteration} under the assumption that the graph has spectral dimension $d$. As a corollary, we get performance results on two types of infinite graphs: the $d$-dimensional grid $\Z^d$ and supercritical percolation bonds in dimension $d$. This supports that the iteration \eqref{eq:iteration} is robust to local perturbations of a graph. 

In Section \ref{sec:jacobi-with-spectral-gap}, we present the adaptation of the Jacobi polynomial iteration to the case where the spectral gap $\gamma$ of $W$ is given to improve the asymptotic rate of convergence. 

In Section \ref{sec:laplacian-solvers}, we describe the parallel between gossip methods and iterative methods for linear systems, and discuss the contributions that our work can bring to the distributed resolution of Laplacian systems over networks. 

In Section \ref{sec:MP}, we show how the message passing gossip algorithm can be interpreted as a polynomial gossip algorithm. We give the convergence rate of message passing in terms of the spectral gap~$\gamma$.

\medskip
\textbf{Code.} The code that generated the simulation results and the figures of this paper is available on the GitHub page \url{https://github.com/raphael-berthier/jacobi-polynomial-iterations}.

\section{Problem setting}
\label{sec:problem_setting}

A network of agents is modeled by an undirected finite graph $G = (V,E)$, where $V$ is the set of vertices of the graph, or agents, and $E$ the set of edges, or communication links. We assume each agent $v$ holds a real value $\xi_v$. Our goal is to design an iterative algorithm that quickly gives each agent the average $\bar{\xi} = (1/n)\sum_{v \in V} \xi_v$, where $n = |V|$ is the number of agents. 

A fundamental operation to estimate the average $\bar{\xi}$ consists in averaging the observations of neighbors in the network. We formalize this notion using a gossip matrix.
\begin{definition}
	\label{def:gossip-matrix}
	A \emph{gossip matrix} $W = (W_{v,w})_{v,w \in V}$ on the graph $G$ is a matrix with entries indexed by the vertices of the graph satisfying the following properties:
	\begin{itemize}[label={--},noitemsep,nolistsep]
		\item $W$ is nonnegative: for all $v,w \in V$, $W_{v,w} \geq 0$.
		\item $W$ is supported by the graph $G$: for all distinct vertices $v,w$ such that $W_{v,w} > 0$, $\{v,w\}$ must be an edge of $G$.
		\item $W$ is stochastic: for all $v \in V$, $\sum_{w \in V} W_{v,w} = 1$.
		\item $W$ is symmetric: for all $v,w \in V$, $W_{v,w} = W_{w,v}$.
	\end{itemize}
\end{definition}
If $W$ is a gossip matrix and $x = (x_v)_{v \in V}$ is a set of values stored by the agents $v$, the product $Wx$ is interpreted as the computation by each agent $v$ of a weighted average of the values $x_w$ of its neighbors~$w$ in the graph (and of its own value $x_v$). This average is computed simultaneously for all agents $v$; indeed in this paper we deal only with \emph{synchronous} gossip. Note that we do not need the symmetry assumption on $W$ to interpret $W$ as an averaging operation. This assumption is usual in gossip frameworks as it allows one to use the spectral theory for $W$, on which our analysis relies heavily. It appears, for instance, in the works \cite{boyd2006randomized,cao2006accelerated,rebeschini2017accelerated}.

In a $d$-regular graph $G$ ($\forall v, \deg v\!=\!d$), a typical gossip matrix is $W = A(G)/d = (\bfone_{\{\{v,w\} \in E\}}/d)_{v,w \in V}$ where $A(G)$ is the adjacency matrix of the graph. More generally, if the graph has all vertices of degree bounded by some quantity $d_{\rm max}$, then a natural gossip matrix is 
\begin{equation}
\label{eq:natural_gossip_matrix}
W = I + \frac{1}{d_{\rm max}}(A-D) \, ,
\end{equation}
where $D$ is the degree matrix, which is the diagonal matrix such that $D_{v,v} = \deg v$.

\begin{definition}[Spectral gap]
Denote $\lambda_1 \geq \lambda_2 \geq \dots \geq \lambda_n$ the real eigenvalues of the symmetric matrix $W$. As $W$ is stochastic, $W\bfone = \bfone$; we can take $\lambda_1 = 1$, that corresponds to the eigenvector $\bfone = (1,\dots,1)$. According to the Perron-Frobenius theorem, all eigenvalues must be smaller than $1$ in magnitude. We define:
\begin{enumerate}
	\item the spectral gap $\gamma = 1-\lambda_2$ as the distance between the two largest eigenvalues of $W$, 
	\item the absolute spectral gap $\tilde{\gamma} = \min(1-\lambda_2,\lambda_n+1)$ as the difference between the moduli of the two largest eigenvalues of $W$ in magnitude.
\end{enumerate}
\end{definition}

\medskip
We now discuss different iterations for the gossip problem.
\medskip

\textbf{Simple gossip.} Simple gossip is a natural algorithm solving the gossip problem that consists in averaging repeatedly values in the graph \cite{boyd2006randomized}. More precisely, we choose a gossip matrix~$W$ on the graph~$G$, initialize $x^0 = \xi = (\xi_v)_{v \in V}$, and at each communication round $t$, compute 
\begin{equation}
\label{eq:simple_gossip}
x^{t+1} = W x^t \, . 
\end{equation}
Note that the latter equation is simply a compact rewriting of \eqref{eq:iteration-simple-gossip}. We can rewrite this iteration as $x^t = W^t \xi$. Note that in this last equation, we used the notation $.^t$ to denote both the index of $x$ and the power of the square matrix $W$. We will frequently make use of the indexation $.^t$ when vectors indexed by the vertices (or the edges) also depend on time. 

The speed of convergence of this method is studied in \cite{boyd2006randomized}. We give a summary here:
\begin{proposition}[from {\cite{boyd2006randomized}}]
Let $\xi$ be an arbitrary family of initial observations and $x^t$ the iterates of simple gossip defined by \eqref{eq:simple_gossip}. Denote $\tilde{\gamma}$ the absolute spectral gap of $W$. Then 
\begin{equation*}
\limsup_{t\to\infty} \Vert x^t-\bar{\xi}\bfone\Vert_2^{1/t} \leq 1-\tilde{\gamma} \, .
\end{equation*}
Moreover, the upper bound is reached if there exists an eigenvector $u$ of $W$, corresponding to an eigenvalue of magnitude $1-\tilde{\gamma}$, such that $\langle \xi,u\rangle \neq 0$.
\end{proposition}

\medskip
\textbf{Shift-register gossip.}  Several acceleration schemes of gossip \cite{cao2006accelerated,rebeschini2017accelerated} store some past iterates to compute higher-order recursions (that thus depend on powers of $W$). For instance, the shift-register iteration of \cite{cao2006accelerated} is of the form
\begin{align}
\label{eq:shift-register}
&x^0 = \xi \, ,  &&x^1 = W \xi \, , && x^{t+1} = \omega W x^t + (1-\omega)x^{t-1} \, ,
\end{align}
where $\omega$ is a parameter that needs to be tuned.
\begin{proposition}[from {\cite[Theorem 2]{liu2013analysis}}]
Let $\xi$ be an arbitrary family of initial observations and $x^t$ the iterates of shift-register gossip defined in \eqref{eq:shift-register} with parameter 
\begin{equation*}
\label{eq:tuning-shift-register}
\omega = 2\frac{1-\sqrt{\gamma(1-\gamma/4)}}{\left(1-\gamma/2\right)^2}\, ,
\end{equation*} 
where $\gamma$ is the spectral gap of the gossip matrix $W$. Then
\begin{equation*}
\limsup_{t\to\infty} \Vert x^t-\bar{\xi}\bfone\Vert_2^{1/t} \leq 1-2\frac{\sqrt{\gamma(1-\gamma/4)}-\gamma/2}{1-\gamma} \, .
\end{equation*}
Moreover, the upper bound is reached if there exists an eigenvector $u$ of $W$, corresponding to the eigenvalue $1-\gamma$, such that $\langle \xi,u\rangle \neq 0$.
\end{proposition}
The important consequence of this result is that the rate of convergence of the shift-register method behaves like $1-2\sqrt{\gamma}+o(\sqrt{\gamma})$ as $\gamma\to 0$. This differs from simple gossip where the rate of convergence behaves like $1-\gamma$. This means that in graphs with a small spectral gap, shift-register enjoys a much better rate of convergence that simple gossip: this is why we say that shift-register enjoys an \emph{accelerated} rate of convergence as opposed to simple gossip which has a \emph{diffusive} or \emph{unaccelerated} rate. This effect on the asymptotic rate of convergence can be seen in Figures~\ref{fig:geometric_graph_simulation} and~\ref{fig:sim_2D_grid_log}.

\medskip
\textbf{Polynomial gossip.} More abstractly, we define a polynomial gossip method as any method combining the past iterates of the simple gossip method:
\begin{equation}
\label{eq:polynomial_gossip}
x^t = P_t(W) \xi \, ,
\end{equation}
where $P_t$ is a polynomial of degree smaller or equal to $t$ satisfying $P_t(1) = 1$. The constraint $P_t(1) = 1$ ensures that $x^t = \bar{\xi}\bfone$ if all initial observations are the same, i.e., $\xi = \bar{\xi}\bfone$. The constraint $\deg P_t \leq t$ ensures that the iterate $x^t$ can be computed in $t$ time steps. Simple gossip corresponds to the particular case of the polynomial $P_t(\lambda)=\lambda^t$. Shift-register gossip is a polynomial gossip method whose corresponding polynomials that can be expressed using the Chebyshev polynomials (see Proposition \ref{prop:polynomials-shift-register-chebyshev}). The method \eqref{eq:iteration} will be derived as the polynomial iteration corresponding to some Jacobi polynomials. 

In this paper, we design polynomial gossip methods whose polynomials $P_t$, $t \geq 0$ satisfy a second-order recursion. This key property ensures that the resulting iterates $x^t = P_t(W)\xi$ can be computed recursively. 

\section{Simulations: comparison of simple gossip, shift-register gossip and the Jacobi polynomial iteration}
\label{sec:simulations}

In this section, we run our methods on grids, percolation bonds and random geometric graphs; the latter being a widely used model for real-world networks \cite[Section 1.1]{penrose2003random}. In each case, we consider both the two-dimensional (2D) structure and its three-dimensional (3D) counterpart. We refer to Figure \ref{fig:graph_drawings} for visualizations of the 2D structures, and to Appendix \ref{ap:details-simulations} for details about the parameters used. 

We compare our Jacobi polynomial iteration \eqref{eq:iteration} with the simple gossip method \eqref{eq:iteration-simple-gossip} and the shift-register algorithm \eqref{eq:iteration-second-order-gossip}. We found experimentally that the behavior of the shift-register algorithm was typical of methods based on the spectral gap such as the splitting algorithm of \cite{rebeschini2017accelerated} or the Chebychev polynomial acceleration scheme \cite{Arioli:2014:CAI:2638909.2639021,scaman2017optimal}; to avoid redundancy we do not present the similar behavior of these methods. We also compare with local averaging, which is given by the formula
\begin{equation*}
x^t_v = \frac{1}{|B_t(v)|} \sum_{w \in B_t(v)} \xi_w \, ,
\end{equation*} 
where $|B_t(v)|$ denotes the ball in $G$, centered in $v$, of radius $t$, for the shortest path distance. Note that local averaging does not correspond in general to any computationally cheap iteration, as opposed to the algorithms we present here. Thus it should not be considered as a gossip method, but rather as a lower bound on the performance achievable by any gossip method. (This is made fully rigorous in the statistical gossip framework of Section \ref{sec:performance-guarantees}.)

In our simulations, we change the graph $G$ that we run our algorithms on, but we always sample $\xi_v \sim_{i.i.d.} \cN(0,1), v \in V$ and measure the performance of gossip methods through the quantity $\Vert x^t - \bar{\xi}\bfone \Vert_2/\sqrt{n}$. Thus the performance of the algorithms is random because the initial values $\xi_v$ are random, and also because percolation bonds and random geometric graphs are random. The results we present here are averaged over $10$ realizations of the graph and the initial values, which is sufficient to give stable results.

\smallskip
\textbf{Tuning.} The optimal tuning of the shift-register gossip method as a function of the spectral gap was determined in \cite[Theorem 2]{liu2013analysis}, it is given by the formula \eqref{eq:tuning-shift-register}; this is the tuning that we use in our simulations. The Jacobi polynomial iteration is tuned by choosing $d=2$ in 2D grid, 2D percolation bonds and 2D random geometric graphs, and $d=3$ for their 3D analogs. 
\smallskip

\begin{figure}
	\begin{subfigure}{0.32\linewidth}
		\includegraphics[width=\linewidth, trim = 0 0 0 0]{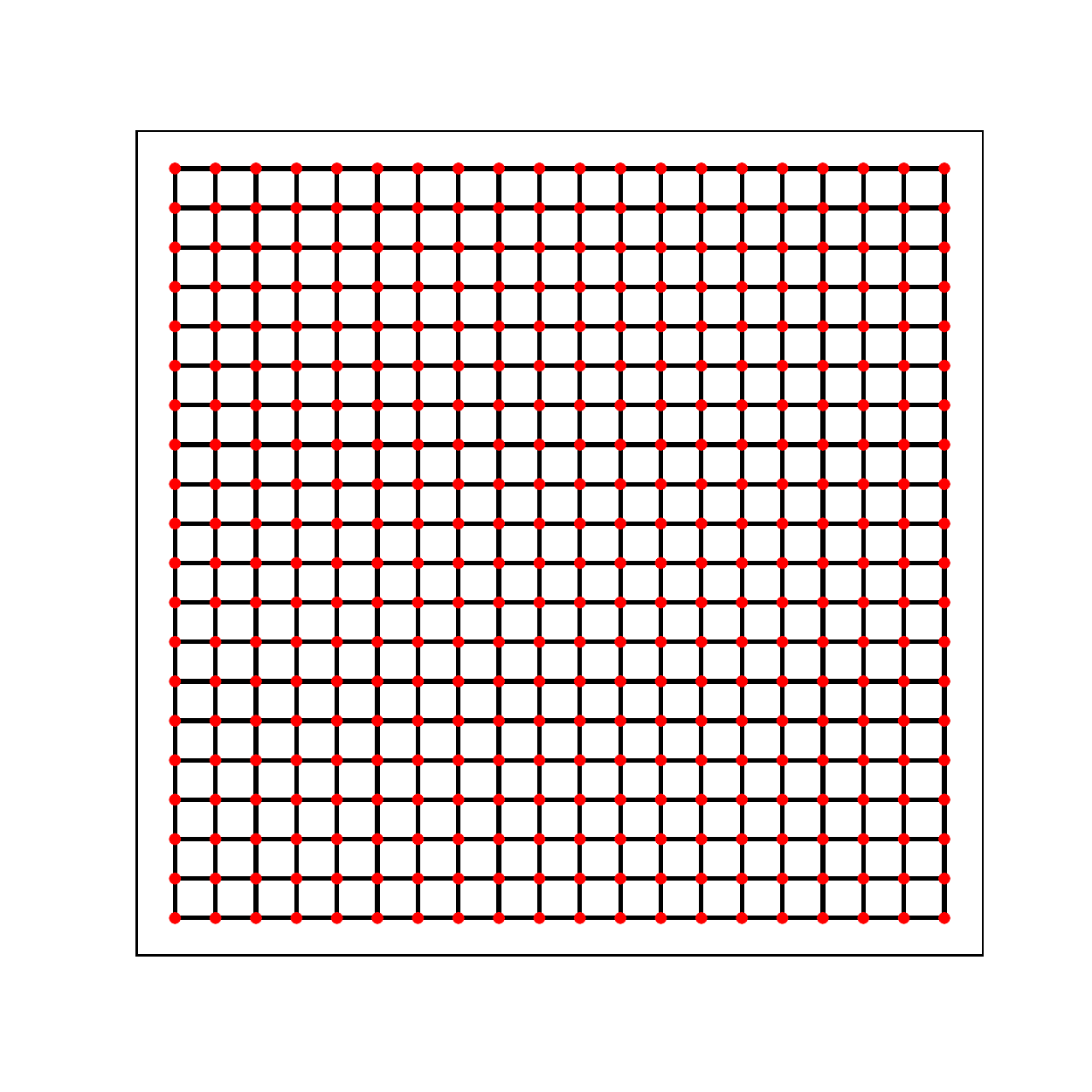}
		\vspace{-1cm}
		\caption{Grid}
		\label{fig:grid_drawing}
	\end{subfigure}
	\begin{subfigure}{0.32\linewidth}
		\includegraphics[width=\linewidth, trim = 0 0 0 0]{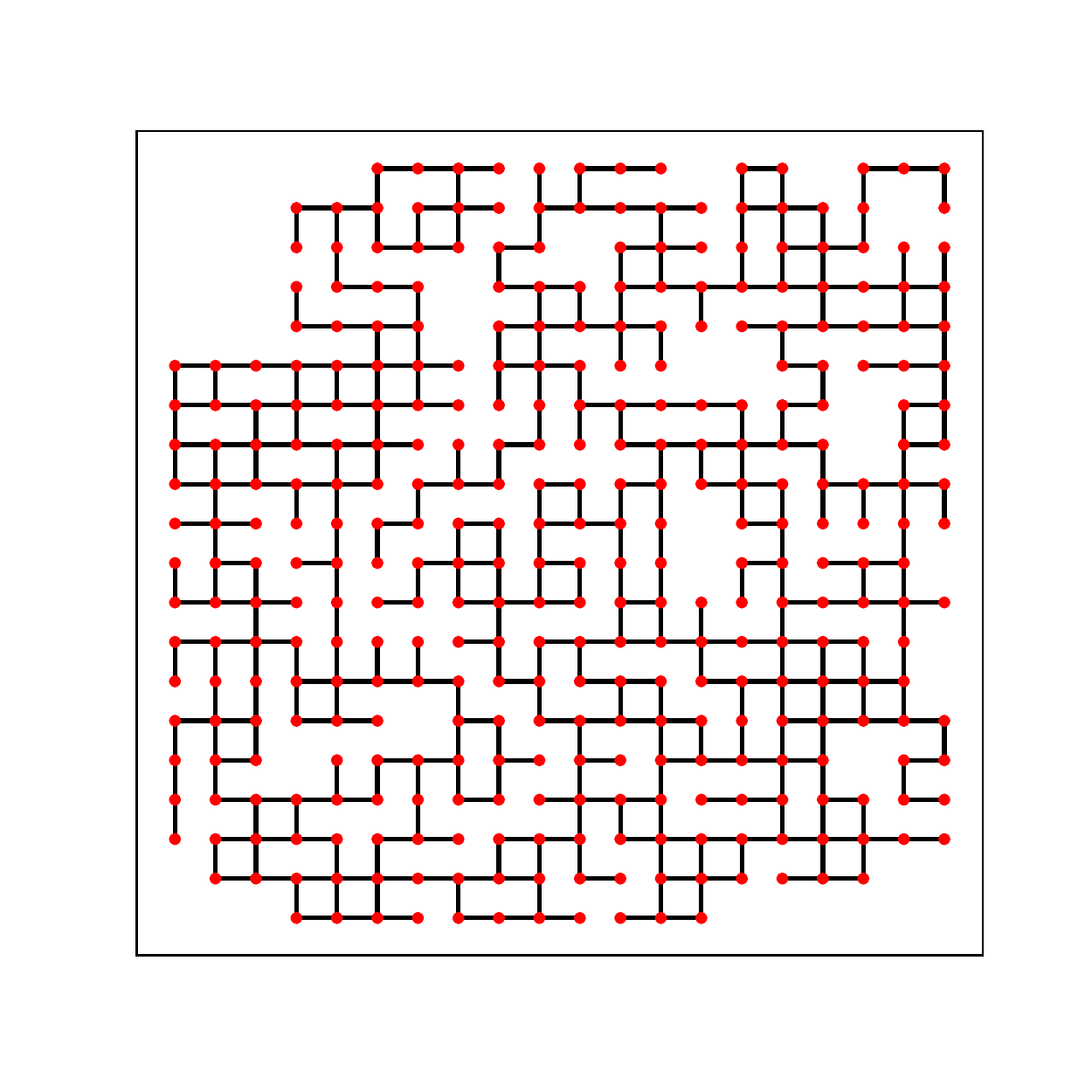}
		\vspace{-1cm}
		\caption{Percolation bond}
		\label{fig:percolation_drawing}
	\end{subfigure}
	\begin{subfigure}{0.32\linewidth}
		\includegraphics[width=\linewidth, trim = 0 0 0 0]{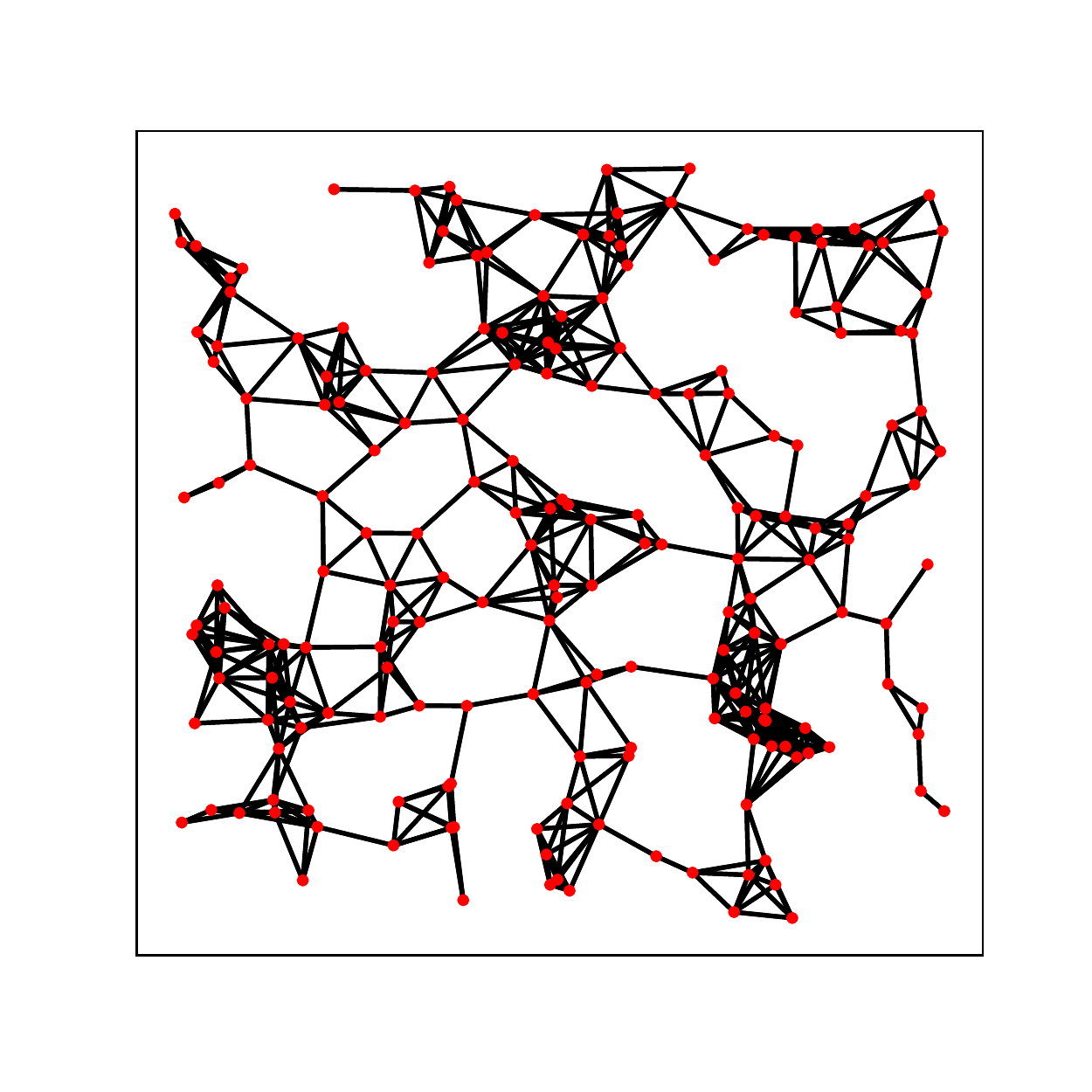}
		\vspace{-1cm}
		\caption{Random geometric graph}
		\label{fig:rgg_drawing}
	\end{subfigure}
	\caption{The three types of two-dimensional graphs considered in simulations.}
	\label{fig:graph_drawings} 
	\vspace{0.7cm}
	\begin{subfigure}{0.40\linewidth}
		\includegraphics[width=\linewidth, trim = 0 0 0 0]{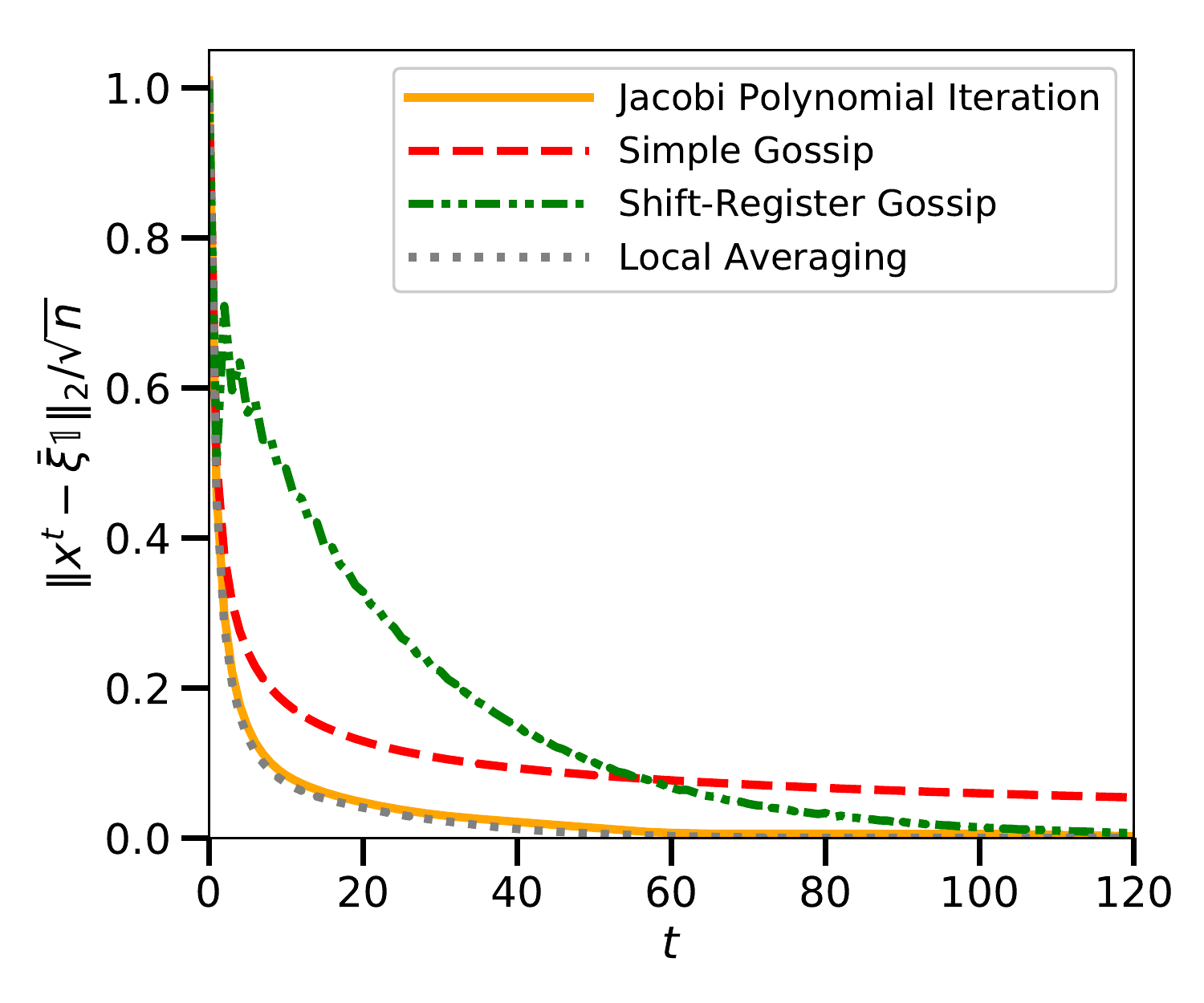}
		\vspace{-0.7cm}
		\caption{2D grid}
		\label{fig:sim_2D_grid}
		\vspace{0.4cm}
	\end{subfigure}
	\begin{subfigure}{0.40\linewidth}
		\includegraphics[width=\linewidth, trim = 0 0 0 0]{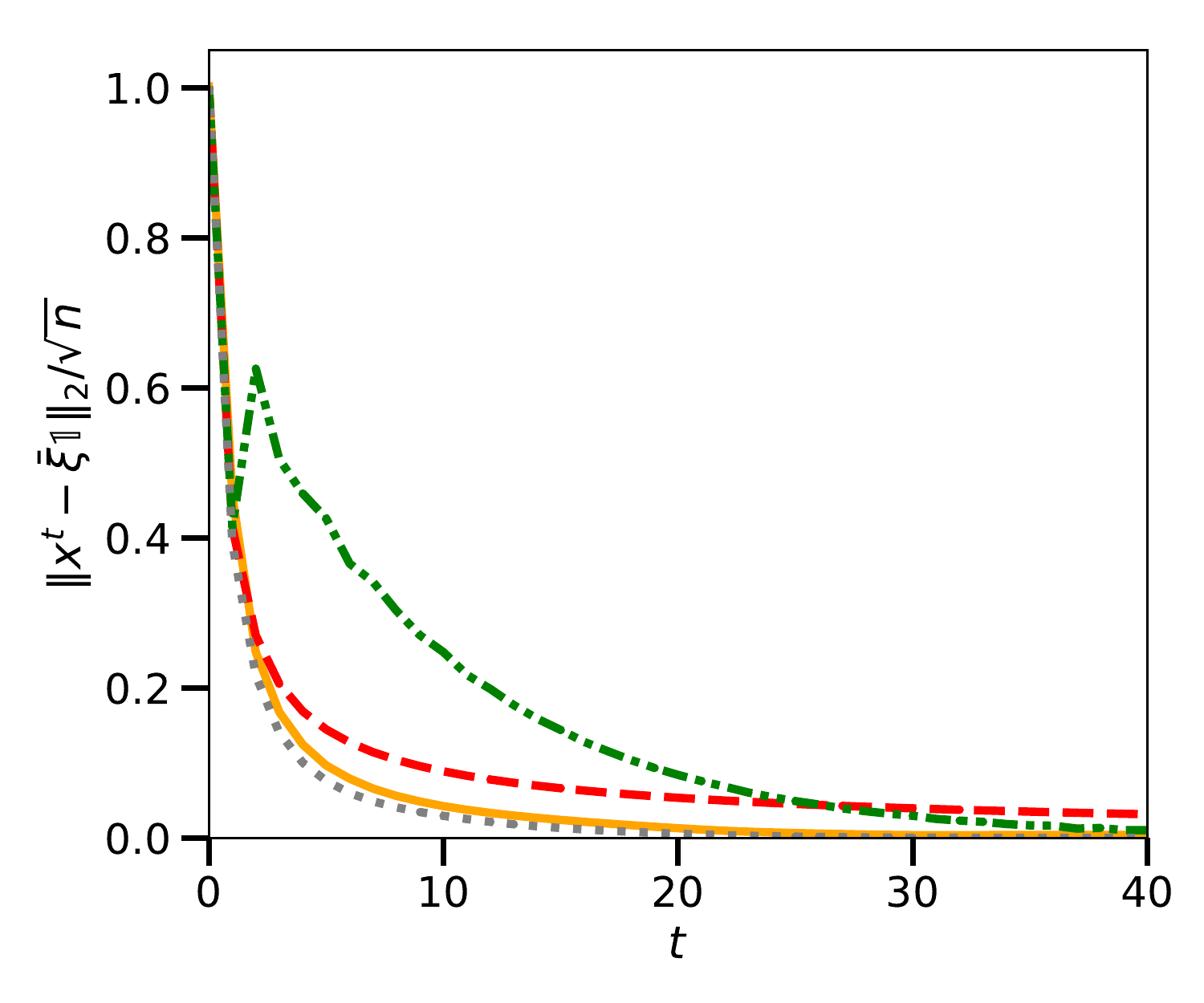}
		\vspace{-0.7cm}
		\caption{3D grid}
		\label{fig:sim_3D_grid}
		\vspace{0.4cm}
	\end{subfigure}
	\begin{subfigure}{0.40\linewidth}
		\includegraphics[width=\linewidth, trim = 0 0 0 0]{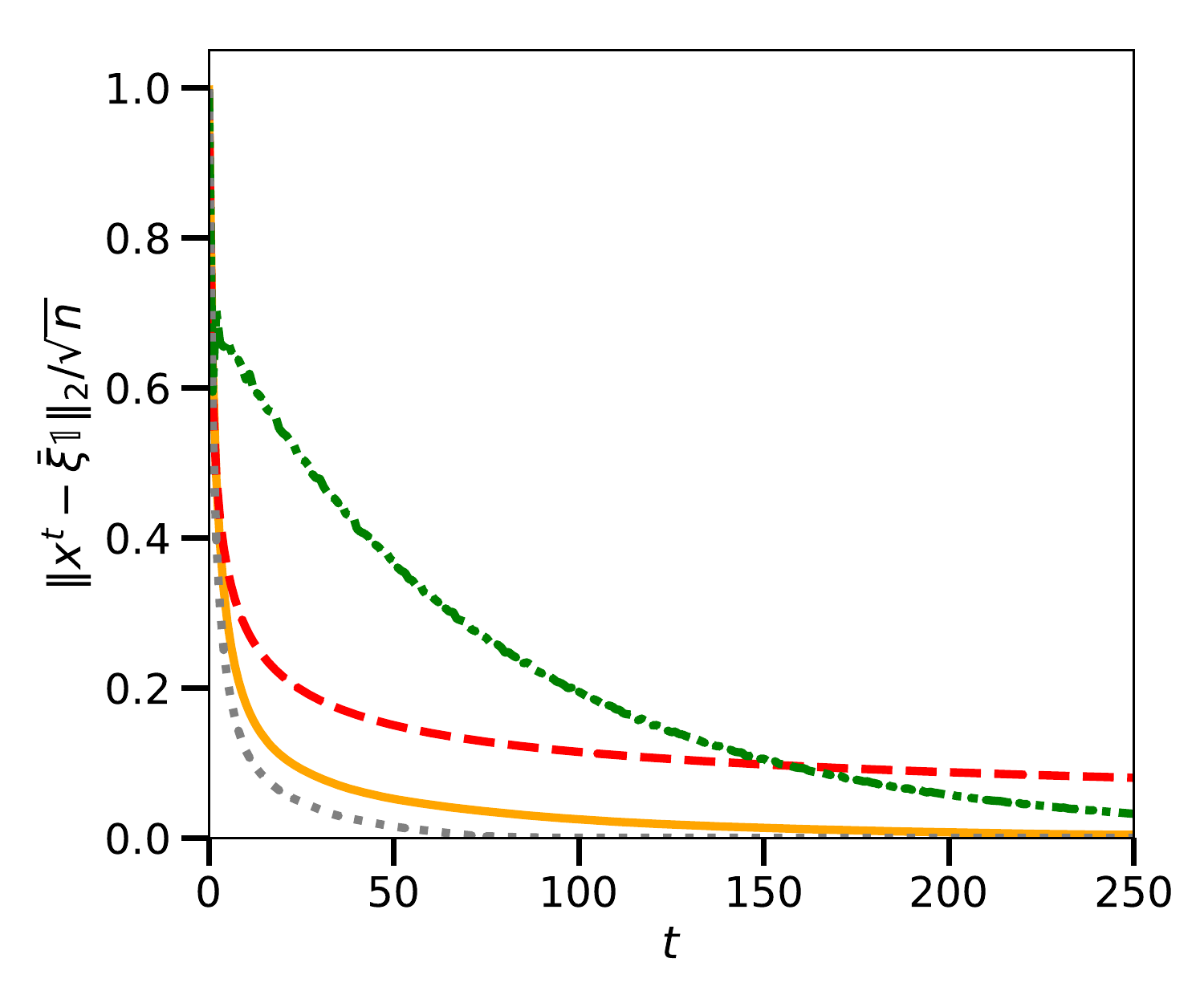}
		\vspace{-0.7cm}
		\caption{2D percolation bond}
		\label{fig:sim_2D_percolation}
		\vspace{0.4cm}
	\end{subfigure}
	\begin{subfigure}{0.40\linewidth}
		\includegraphics[width=\linewidth, trim = 0 0 0 0]{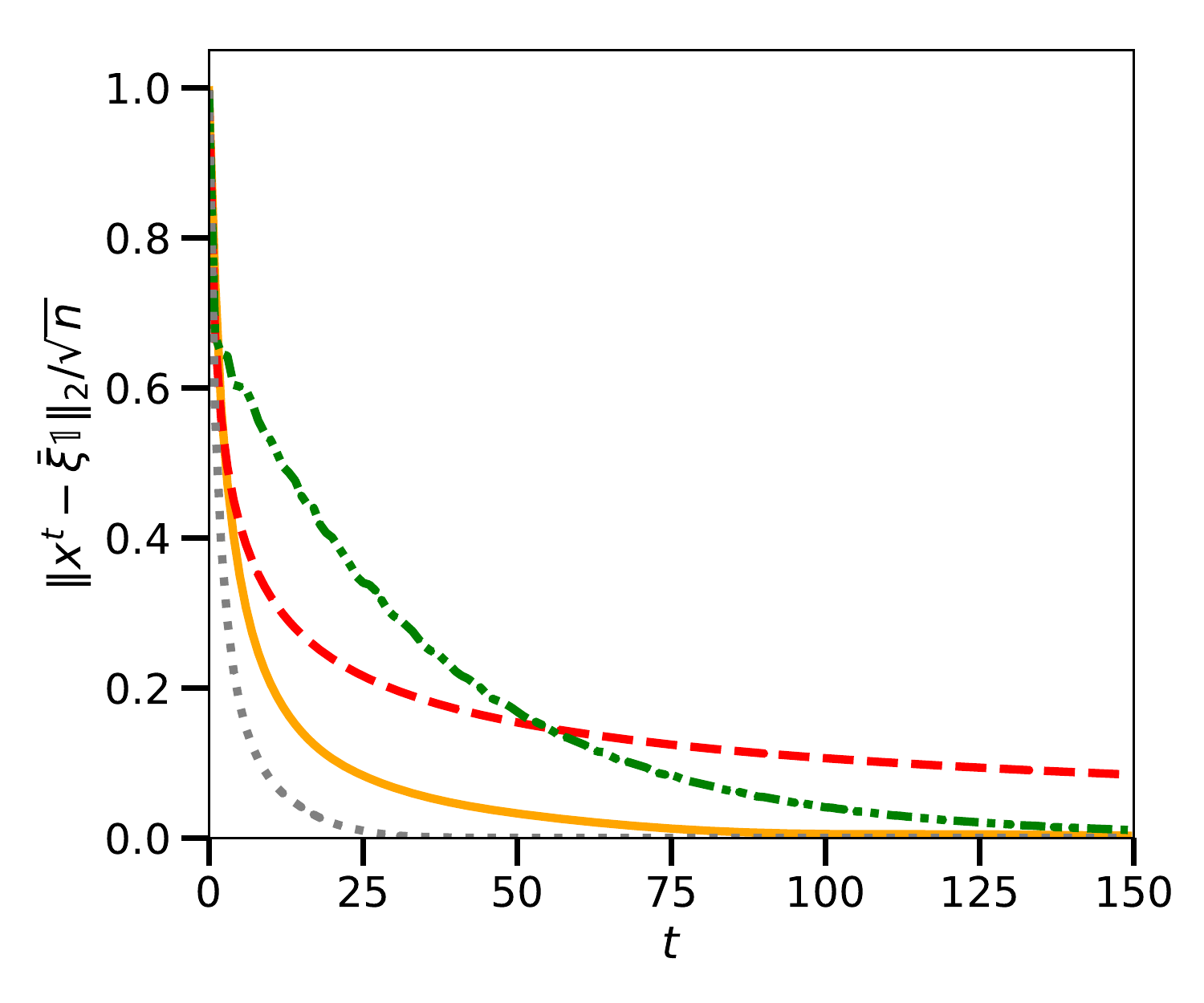}
		\vspace{-0.7cm}
		\caption{3D percolation bond}
		\label{fig:sim_3D_percolation}
		\vspace{0.4cm}
	\end{subfigure}
	\begin{subfigure}{0.40\linewidth}
		\includegraphics[width=\linewidth, trim = 0 0 0 0]{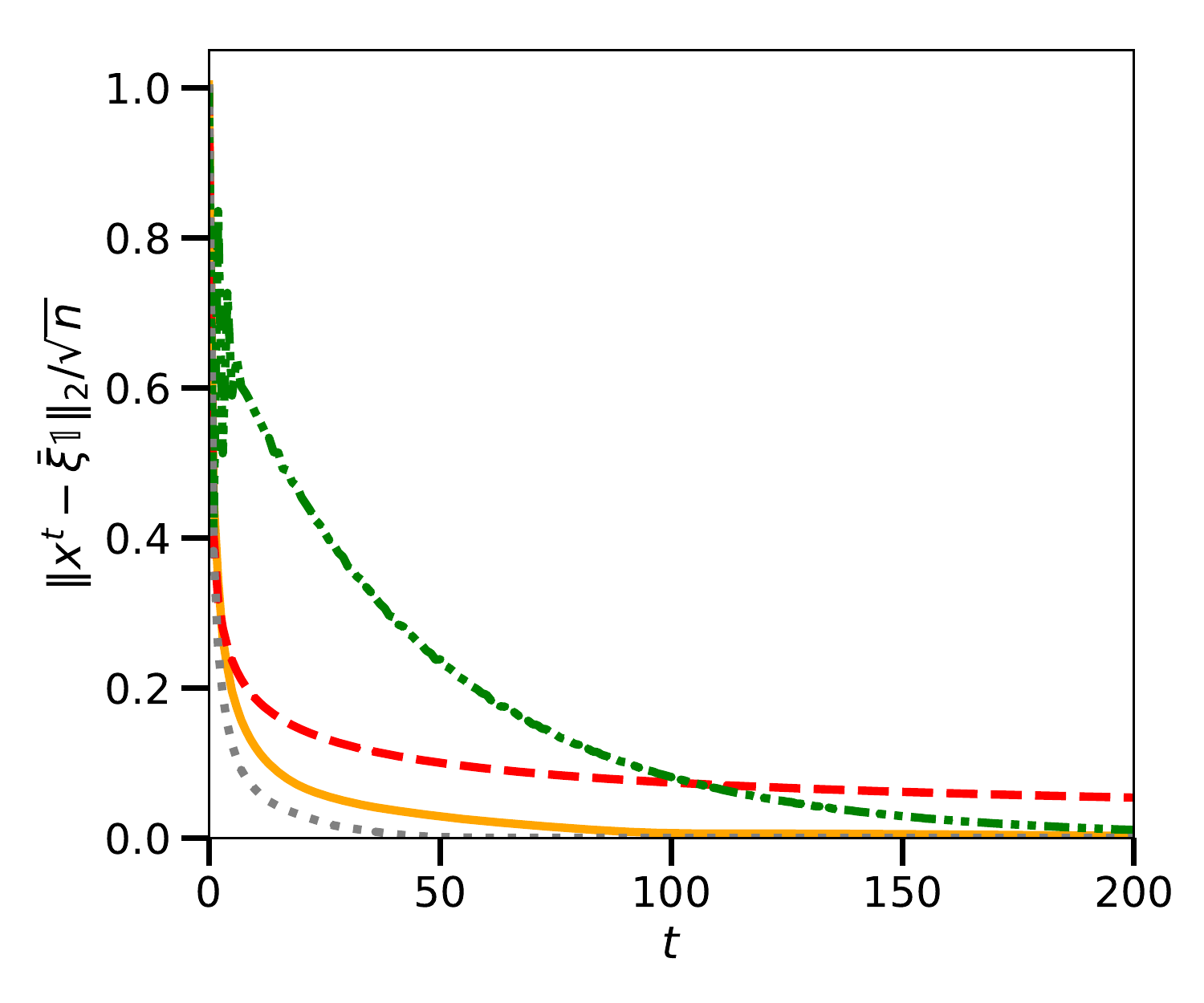}
		\vspace{-0.7cm}
		\caption{2D random geometric graph}
		\label{fig:sim_2D_rgg}
		\vspace{0.4cm}
	\end{subfigure}
	\begin{subfigure}{0.40\linewidth}
		\includegraphics[width=\linewidth, trim = 0 0 0 0]{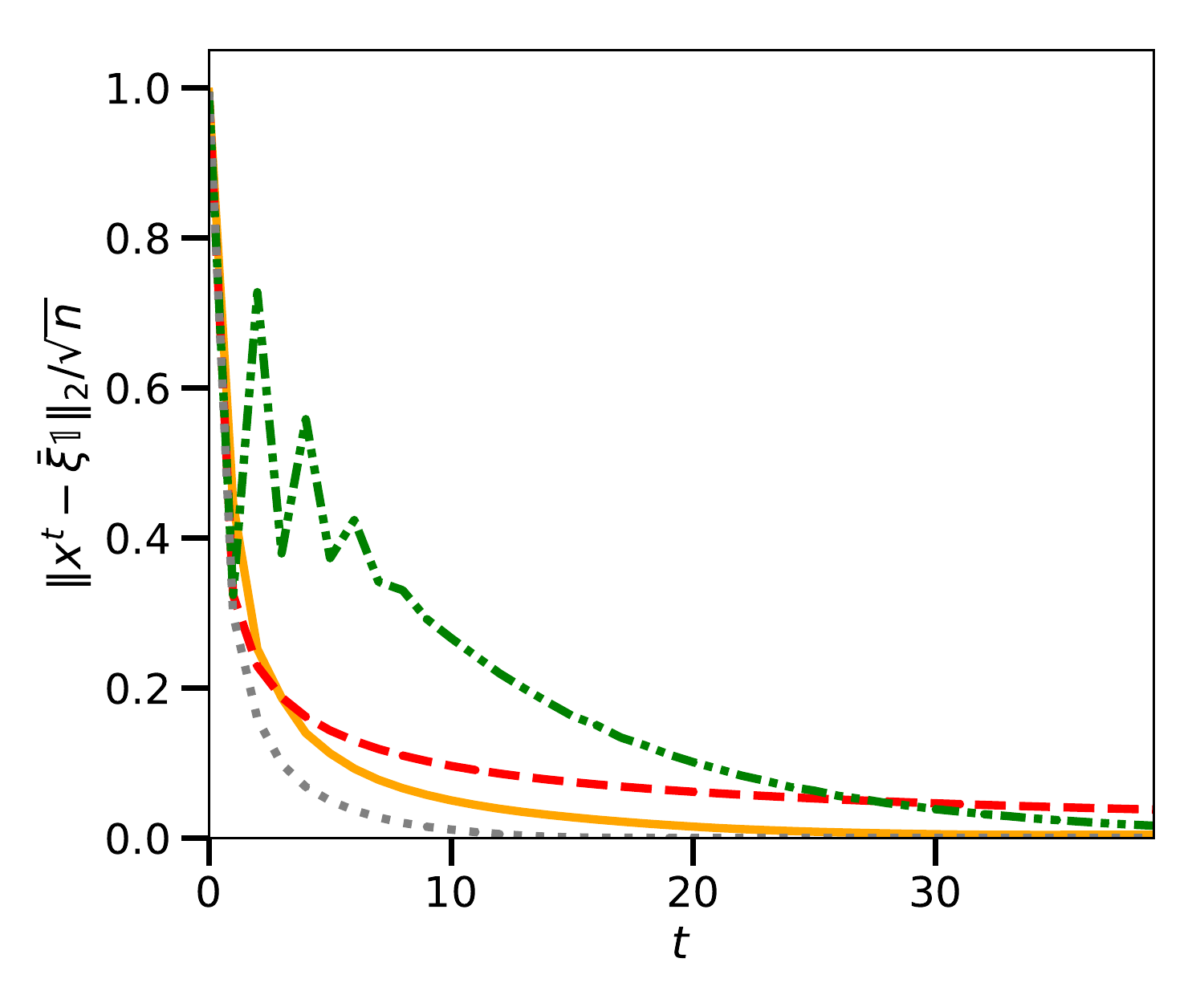}
		\vspace{-0.7cm}
		\caption{3D random geometric graph}
		\label{fig:sim_3D_rgg}
		\vspace{0.4cm}
	\end{subfigure}
	\caption{Performance of different gossip algorithms running on graphs with an underlying low-dimensional geometry, as measured by $\Vert x^t - \bar{\xi} \bfone \Vert_2 / \sqrt{n}$.}
	\label{fig:geometric_graph_simulation} 
\end{figure}

\begin{figure}
	\begin{subfigure}{0.49\linewidth}
		\includegraphics[width=\linewidth, trim = 0 0 0 0]{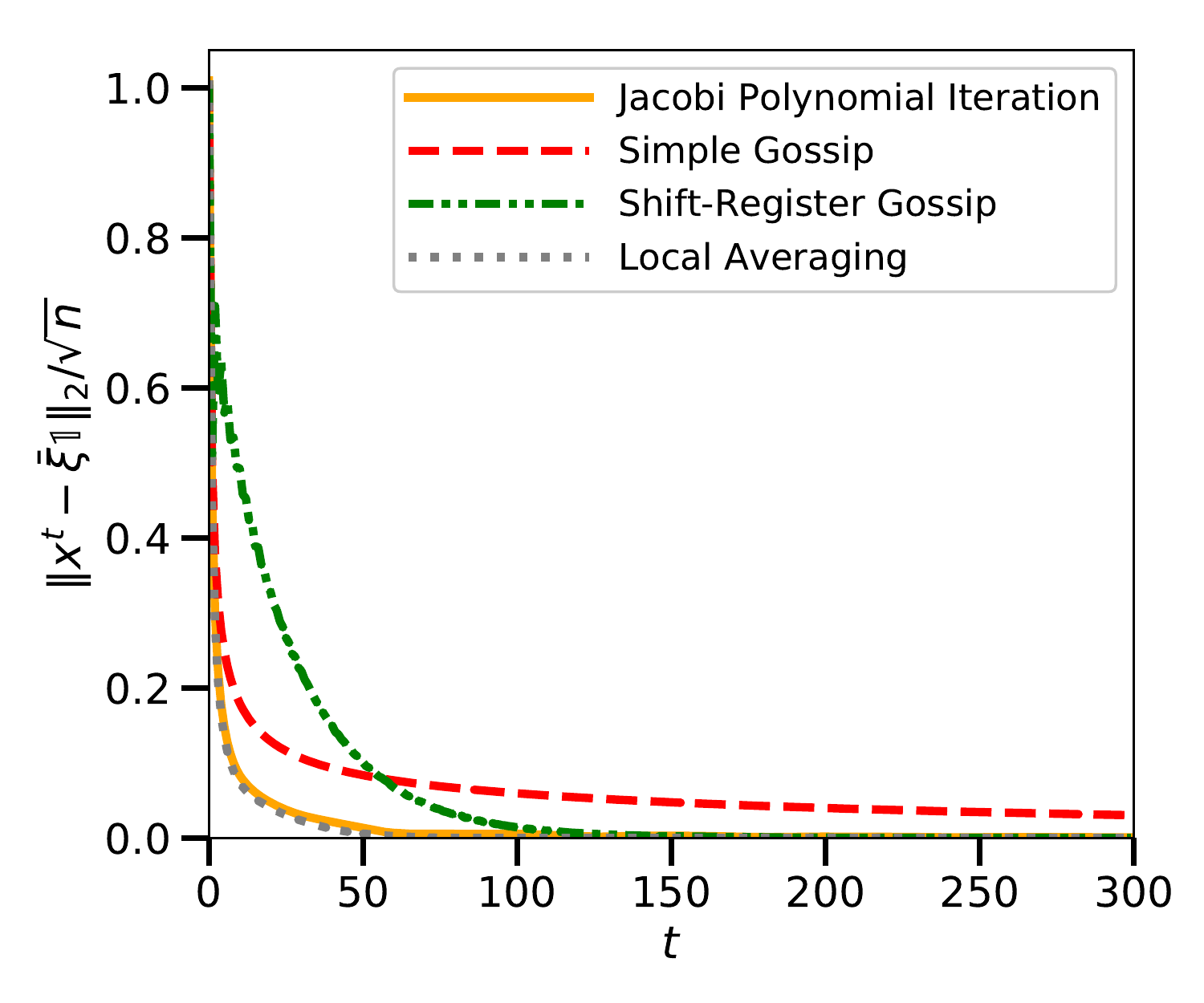}
		\caption{in linear scale}
	\end{subfigure}
	\begin{subfigure}{0.49\linewidth}
		\includegraphics[width=\linewidth, trim = 0 0 0 0]{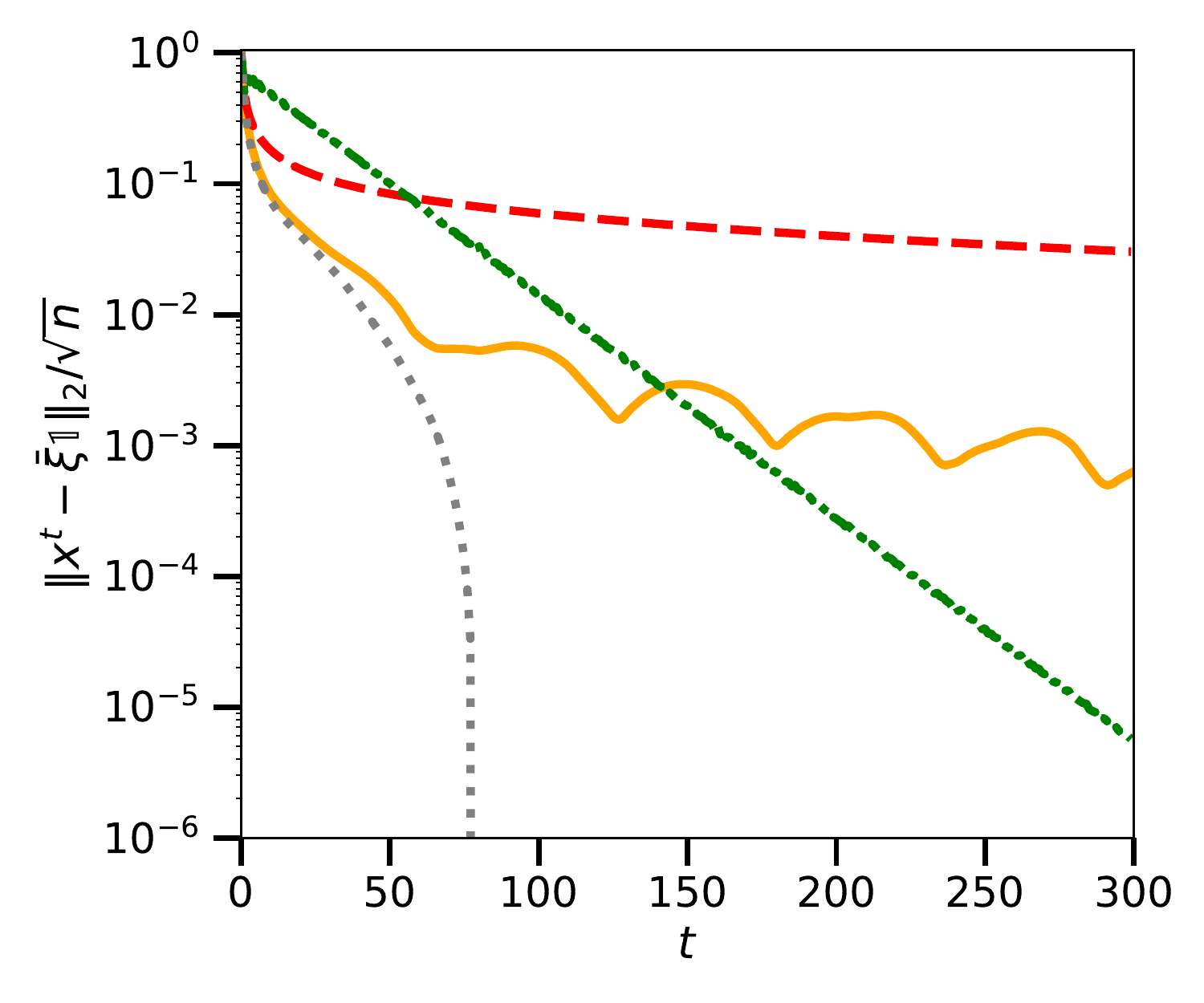}
		\caption{in log-scale}
	\end{subfigure}
	\caption{Performance of different gossip algorithms running on the 2D grid.}
	\label{fig:sim_2D_grid_log}
\end{figure}

\textbf{Interpretation of the results.} The results of the simulations are exposed in Figure \ref{fig:geometric_graph_simulation}. The qualitative picture remains the same across different graphs. Simple gossip performs better than shift-register gossip in a first phase, but in a large $t$ asymptotic, simple gossip converges slowly where shift-register gossip converges quickly. The Jacobi polynomial iteration enjoys both the quick diffusion of simple gossip in the first phase, and reaches the full mixing before shift-register gossip. As a consequence, the Jacobi polynomial iteration gets considerably closer to the local averaging optimal bound, especially in very regular structures like grids. 

These results should be mitigated with the large $t$ asymptotic: in Figure \ref{fig:sim_2D_grid_log}, we show the comparison of gossip methods on a longer time scale, in linear and log-scale y-axis. We only present the results on the 2D grid as they are typical of the behavior on other structures. We observe that shift-register gossip enjoys a much better asymptotic rate of convergence than simple gossip and the Jacobi polynomial iteration. 

Methods that use the spectral gap are designed to achieve the best possible asymptotic (see \cite{cao2006accelerated}, \cite{rebeschini2017accelerated}), thus the above observation is not surprising. These methods however fail in the non-asymptotic regime, where they are outperformed by the Jacobi polynomial iteration and simple gossip. We believe that in applications where a high precision on the average is not needed, the Jacobi polynomial iteration brings important improvements over existing methods, let alone the fact that it is considerably easier to tune. However, in Section \ref{sec:jacobi-with-spectral-gap}, we present a Jacobi polynomial iteration that uses the spectral gap of the gossip matrix to obtain the accelerated convergence rate. 

\section{Design of best polynomial gossip iterations}
\label{sec:second_order_gossip}

We now turn to the design of efficient polynomial iterations of the form $x^t = P_t(W)\xi$. An important result of this section is that the best iterates of this form can be computed in an online fashion as they result from a second-order recurrence relation.

The approach presented in this section is similar to \cite[Section 3.3]{diekmann1999efficient}, although therein it is applied to the slightly different problem of load balancing. We repeat here the derivations as we take a slightly different approach: here we derive the best polynomial $P_t$ with fixed $W$ and $\xi$; while in \cite{diekmann1999efficient} the matrix $W$ is fixed, but a polynomial $P_t$ efficient uniformly over $\xi$ is sought. We then discuss why the resulting recursion may be impractical. The next section introduces some approximation of the impractical scheme that leads to the practical iteration \eqref{eq:iteration}.

\medskip
Our measure of performance of a polynomial gossip iteration is the sum of squared errors over the agents of the network:
\begin{equation*}
\cE(P_t) = \sum_{v \in V} (x^t_v - \bar{\xi})^2 = \Vert x^t - \bar{\xi}\bfone \Vert_2^2 = \Vert P_t(W)\xi - \bar{\xi}\bfone \Vert_2^2 \, .
\end{equation*}
Denote $\lambda_1, \lambda_2, \dots, \lambda_n$ the real eigenvalues of the symmetric matrix $W$ and $u^1, u^2,\dots,u^n$ are the associated eigenvectors, normalized such that $\Vert u^i \Vert_2=1$. The diagonalization of $W$ gives the new expression of the error 
\begin{equation}
\label{eq:error-as-spectral-finite}
\cE(P_t) =  \sum_{i=2}^{n} \langle \xi, u^i \rangle^2 P_t(\lambda_i)^2 = \int_{-1}^{1} P_t(\lambda)^2 \diff\sigma(\lambda) \, , \qquad \diff\sigma(\lambda) = \sum_{i=2}^{n} \langle \xi, u^i \rangle^2 \delta_{\lambda_i} \, , 
\end{equation}
where $\langle ., . \rangle$ denotes the canonical scalar product on $\R^n$ and $\delta_\lambda$ is the Dirac mass at $\lambda$.

The polynomial $\pi_t$ minimizing the error $\cE(P_t)$ must be chosen as
\begin{equation}
\label{eq:best_polynomial}
\pi_t \in \underset{P(1)=1, \, \deg P \leq t}{\rm argmin} \int_{-1}^{1} P(\lambda)^2 \diff\sigma(\lambda) \, .
\end{equation} 

We now show that the sequence of best polynomials $\pi_0, \pi_1, \pi_2, \dots$ can be computed as the result of a second-order recursion, which leads to a second-order gossip method, whose coefficients depend on $\sigma$. As noted in \cite{cao2006accelerated}, having iterates $x^t$ that satisfy a low-order recurrence relation is valuable as it ensures that they can be computed online with limited memory cost. In order to prove this property for our iterates, we use that these polynomials are orthogonal with respect to some measure $\tau$.

\begin{definition}[Orthogonal polynomials w.r.t.~$\tau$]
	\label{def:orthogonal_polynomials}
	Let $\tau$ be a measure on $\R$ whose moments are all finite. Endow the set of polynomials $\R[X]$ with the scalar product 
	\begin{equation*}
	\left\langle P,Q\right\rangle_\tau = \int_{\R} P(\lambda)Q(\lambda) \diff\tau(\lambda) \, .
	\end{equation*}
	Denote $T \in \N\cup\{\infty\}$ the cardinal of the support of $\tau$. Then there exists a family $\pi_0, \pi_1, ..., \pi_{T-1}$ of polynomials, such that for all $t < T$, $\pi_0, \pi_1, ..., \pi_{t}$ form an orthogonal basis of $(\R_t[X],\left\langle .,.\right\rangle_\tau)$, where $\R_t[X]$ denotes the set of polynomials of degree smaller or equal to $t$. In other words, for all $s,t < T$, 
	\begin{align*}
	& \deg \pi_t = t \, , && \left\langle \pi_s,\pi_t\right\rangle_\tau = 0 \qquad \textrm{if }s \neq t \, .
	\end{align*}
	$\pi_0, \pi_1, ..., \pi_{T-1}$ is called a sequence of \emph{orthogonal polynomials with respect to $\tau$ (w.r.t.~$\tau$)}. Moreover, the family of orthogonal polynomials $\pi_0, \pi_1, ..., \pi_{T-1}$ is unique up to a rescaling of each of the polynomials. 
\end{definition}

An extensive reference on orthogonal polynomials is the book \cite{szeg1939orthogonal}. An introduction from the point of view of applied mathematics can be found in \cite{gautschi2004orthogonal}. In Appendix \ref{ap:toolbox-orthogonal-polynomials}, we recall the results from the theory of orthogonal polynomials that we use in this paper. The next proposition states that the optimal polynomials sought in \eqref{eq:best_polynomial} are orthogonal polynomials.

\begin{proposition}
	\label{prop:opt_polynomial}
	Let $\sigma$ be some finite measure on $[-1,1]$ and let $T \in \N\cup\{\infty\}$ be the cardinal of ${\rm Supp \,}\sigma - \{1\}$. For $0 \leq t \leq T-1$, the minimizer $\pi_t$ of 
	\begin{equation*}
 \underset{P(1)=1, \, \deg P \leq t}{\min} \int_{-1}^{1} P(\lambda)^2 \diff\sigma(\lambda) 
	\end{equation*}
	is unique. Moreover, $\pi_0, \dots, \pi_{T-1}$ is the unique sequence of orthogonal polynomials w.r.t.~$\diff\tau(\lambda) = (1-\lambda)\diff\sigma(\lambda)$ normalized such that $\pi_t(1)=1$.
\end{proposition}

This result is well-known and usually stated without proof \cite[Sections 3, 4.1]{nevai1986geza}, \cite[Section 2]{nevai1979orthogonal}; we give the short proof in Appendix \ref{ap:opt_polynomial_proof}. In the following, the phrase ``\emph{the} orthogonal polynomials w.r.t.~$\tau$'' will refer to the unique family of orthogonal polynomials w.r.t.~$\tau$ \emph{and normalized such that $\pi_t(1)=1$.}

\begin{remark}
	\label{rem:perfect-gossip}
When $T$ is finite and $t \geq T$, finding a minimizer of $\int_{-1}^{1} P(\lambda)^2 \diff\sigma(\lambda)$ over the set of polynomials such that $P(1)=1$, $\deg P \leq t$ is trivial. Indeed, one can consider the polynomial 
		\begin{equation*}
		\pi_T(\lambda) = \frac{\prod_{\lambda' \in {\rm Supp \,}\sigma - \{1\}} (\lambda-\lambda')}{  \prod_{\lambda' \in {\rm Supp \,}\sigma - \{1\}} (1-\lambda')}
		\end{equation*}
		which is of degree $T$, satisfies $\pi_T(1)=1$ and $\int_{-1}^{1} \pi_T(\lambda)^2 \diff\sigma(\lambda) = \sigma(\{1\})$. This is the best value that a polynomial $P$ of any degree, such that $P(1)=1$, can get.
\end{remark}
A fundamental result on orthogonal polynomials states that they follow a second-order recursion.

\begin{proposition}[Three-term recurrence relation, {from \cite[Theorem 3.2.1]{szeg1939orthogonal}}]
	\label{prop:recurrence_relation}
	Let $\pi_0, \dots, \pi_{T-1}$ be a sequence of orthogonal polynomials w.r.t.~some measure $\tau$.
	There exist three sequences of coefficients $(a_t)_{1 \leq t \leq T-2}$, $(b_t)_{1 \leq t \leq T-2}$ and $(c_t)_{1 \leq t \leq T-2}$ such that for $1 \leq t \leq T-2$, 
	\begin{equation*}
	 \pi_{t+1}(\lambda) = (a_t\lambda+b_t) \pi_t(\lambda) - c_{t} \pi_{t-1}(\lambda) \, .
	\end{equation*} 
\end{proposition}
The classical proof of this proposition is given in Appendix \ref{ap:general-properties}. Taking $\sigma$ to be the spectral measure of \eqref{eq:error-as-spectral-finite} in Proposition \ref{prop:opt_polynomial}, we get that the best polynomial gossip algorithm is a second-order method whose coefficients are determined by the graph $G$, the gossip matrix $W$ and the vertex $v$. Indeed, as $\pi_0, \dots, \pi_{T-1}$ is a family of orthogonal polynomials, there exists coefficients $a_t, b_t,c_t$ such that 
\begin{equation*}
	\pi_{t+1}(\lambda) = (a_t\lambda+b_t) \pi_t(\lambda) - c_t \pi_{t-1}(\lambda) \, ,
\end{equation*} 
and thus
\begin{equation*}
\pi_{t+1}(W) = a_t W \pi_t(W) + b_t \pi_t(W) - c_t \pi_{t-1}(W) \, .
\end{equation*} 
Decomposing $\pi_1(\lambda)= a_0\lambda + b_0$ and applying the previous relation in $\xi$ gives the second-order recursion for the best polynomial estimators $x^t = \pi_t(W)\xi$:
\begin{align}
\label{eq:second-order-polynomial-gossip}
&x^0 = \xi \, , &&x^1 = a_0 W\xi + b_0 \xi \, , &&x^{t+1} = a_t W x^t + b_t x^t - c_t x^{t-1} \, .
\end{align} 
Note that the dependence of the gossip method in the graph $G$, the gossip matrix $W$ and the vertex~$v$ is entirely hidden in the coefficients $a_t, b_t, c_t$. Thus the choice of the coefficients is central. In \cite{diekmann1999efficient}, it is argued that the coefficients can be computed in a ``preprocessing step''. Indeed, the coefficients can be computed in a centralized or decentralized manner, at the cost of extra communication steps. The gossip method that consists in computing the optimal coefficients $a_t,b_t,c_t$ and running Eq.~\eqref{eq:second-order-polynomial-gossip} will be refered to as \emph{parameter-free polynomial iteration}, as it does not require any tuning of parameters, and by analogy with the terminology used in polynomial methods for the resolution of linear systems (see \cite[Section 6]{fischer1996polynomial}). It corresponds to the optimal polynomial iteration. For a detailed exposition on the parameter-free polynomial iteration and a discussion of its practicability, see Appendix \ref{ap:parameter_free_polynomial_iteration}.

However, in dynamic networks that are constantly changing, it is not a valid option to keep repeating the preprocessing step to update the coefficients $a_t, b_t, c_t$. Our approach consists in observing that there are sequences of coefficients like \eqref{eq:coeff} that --albeit they are not optimal-- work reasonably well on a large set of graphs. This implies that even if the details of the graph are not known to the algorithmic designer, she can make a choice of coefficients that have a fair performance.

More formally, we approximate the true spectral measure $\sigma$ of the graph with a simpler measure~$\tilde{\sigma}$, whose associated polynomials have known recursion coefficients $a_t, b_t, c_t$. We will show that in some cases, substituting the orthogonal polynomials w.r.t.~$\sigma$ with the ones orthogonal to~$\tilde{\sigma}$ does not worsen the efficiency of the gossip method much. In the next sections, we argue for two choices of the approximating measure $\tilde{\sigma}$. The first uses only the spectral dimension $d$ of the network, and gives the Jacobi polynomial iteration \eqref{eq:iteration}. The second one uses both the spectral dimension $d$ and the spectral gap $\gamma$ of $W$, and gives the Jacobi polynomial iteration with spectral gap. 

Figure \ref{fig:sim_2D_grid_log_all} reproduces Figure \ref{fig:sim_2D_grid_log} and adds the performance of the parameter-free polynomial iteration and the Jacobi polynomial iteration with spectral gap. It shows that in linear scale, the performance of the parameter-free polynomial iteration is indistinguishable from the performance of the Jacobi polynomial iterations with or without spectral gap, which are obtained through approximations of the spectral measure $\sigma$. However, the figure in log-scale shows that the asymptotic convergence of the methods depends on the coarseness of the approximation. The relevance of this asymptotic convergence to the practice depends on the application. 

\begin{figure}
	\begin{subfigure}{0.49\linewidth}
		\includegraphics[width=\linewidth, trim = 0 0 0 0]{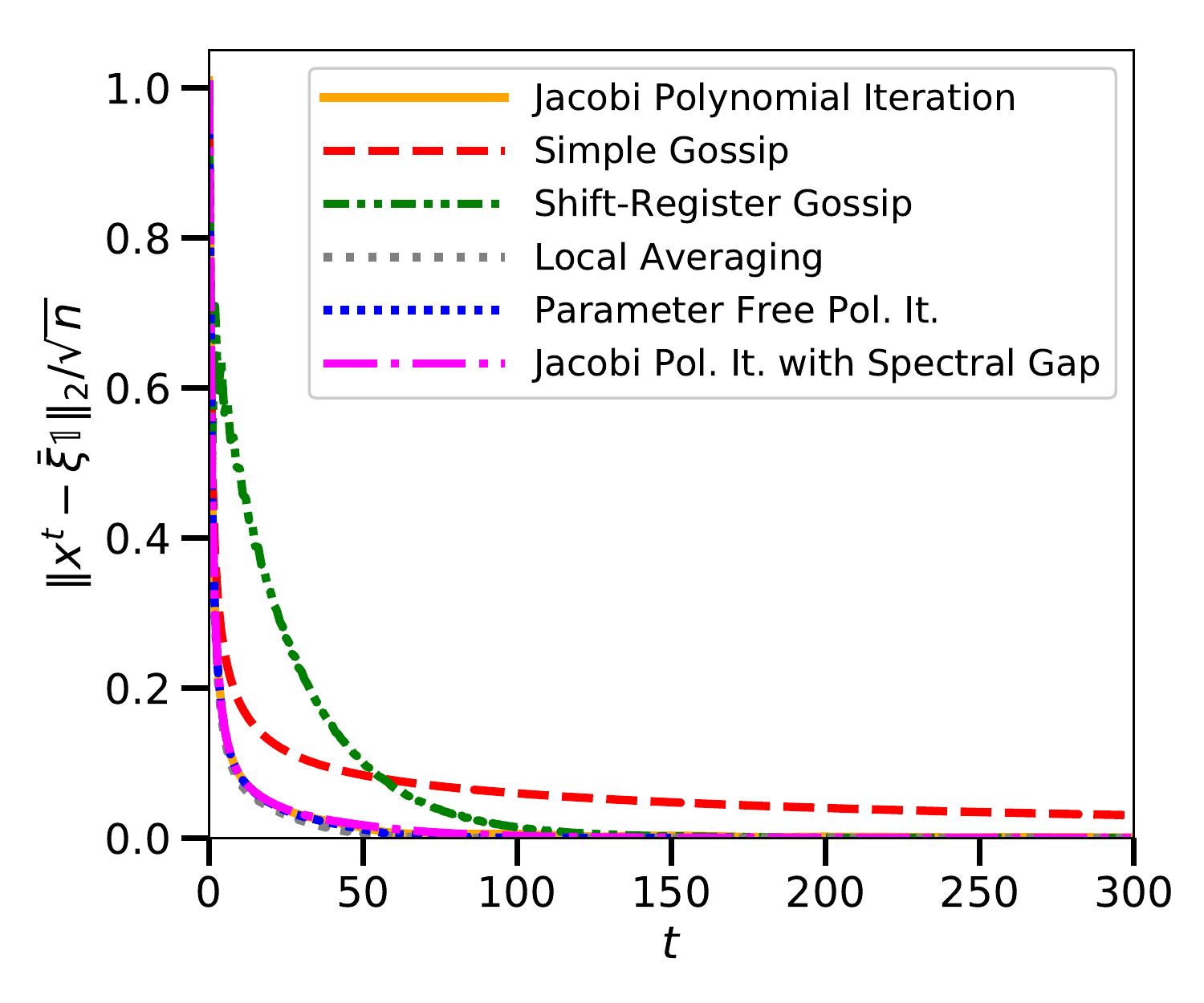}
		\caption{in linear scale}
	\end{subfigure}
	\begin{subfigure}{0.49\linewidth}
		\includegraphics[width=\linewidth, trim = 0 0 0 0]{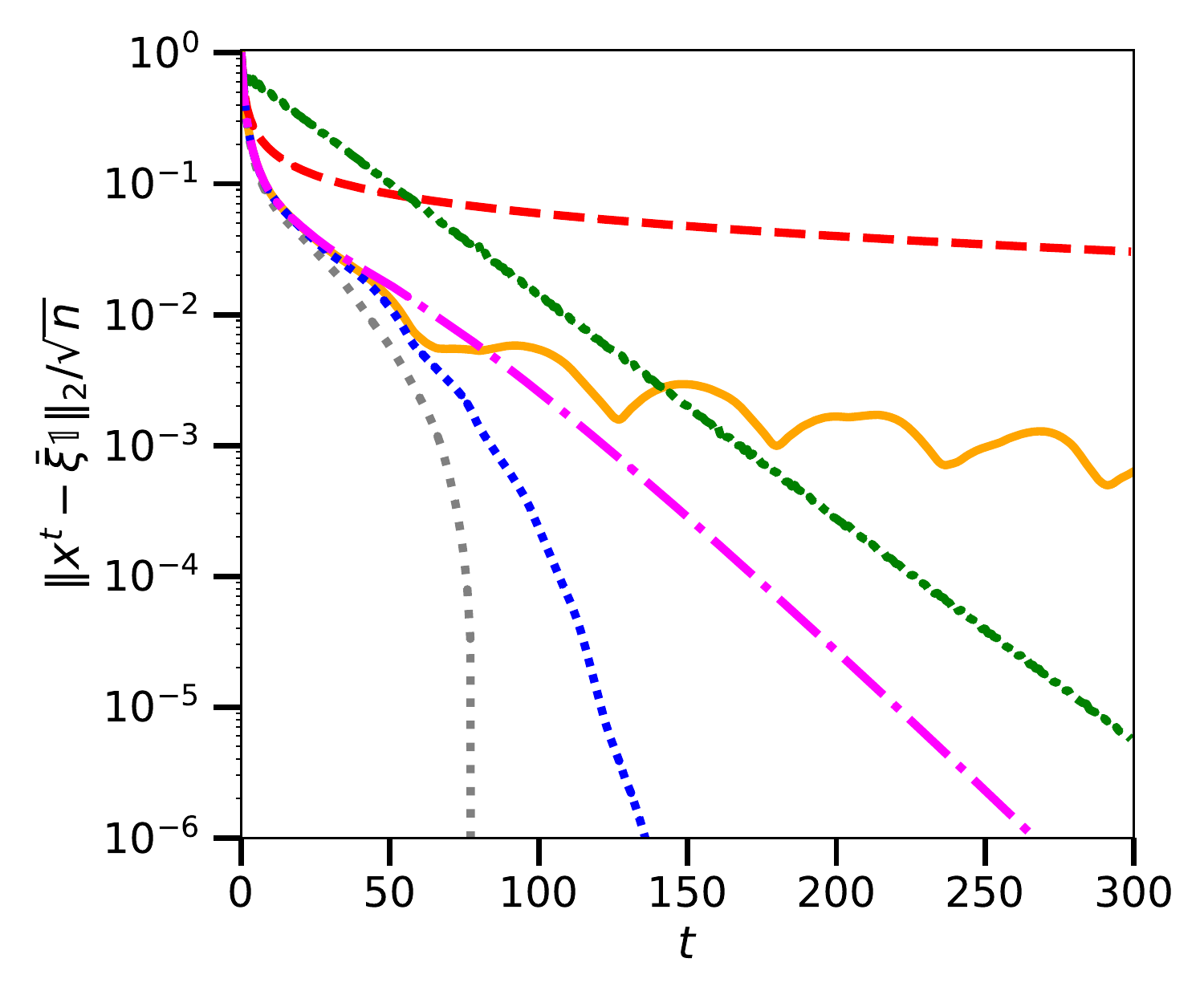}
		\caption{in log-scale}
	\end{subfigure}
	\caption{Performance of different gossip algorithms running on the 2D grid.}
	\label{fig:sim_2D_grid_log_all}
\end{figure}

\begin{remark}
The shift-register iteration $x^t = P_t(W)\xi$ defined in \eqref{eq:shift-register} can be seen as a best polynomial gossip iteration with some approximating measure. Indeed, the polynomials $P_t$, $t \geq 0$ are the orthogonal polynomials w.r.t.~some measure whose support is strictly included in $[-1,1]$ (see Proposition \ref{prop:shift-register-orthogonality-measure}). 
\end{remark}

\section{Design of polynomial gossip algorithms for graphs of given spectral dimension}
\label{sec:Jacobi-polynomial-gossip}

\subsection{The dimension $d$ and the rate of decrease of the spectral measure near $1$.}

We now assume that we are given a graph $G$ on which we would like to run the optimal polynomial gossip algorithm \eqref{eq:second-order-polynomial-gossip}. However, we do not know the spectral measure $\sigma$, nor the coefficients $a_t, b_t, c_t$. In this section, we give a heuristic motivating an approximation $\tilde{\sigma}$ of the spectral measure $\sigma$ using only the dimension $d$ of the graph. The heuristic is supported by the simulations of Section \ref{sec:simulations} and some rigorous theoretical support in Section \ref{sec:performance-guarantees}. 

Our approximation is given by the following non-rigorous intuition:
\begin{equation}
\label{eq:intuition}
\textrm{the graph $G$ is of dimension $d$} \qquad \Leftrightarrow \qquad \sigma([1-E,1]) \approx C E^{d/2} \qquad \textrm{as } E \ll 1 \, ,
\end{equation}
for some constant $C$. Of course, we have not defined the dimension of a graph, nor given a rigorous signification of the symbols ``$\approx$'' and ``$\ll$''. We come back to these questions in Section \ref{sec:spectral-dimension-infinite-graph}, but for now we assume that the reader has an intuitive understanding of these notions and finish drawing the heuristic picture. 

Eq.~\eqref{eq:intuition} describes the repartition of the mass of $\sigma$ near $1$. This mass near $1$ challenges the design of polynomial methods as the gossip polynomials $P$ are constrained to satisfy $P(1) = 1$ while minimizing $\int P^2\diff\sigma$. Moreover, eigenvalues of a graph close to $1$ are known to describe the large-scale structure of the graph and thus must be central in the design of gossip methods. The traditional design of gossip algorithms considered the spectral gap $\gamma$ between $1$ and the second largest eigenvalue, a quantity that typically gets very small in large graphs. Intuition \eqref{eq:intuition} also describes the behavior of the spectrum near $1$, but on a larger scale than the spectral gap. It describes how the set of the largest eigenvalues is distributed around $1$. 

\subsection{The Jacobi iteration for graphs of given dimension}
\label{sec:Jacobi-polynomial-iteration-subsection}

When a spectral measure satisfies the edge estimate \eqref{eq:intuition}, we approximate it with a measure satisfying the same estimate, namely 
\begin{equation*}
\diff \tilde{\sigma}(\lambda) = (1-\lambda)^{d/2-1}\bfone_{\{\lambda \in (-1,1)\}}\diff\lambda \, .
\end{equation*}
Note that we do not elaborate on the normalization of the approximate measure $\diff\tilde{\sigma}$ as it is only used to define an orthogonality relation between polynomials, in which the normalization does not matter. The orthogonal polynomials w.r.t the modified spectral measure $(1-\lambda)\diff \tilde{\sigma}(\lambda) = (1-\lambda)^{d/2}\bfone_{\{\lambda \in (-1,1)\}}\diff \lambda$ and their recursion coefficients are known as they correspond to the well-studied Jacobi polynomials \cite[Chapter IV]{szeg1939orthogonal}:
\begin{equation}
\label{eq:coeffs-Jacobi-d}
\begin{aligned}
&a_0^{(d)} = \frac{d+4}{2(2+d)} \, ,  && b_0^{(d)} = \frac{d}{2(2+d)} \, , \\
&a_t^{(d)} = \frac{(2t+d/2+1)(2t+d/2+2)}{2(t+1+d/2)^2} \, , && b_t^{(d)} = \frac{d^2(2t+d/2+1)}{8(t+1+d/2)^2(2t+d/2)} \, , \\
&c_t^{(d)} = \frac{t^2(2t+d/2+2)}{(t+1+d/2)^2(2t+d/2)} \, .
\end{aligned}
\end{equation}
These coefficients are derived in Appendix \ref{ap:Jacobi}. This approximation of the spectral measure gives the practical recursion
\begin{align}
\label{eq:Jabobi-polynomial-iteration}
&x^0 = \xi \, , &&x^1 = a_0^{(d)} W\xi + b_0^{(d)} \xi \, , &&x^{t+1} = a_t^{(d)} W x^t + b_t^{(d)} x^t - c_t^{(d)} x^{t-1} \, ,
\end{align} 
that only depends on $d$. It is just a rewriting of the Jacobi polynomial iteration \eqref{eq:iteration} given in the introduction of this paper. The Jacobi polynomial $\pi_t^{(d/2,0)}(\lambda)$ such that $x^t = \pi_t^{(d/2,0)}(W)\xi$ is plotted in Figure \ref{fig:comparison_polynomials_simple_jacobi} with $d=2$ and $t = 6$, along with the polynomial $\lambda^6$ associated with simple gossip. The Jacobi polynomial is smaller in magnitude near the edge of the spectrum.

\begin{figure}
	\includegraphics[width=0.5\linewidth]{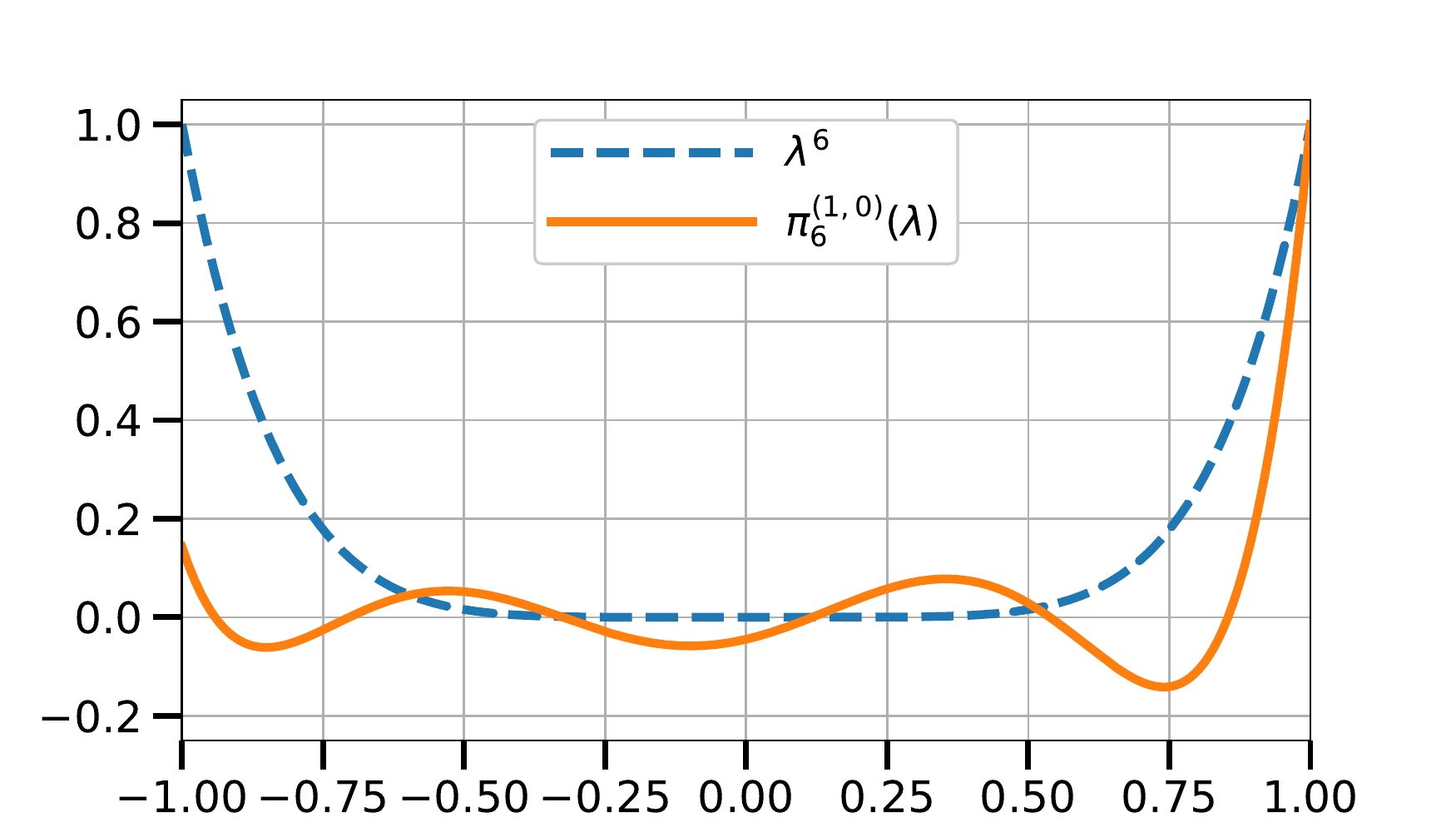}
	\caption{Comparison of the Jacobi polynomial $\pi_6^{(1,0)}(\lambda)$ with the polynomial of simple gossip $\lambda^6$.}
	\label{fig:comparison_polynomials_simple_jacobi}
\end{figure}

\subsection{Spectral dimension of a graph.}
\label{sec:spectral-dimension-infinite-graph}
In this section, we discuss the meaning of intuition \eqref{eq:intuition}. There are several definitions of the dimension of a graph. 

When referring to the dimension of a graph, many authors actually refer to some quantity $d$ that has been used in the construction of the graph. An example is the $d$-dimension grid $\{1,\dots,n\}^d$. Another example consists in removing edges in $\Z^d$ with probability $1-p$, independently of one another. The resulting graph $G$ is called a percolation bond \cite{grimmett1999percolation}. It is natural to consider that this graph is of dimension $d$. A more complicated example is the random geometric graph: choose $d \geq 1$, and sample $n$ points uniformly in the $d$-dimensional cube $[0,1]^d$, and connect with an edge all pairs of points closer than some chosen distance $r > 0$. It is natural to say that this random geometric graph is $d$-dimensional as it is the dimension of the surface it is built on. 

Mathematicians have developed more intrinsic definitions of the dimension of a graph \cite{durhuus2009hausdorff}; here we use the notion of \emph{spectral dimension}. This definition is of interest only for infinite graphs $G = (V,E)$. Here, we consider only locally finite graphs, meaning that each node has only a finite number of neighbors. As with Definition \ref{def:gossip-matrix}, one can define a gossip matrix $W$ with entries indexed by $V \times V$. If $G$ is infinite, $W$ is a doubly infinite array, but with only a finite number of non-zero elements in each line and column as the graph is locally finite. 

The spectral dimension of a graph $G$ is defined using a random walk on the graph, typically the simple random walk on $G$, but here we consider more generally the lazy random walk with transition matrix $\tilde{W} = (I+W)/2$. (We take the \emph{lazy} random walk to avoid periodicity issues.) 

\begin{definition}[Spectral dimension]
	\label{def:spectral-dimension}
	Denote $p_t$ the probability that the lazy random walk, when started from $v$, returns at $v$ at time $t$. The spectral dimension of the graph is, if it exists and is finite, the limit
\begin{equation*}
d_s = d_s(G,W,v) = -2 \, \lim_{t \to \infty} \frac{\ln p_t}{\ln t} \, .
\end{equation*}
\end{definition}

If the graph is connected and $W$ is the transition matrix of the simple random walk, this definition does not depend on the choice of the vertex $v$. Motivations for this definition are:

\begin{proposition}
	\label{prop:spectral-dim-Z^d}
	The spectral dimension of $\left(\Z^d, W\right)$ with $W = A(\Z^d)/d$ is $d$. 
\end{proposition}
\begin{proposition}[The spectral dimension of the supercritical percolation cluster is $d$]
	\label{prop:spectral-dimension-percolation}
	Let $G_0$ be a supercritical percolation bond in $\Z^d$ with edge probability $p \in (p_c,1]$, meaning that a.s., there is an infinite connected component $G$ in $G_0$. Endow $G$ with the gossip matrix $W = I + (A-D)/(2d)$, where $A$ and $D$ are respectively the adjacency and the degree matrices of $G$. Fix $v \in V$. Then a.s.~on the event $\{v \in G\}$, $d_s(G,W,v) = d$. 
\end{proposition}

The proofs of Propositions \ref{prop:spectral-dim-Z^d}, \ref{prop:spectral-dimension-percolation} are given in Appendix \ref{ap:spectral-dimensions}. The spectral dimension of a graph is related to the decay of the spectrum of $W$ near $1$.

\begin{definition}[Spectral measure of a possibly infinite graph]
	\label{def:spectral_measure_infinite_graph}
	Let $G$ be a graph and $W$ its gossip matrix. Fix $v \in V$. As $W$ is an auto-adjoint operator, bounded by $1$, acting on $\ell^2(V)$, there exists a unique positive measure $\sigma = \sigma(G,W,v)$ on $[-1,1]$, called the \emph{spectral measure}, such that for all polynomial $P$,
	\begin{equation*}
	\left\langle e_v, P(W)e_v \right\rangle_{\ell^2(V)} = \int_{-1}^{1} P(\lambda) \diff\sigma(\lambda) \, .
	\end{equation*} 
\end{definition}

For a deeper presentation of spectral graph theory, see \cite{mohar1989survey} and references therein. Note that when the graph $G$ is finite, it is easy to check that the spectral measure is the discrete measure $\sigma(G,W,v) = \sum_{i=1}^{n} (u_{v}^i)^2 \delta_{\lambda_i}$ where $\lambda_1, \dots, \lambda_n$ are the eigenvalues of $W$ and $u^1,\dots,u^n$ are the associated normalized eigenvectors. However, when the graph $G$ is infinite, the spectrum may exhibit a continuous part w.r.t.~the Lebesgue measure. 

\begin{proposition}[The spectral dimension is the spectral decay]
	\label{prop:spectral-dimension-spectral-decay}
Let $G$ be a graph, $W$ a gossip matrix on $G$ and $v$ a vertex. We denote $d_s = d_s(G,W,v)$ the spectral dimension and $\sigma = \sigma(G,W,v)$ the spectral measure. Then the limit $\lim_{E \to 0} \ln \sigma([1-E,1])/\ln E$ exists and is finite if and only if $d_s$ exists and is finite. In that case, 
\begin{equation*}
\lim_{E \to 0} \frac{\ln \sigma([1-E,1])}{\ln E} = \frac{d_s}{2} \, .
\end{equation*}
\end{proposition}
For a proof, see Appendix \ref{ap:proof-spectral-dimension-spectral-decay}. This proposition gives a rigorous equivalent to intuition \eqref{eq:intuition}. It uses the spectral dimension of the graph, which is an intrinsic property of the graph and turns out to coincide with our intuition of the dimension of a graph in examples of interest. Note that in Section \ref{sec:second_order_gossip}, the spectral measure $\sigma$ is defined as $\diff\sigma(\lambda) = \sum \langle\xi,u^i\rangle^2 \delta_{\lambda_i}$ whereas in this section, it is defined for finite graphs as $\diff\sigma(\lambda) = \sum (u^i_v)^2 \delta_{\lambda_i}$. Roughly speaking, intuition \eqref{eq:intuition} is valid for the former if $\xi$ projects evenly on all eigenvectors $u^i$. It is the case if $\xi$ has random i.i.d.~components for instance; this is used in Section \ref{sec:performance-guarantees}.

\section{Performance guarantees in graphs of spectral dimension $d$}
\label{sec:performance-guarantees}

In this section, we seek to give theoretical support to the empirical observations of Section \ref{sec:simulations}: Jacobi polynomial gossip improves on the non-asymptotic phase over existing methods. This is challenging because the analysis of gossip methods is simpler in the asymptotic regime. In our case, we use asymptotic properties of the Jacobi polynomials as $t \to \infty$. 

In order to be able to run an asymptotic analysis without falling in the asymptotic phase of exponential convergence, we run our method on infinite graphs $G = (V,E)$. In infinite graphs, it is impossible for information to have reached every node in any finite time. In practice, the conclusions drawed on infinite graphs should be taken as approximations of the behavior on very large graphs.

Of course, it is impossible for any gossip method to estimate the average of the values in the infinite graphs: indeed, within time $t$ the node $v$ can only share information with nodes that are closer than $t$ (w.r.t.~the shortest path distance in the graph). Even worse, the average of an infinite number of values is ill-defined. Thus additional assumptions on the observations $\xi_v$ are needed. Several choices could be possible here, to keep the discussion simple we assume that the observations $\xi_v$ are independent identically distributed (i.i.d.)~samples from a probability law $\nu$. The agents then seek to estimate the statistical mean $\mu = \int_\R \xi \diff\nu(\xi)$ of $\nu$. 

In practice, to build good estimates, the nodes should average their samples, thus it is natural to run gossip algorithms in this situation. An estimator performs well if it averages a lot of samples and averages them uniformly. Thus the mean square error (MSE) of the estimators measures the capacity of a gossip methods to average locally in the graph. 

This statistical gossip framework was already present in \cite{braca2008enforcing} and is not only used for its technical advantages. It is also a reasonable modeling of gossip of signals with a statistical structure in large networks. For instance, in sensor networks, observations are measurements of the environment corrupted by noise. The purpose of the gossip algorithm is to average observations to get a better estimate of the ground truth. Gossip algorithms are also used as building blocks in distributed statistical learning problems such as distributed optimization (see \cite{nedic2009distributed,scaman2017optimal,sayed2014adaptation,ram2010distributed,duchi2012dual,chen2012diffusion}) or distributed bandit algorithms (see \cite{szorenyi2013gossip,landgren2016distributed,korda2016distributed}). All of these problems have a statistical structure that simplifies the underlying gossip problem. For instance, in sensor networks, good estimates of the mean may not require using observations from nodes extremely far in the network.

\medskip
Let us now sum up the setting. The network of agents is modeled by a (possibly infinite, locally finite) graph $G=(V,E)$, that we endow with a gossip matrix $W$. We consider a probability law $\nu$ on $\R$, and $\mu = \int_{\R} \xi \, d\nu(\xi)$ its statistical mean. Each agent $v \in V$ is given a sample from $\nu$:  
\begin{equation*}
\xi_v, v \in V \underset{\rm i.i.d.}{\sim} \nu \, .
\end{equation*}
The following theorem gives the asymptotic MSE of the estimators built by the simple gossip method and the Jacobi polynomial iteration.

\begin{theorem}
	\label{thm:rate-decrease-Jacobi}
	Fix a vertex $v$ and denote $d_s = d_s(G,W,v)$ the spectral dimension of the graph.
	\begin{enumerate}
		\item\label{enu:simple-gossip} Let $x^t$ be the iterates of the simple gossip method \eqref{eq:simple_gossip}, or the iterates of the shift-register gossip method \eqref{eq:shift-register} with some parameter $\omega \in [1,2]$. Then
			\begin{equation}
			\liminf_{t\to\infty} \frac{\ln \E[(x^t_v - \mu)^2]}{\ln t} \geq - \frac{d_s}{2} \, .
			\end{equation} 
		\item\label{enu:jacobi} Let $x^t$ be the iterates of the Jacobi polynomial iteration \eqref{eq:Jabobi-polynomial-iteration} with parameter $d=d_s$. Then 
			\begin{equation}
			\label{eq:rate-decrease-Jacobi}
			\limsup_{t\to\infty} \frac{\ln \E[(x^t_v - \mu)^2]}{\ln t} \leq - d_s \, .
			\end{equation}
	\end{enumerate}
\end{theorem}

See Appendix \ref{ap:proof-thm-rate-decrease-Jacobi} for a proof. The above theorem shows that the asymptotic MSE of the Jacobi polynomial iteration can be upper bounded using only the spectral dimension of the graph. The power decay of the MSE with the Jacobi polynomial iteration enjoys a better rate than with simple gossip and the shift-register iteration (regardless of the choice of $\omega$). In some cases, this rate can be proved optimal using the Hausdorff dimension of the graph. 

\begin{definition}[Hausdorff dimension]
The Hausdorff dimension of the graph $G$ at vertex $v$ is, if it exists, the limit 
\begin{equation*}
d_h = d_h(G,v) = \lim_{t \to\infty} \frac{\ln |B_t(v)|}{\ln t} \, .
\end{equation*}
If $G$ is connected, then $d_h$ does not depend on the choice of $v$.
\end{definition}

\begin{proposition}
	\label{prop:lower-bound}
Let $x^t = P_t(W)\xi$ be any polynomial gossip method on a graph $G$ with Hausdorff dimension $d_h$. Then 
\begin{equation}
\liminf_{t\to\infty} \frac{\ln \E[(x^t_v - \mu)^2]}{\ln t} \geq -d_h \, .
\end{equation}
\end{proposition}
See Appendix \ref{ap:proof-lower-bound} for a proof. Note that this lower bound it attained if $x^t$ is the local average of values:
\begin{equation*}
x^t_v = \frac{1}{|B_t(v)|} \sum_{w \in B_t(v)} \xi_w \, .
\end{equation*} 
Thus reaching this lower bound means that the polynomial gossip method averages locally. Theorem \ref{thm:rate-decrease-Jacobi} shows that it is the case with the Jacobi polynomial iteration if $d= d_s = d_h$.
\begin{corollary}
	Assume that the spectral and the Hausdorff dimensions have the same value $d= d_h = d_s$. If $x^t$ are the iterates of the Jacobi polynomial iteration \eqref{eq:Jabobi-polynomial-iteration}, we obtain the optimal asymptotic convergence rate
	\begin{equation*}
	\lim_{t\to\infty} \frac{\ln \E[(x^t_v - \mu)^2]}{\ln t} = - d_h \, .
	\end{equation*}
\end{corollary}
\smallskip
\textbf{Application to the grid.} Proposition \ref{prop:spectral-dim-Z^d} states that the spectral dimension of $\Z^d$ is $d$, which coincides the Hausdorff dimension. 
\begin{corollary}
Let $x^t$ be the iterates of the Jacobi polynomial iteration \eqref{eq:Jabobi-polynomial-iteration} on the grid $\Z^d$. Then we obtain the optimal asymptotic convergence rate
\begin{equation*}
\lim_{t\to\infty} \frac{\ln \E[(x^t_v - \mu)^2]}{\ln t} = - d \, .
\end{equation*}
\end{corollary}
Note that Theorem \ref{thm:rate-decrease-Jacobi} also gives that if $x^t$ are the iterates of the simple gossip method, then $\lim_{t\to\infty} \ln \E[(x^t_v - \mu)^2] / \ln t = - d/2$. (The theorem actually only gives the lower bound, but the proof technique, combined with the fact that the spectrum of $\Z^d$ is symmetric, actually gives the result.) This result could have been anticipated intuitively as follows. Under the simple gossip iteration, the information of the measurement $\xi_v$ diffuses following a simple random walk on the grid. According to the central limit theorem, at large $t$, the information is approximately distributed according to a Gaussian distribution of standard deviation $\sqrt{t}$, which is approximately supported by $\Theta(\sqrt{t}^d)$ nodes. This means that at time $t$, a node $v$ gets the information of $\Theta(t^{d/2})$ neighbors. As a consequence, the MSE $\E[(x^t_v - \mu)^2]$ scales like $t^{-d/2}$.

\smallskip
\textbf{Application to the percolation bonds.} Let $G$ be the random infinite cluster of a supercritical percolation in $\Z^d$ as defined in Proposition \ref{prop:spectral-dimension-percolation}. The proposition gives that the spectral dimension of $G$ is a.s.~$d$, which is also a lower bound for the Hausdorff dimension. But it is trivial that the Hausdorff dimension is smaller than $d$, thus the two coincide.
\begin{corollary}
Let $G$ be the random infinite cluster of a supercritical percolation in $\Z^d$, and $v\in\Z^d$. Let $x^t$ be the iterates of the Jacobi polynomial iteration \eqref{eq:Jabobi-polynomial-iteration}. Then a.s.~on the event $\{v\in G\}$, 
\begin{equation*}
\lim_{t\to\infty} \frac{\ln \E_\xi[(x^t_v - \mu)^2]}{\ln t} = - d \, .
\end{equation*}
\end{corollary}

\begin{remark}
The Jacobi polynomial iteration \eqref{eq:Jabobi-polynomial-iteration} is derived so that $x^t = \pi_t^{(\alpha,\beta)}(W)\xi$, where $\pi_t^{(\alpha,\beta)}$ are the orthogonal polynomials w.r.t.~the Jacobi measure {$\sigma^{(\alpha,\beta)}(\diff\lambda) = (1-\lambda)^\alpha(1+\lambda)^\beta\diff\lambda$} on $[-1,1]$, with $\alpha = d/2, \beta=0$, $d$ is the spectral dimension. A curious reader could wonder what happens for other choices of $\alpha$ and $\beta$ (while keeping $d$ fixed). This question is investigated at length in Appendix \ref{ap:tuning-Jacobi}. The conclusion is that the natural choice $\alpha = d/2, \beta=0$ is optimal (up to constant factors) but there are other choices that are optimal.
\end{remark}

\section{The Jacobi polynomial iteration with spectral gap}
\label{sec:jacobi-with-spectral-gap}

In this section, we adapt the Jacobi polynomial iteration to the case where the spectral gap $\gamma$ of the gossip matrix $W$ is given. This allows to obtain accelerated asymptotic rates of convergence, that compete with the state-of-the-art accelerated algorithms for gossip. 

We assume that we are given the spectral dimension $d$ of the graph, which determines the density of eigenvalues near $1$, and the spectral gap $\gamma = 1-\lambda_2(W)$, the distance between the largest and the second largest eigenvalue. Given these parameters, we can approximate the spectral measure of~$W$ with 
\begin{equation*}
\diff\tilde{\sigma}(\lambda) = ((1-\gamma)-\lambda)^{d/2-1}\bfone_{\{\lambda\in (-1,1-\gamma)\}}\diff\lambda \, .
\end{equation*}
Following the recommendation of Proposition \ref{prop:opt_polynomial}, this means that we should consider the polynomial iteration associated with the orthogonal polynomials w.r.t.~$(1-\lambda)\diff\tilde{\sigma}(\lambda) = (1-\lambda)((1-\gamma)-\lambda)^{d/2-1}\bfone_{\{\lambda\in (-1,1-\gamma)\}}\diff\lambda$. We do not know how to compute the recurrence formula for this measure, thus we used the orthogonal polynomials w.r.t.~$((1-\gamma)-\lambda)\diff\tilde{\sigma}(\lambda) = ((1-\gamma)-\lambda)^{d/2}\bfone_{\{\lambda\in (-1,1-\gamma)\}}\diff\lambda$, which is a rescaled version of a Jacobi measure. The corresponding polynomial method is called the \emph{Jacobi polynomial iteration with spectral gap}.

A recursive formula for orthogonal polynomials w.r.t.~$((1-\gamma)-\lambda)\diff\tilde{\sigma}(\lambda)$ is derived in Section~\ref{ap:recurrence-relation-rescaled-jacobi}. Taking $\alpha = d/2$ and $\beta = 0$ in equations \eqref{eq:rescaled-Jacobi}, we get the practical recursion:
\begin{equation}
\label{eq:Jac_pol_it_gap}
\begin{aligned}
&x^t = \frac{y^t}{\delta_t} \, ,  \\
&y^0 = \xi \, , \qquad \delta_0 = 1 \, , \\
&y^1 = a_0^{(d,\gamma)}W\xi + b_0^{(d,\gamma)}\xi \, , \qquad \delta_1 = a_0^{(d,\gamma)} + b_0^{(d,\gamma)} \, , \\
&y^{t+1} = a_t^{(d,\gamma)}Wy^t + b_t^{(d,\gamma)}y^t-c_t^{(d,\gamma)}y^{t-1} \, , \qquad t \geq 1 \, ,\\
&\delta_{t+1} = \left(a_t^{(d,\gamma)} + b_t^{(d,\gamma)}\right)\delta_t-c_t^{(d,\gamma)}\delta_{t-1} \, , \qquad t \geq 1 \, , \\
&a_t^{(d,\gamma)} = a_t^{(d)}\left(1-\frac{\gamma}{2}\right)^{-1} \, , \qquad b_t^{(d,\gamma)} = b_t^{(d)} + \frac{\gamma}{2}\left(1-\frac{\gamma}{2}\right)^{-1}a_t^{(d)} \, , \qquad t \geq 0 \, , \\
&c_t^{(d,\gamma)} = c_t^{(d)} \, , \qquad t \geq 1 \, , 
\end{aligned}
\end{equation}
where the coefficients $a_t^{(d)}, b_t^{(d)}, c_t^{(d)}$ are defined in \eqref{eq:coeffs-Jacobi-d}.
\begin{theorem}[Asymptotic rate of convergence]
	\label{thm:asymptotic-rate-with-spectral-gap}
Let $\gamma>0$ be a lower bound on the spectral gap of the gossip matrix $W$ and $d$ any positive real. Let $\xi = (\xi^t_v)_{v\in V}$ be any family of initial observations and $x^t = (x^t_v)_{v\in V}$ be the sequence of iterates generated by the Jacobi polynomial iteration with spectral gap \eqref{eq:Jac_pol_it_gap}. Then
\begin{equation*}
\limsup_{t \to \infty} \Vert x^t-\bar{\xi}\bfone \Vert_2^{1/t} \leq \frac{1-\gamma/2}{(1+\sqrt{\gamma/2})^2} \, .
\end{equation*}
\end{theorem}
This shows that the Jacobi polynomial iteration with spectral gap enjoys linear convergence. The asymptotic rate of convergence is equivalent to $1-\sqrt{2\gamma}$ as $\gamma \to 0$. This justifies that we obtain an accelerated asymptotic rate of convergence that compares with the state-of-the art accelerated gossip methods (see Figure \ref{fig:sim_2D_grid_log_all}).

Note that the asymptotic rate of convergence does not depend on $d$. However, the choice of~$d$ may have an important effect during the non-asymptotic phase $t < 1/\sqrt{\gamma}$. In this phase, the spectral gap $\gamma$ can be neglected in the approximation of the spectral measure, and it is important that the densities of eigenvalues of $\sigma$ and $\tilde{\sigma}$ match near the upper edge of the spectrum. This is why one should choose $d$ as the spectral dimension of the graph. 

\section{The parallel between the gossip methods and distributed Laplacian solvers}
\label{sec:laplacian-solvers}

There is a natural parallel between gossip methods and iterative methods that solve linear systems. Loosely speaking, simple gossip corresponds to gradient descent on the quadratic minimization problem associated to the linear system, shift-register gossip to Polyak's heavy-ball method and the parameter-free polynomial iteration to the conjugate gradient algorithm (see \cite{fischer1996polynomial} or \cite{polyakintroduction} for references on these subjects). In this parallel, the fact that we can reach perfect gossip in $n$ steps (see Remark \ref{rem:perfect-gossip}) translates into the finite convergence of the conjugate gradient algorithm in a number of iterations equal to the dimension of the ambient space. In the distributed resolution of linear systems, the problem that the recursion coefficients $a_t, b_t, c_t$ can not be computed in a centralized manner has also appeared and it motivated the development of inner-product free iterations.

The Jacobi polynomial iterations presented above were motivated by the facts that (a) the parameter-free polynomial iteration is not feasible in the distributed setting of gossip, and (b) the gossip matrix $W$ exhibits a structure due to the low-dimension manifold on which the agents live. Interestingly, the literature on multi-agent systems  deals with some minimization problems with the same properties. Examples are given by the estimation of quantities on graphs from relative measurements, in which the agents $v \in V$ try to estimate some quantity $x_v, v\in V$ defined over the graph, from noisy relative measurements over the edges of the graph:
\begin{equation*}
\xi_{v,w} = x_v - x_w + \eta_{v,w} \, , \qquad \{v,w\} \in E \, . 
\end{equation*}
This problem has applications in network localization, where the $x_v$ are the positions of the agents and the $\xi_{v,w}$ come from measurements of the distances and directions between the neighbors. It also has similar applications in time synchronization of clocks over networks, where $x_v$ is the offset of the clock of node $v$; and to motion consensus, where $x_v$ is the speed of agent $v$. For an introduction to estimation on graphs from relative measurements and its applications, see \cite{barooah2008estimation} and references therein. Note that the quantities $x_v$ can only be determined up to a global constant from the measurements; either we seek the true solution up to a constant only, either we assume that some agents know their true value. 

In natural approach to solve the problem is to determine estimates $y_v$ of $x_v$ that minimize 
\begin{equation*}
\frac{1}{2} \sum_{v,w} W_{v,w}\left(\xi_{v,w}-(y_v-y_w)\right)^2 \, ,
\end{equation*}
where $W_{v,w}$ are some weights on the edges of the graph. Indeed, this corresponds to finding the maximum likelihood estimator if the noise $\eta_{v,w}$ is i.i.d.~Gaussian and $W_{v,w}$ is the inverse variance of $\eta_{v,w}$. The above minimization problem is a quadratic problem whose covariance matrix is the Laplacian $I-W$. It can be solved using gradient descent or spectral-gap based accelerations like the heavy-ball method. However, the conjugate gradient algorithm can not be applied here as it involves centralized computations. The Jacobi polynomial iterations developed in this paper can be adapted to this situation in order to develop accelerations exploiting the structure of the Laplacian $I-W$. Experimenting how this performs in real-world situations is left for future work. 

\section{Message passing seen as a polynomial gossip algorithm}
\label{sec:MP}

This section develops another application of the polynomial point-of-view on gossip algorithms. It is independent of the Jacobi polynomial iterations developed in Sections \ref{sec:Jacobi-polynomial-gossip}-\ref{sec:jacobi-with-spectral-gap}; we show that the message passing algorithm for gossip of \cite{moallemi2005consensus} has a natural derivation as a polynomial gossip algorithm and uses this point of view to derive convergence rates.

The message passing algorithm of \cite{moallemi2005consensus} (in its zero-temperature limit) defines quantities on the edges of the graph $G$ with the following recursion: for $v,w \in V$ linked by an edge in the graph $G$, it defines $K_{vw}^0 = 0$, $M_{vw}^0 = 0$, and 
\begin{align}
& K_{vw}^{t+1} = 1 + \sum_{u \in \cN(v), \, u \neq w} K_{uv}^{t} \, , &&  M_{vw}^{t+1} =\frac{1}{K_{vw}^{t+1}} \bigg(\xi_v + \sum_{u \in \cN(v), \, u \neq w} K_{uv}^{t} M_{uv}^{t}\bigg)  \, , \label{eq:MP_1}
\end{align}
where $\cN(v)$ denotes the set of neighbors of $v$. $K_{vw}$ and $M_{vw}$ are interpreted as messages going from $v$ to $w$ in $G$: $M_{vw}^t$ corresponds to an average of observations gathered by $v$ and transmitted to $w$; $K_{vw}^t$ is the corresponding number of observations. We recommend \cite[Section II.A]{moallemi2005consensus} and Lemma \ref{lem:MP_technical_lemma} for a detailed description of this intuition. At each time step $t$, the output of the algorithm is 
\begin{equation}
\label{eq:MP_2}
x_v^t = \frac{\xi_v + \sum_{u \in \cN(v)} K_{uv}^{t} M_{uv}^{t}}{1 + \sum_{u \in \cN(v)} K_{uv}^{t}} \, .
\end{equation}
This gossip methods performs exact local averaging on trees, as shown by the following proposition. 
\begin{proposition}
	\label{prop:message-passing_trees}
	Assume that $G$ is a tree. Then for all $t \geq 1$, $v \in V$,
	\begin{align*}
	&x_v^t = \frac{1}{\vert  B_{v}(t) \vert }  \sum_{w \in B_{v}(t)} \xi_w \, .
	\end{align*}
\end{proposition}
See Appendix \ref{ap:proof-message-passing-trees} for a proof. Nothing prevents from running the message passing recursion \eqref{eq:MP_1}-\eqref{eq:MP_2} in a graph $G$ with loops. In the case of regular graphs, we are able to interpret the message passing algorithm as a polynomial gossip algorithm.
\begin{theorem}
	\label{thm:connection_mp_polynomial}
	Assume $G$ is $d$-regular, meaning that each vertex has degree $d$, $d \geq 2$. Assume further that $W = A(G)/d$. Denote $\sigma(\T_d) = \sigma(\T_d,W,v)$ the spectral measure of the infinite $d$-regular tree at any vertex $v$ (see Definition \ref{def:spectral_measure_infinite_graph}). Then the output $x^t$ of the message passing algorithm \eqref{eq:MP_1}-\eqref{eq:MP_2} on $G$ can also be obtained as $x^t = \pi_t(W)\xi$ where $\pi_0, \pi_1, \dots$ are the orthogonal polynomials w.r.t.~$(1-\lambda)\sigma(\T_d)(\diff\lambda)$.
\end{theorem} 

See Appendix \ref{ap:proof-connection-mp-polynomial} for a proof. In words, the theorem above states that message passing corresponds to the best polynomial gossip algorithm when one \emph{believes} the graph is a tree. This is not surprising as message passing algorithms are often derived by neglecting loops in a graph.

An easy follow-up of this theorem is that the iterates $x^t$ defined in \eqref{eq:MP_1}-\eqref{eq:MP_2} follow a second-order recursion (in $d$-regular graphs). Actually the spectral measure $\sigma(\T_d)$ of the infinite $d$-regular tree, also called the Kesten-McKay measure, can be computed explicitly (see \cite[Section 2.2]{sodin2007random}),
\begin{equation*}
\sigma(\T_d)(\diff\lambda) = \frac{d}{2\pi (1 - \lambda^2)}  {\left({\frac{4(d-1)}{d^2}-\lambda^2}\right)^{1/2}} \bfone_{\left[-2\sqrt{d-1}/d,2\sqrt{d-1}/d\right]}(\lambda)\diff\lambda \, .
\end{equation*}
The recurrence relation of the modified Kesten-McKay measure $(1-\lambda)\sigma(\T_d)(\diff\lambda)$ is derived in Appendix \ref{ap:Kesten-McKay}. It shows that 
\begin{align*}
&x^0 = \xi \, , \qquad x^1 = a_0 W\xi + b_0 \xi \, , \qquad x^{t+1} = a_t W x^t - c_t x^{t-1} \, , \\
&a_0 = \frac{d}{d+1} \, , \qquad b_0 =  \frac{1}{d+1} \, , \qquad a_t = \frac{\frac{d}{d-1}-2(d-1)^{-(t+1)}}{1-\frac{2}{d}(d-1)^{-(t+1)}} \, , \qquad c_t = \frac{\frac{1}{d-1}-\frac{2}{d}(d-1)^{-t}}{1-\frac{2}{d}(d-1)^{-(t+1)}} \, , \quad t \geq 1 \, .
\end{align*} 
Theorem \ref{thm:connection_mp_polynomial} gives a way to study the convergence of the message passing algorithms on $d$-regular graphs with loops. For instance, using asymptotic properties of the orthogonal polynomials w.r.t.~$(1-\lambda)\sigma(\T_d)(\diff\lambda)$, we obtain the convergence rate of the message passing algorithm as a function of the spectral gap $\gamma$ of the matrix:
\begin{theorem}
\label{thm:convergence_rate_message_passing}
Assume $G$ is $d$-regular, meaning that each vertex has degree $d$, $d \geq 3$. Assume further that $W = A(G)/d$, and denote $\tilde{\gamma}$ its absolute spectral gap. Let $\xi = (\xi^t_v)_{v\in V}$ be any family of initial observations and $x^t = (x^t_v)_{v\in V}$ be the sequence of iterates generated by equations \eqref{eq:MP_1}-\eqref{eq:MP_2}. Then 
\begin{enumerate}
	\item If $\tilde{\gamma} < 1- 2\sqrt{d-1}/d$,
	\begin{equation*}
	\limsup_{t \to \infty} \Vert x^t-\bar{\xi}\bfone \Vert_2^{1/t} \leq  \frac{(1-\tilde{\gamma})+\sqrt{(1-\tilde{\gamma})^2-4(d-1)/d^2}}{1+\sqrt{1-4(d-1)/d^2}} \, .
	\end{equation*}
	Moreover, the upper bound is reached if there exists an eigenvector $u$ corresponding to an eigenvalue of $W$ of magnitude $1-\tilde{\gamma}$ such that $\langle u,\xi \rangle \neq 0$.
	\item If $\tilde{\gamma} \geq 1- 2\sqrt{d-1}/d$,
	\begin{equation*}
	\limsup_{t \to \infty} \Vert x^t-\bar{\xi}\bfone \Vert_2^{1/t} \leq  \frac{2\sqrt{d-1}/d}{1+\sqrt{1-4(d-1)/d^2}} \, .
	\end{equation*}
\end{enumerate}
\end{theorem}
A consequence of this theorem is that the rate of convergence of the message passing algorithm is $1-c\tilde{\gamma}+o(\tilde{\gamma})$ as $\tilde{\gamma}\to 0$, for some constant $c$. This proves that message passing has a diffusive (or unaccelerated) behavior on graphs with a small spectral gap. Figure \ref{fig:sim_2D_grid_log_mp} shows this diffusive convergence rate on the 2D grid. 
\begin{figure}
	\begin{subfigure}{0.49\linewidth}
		\includegraphics[width=\linewidth, trim = 0 0 0 0]{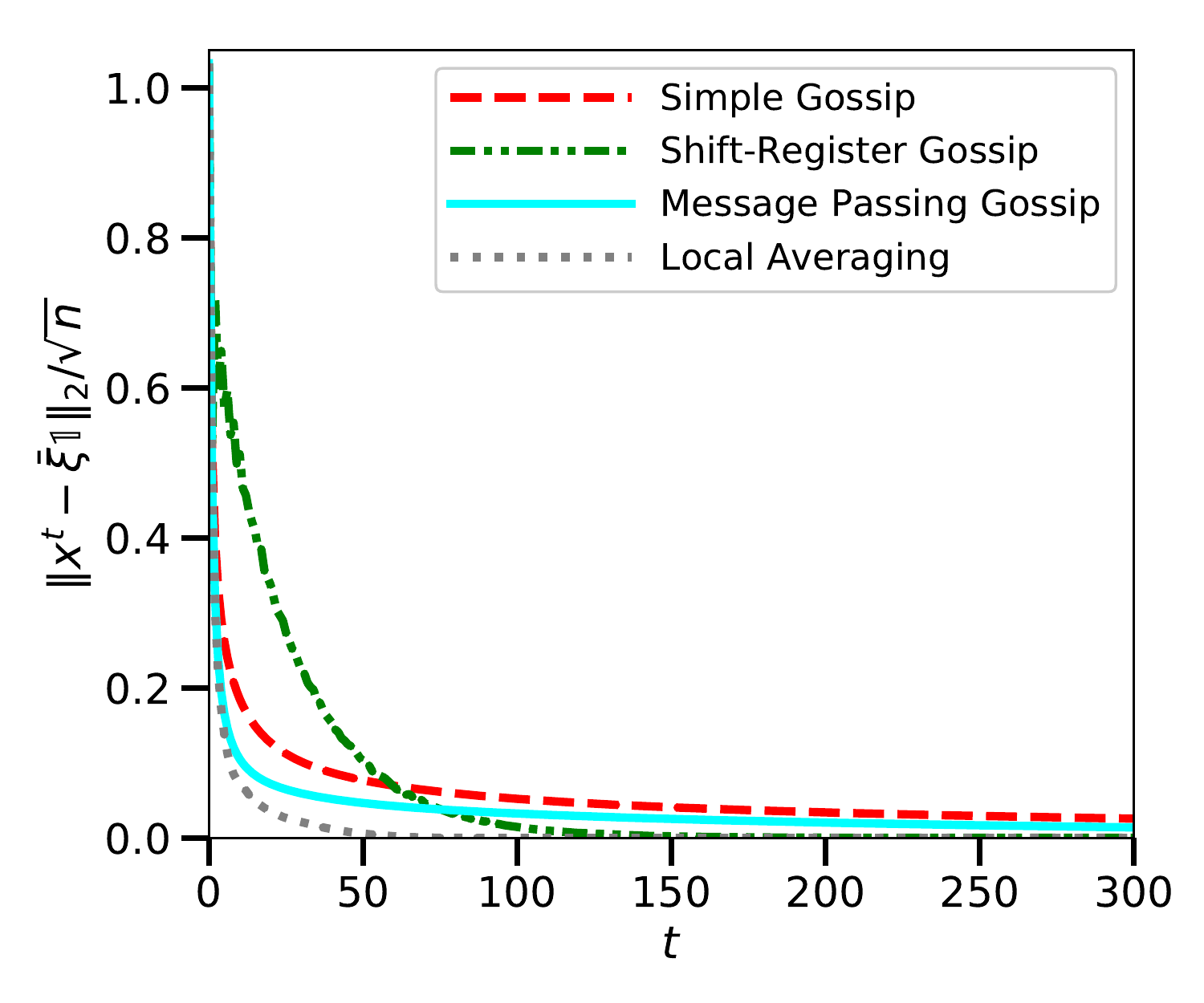}
		\caption{in linear scale}
	\end{subfigure}
	\begin{subfigure}{0.49\linewidth}
		\includegraphics[width=\linewidth, trim = 0 0 0 0]{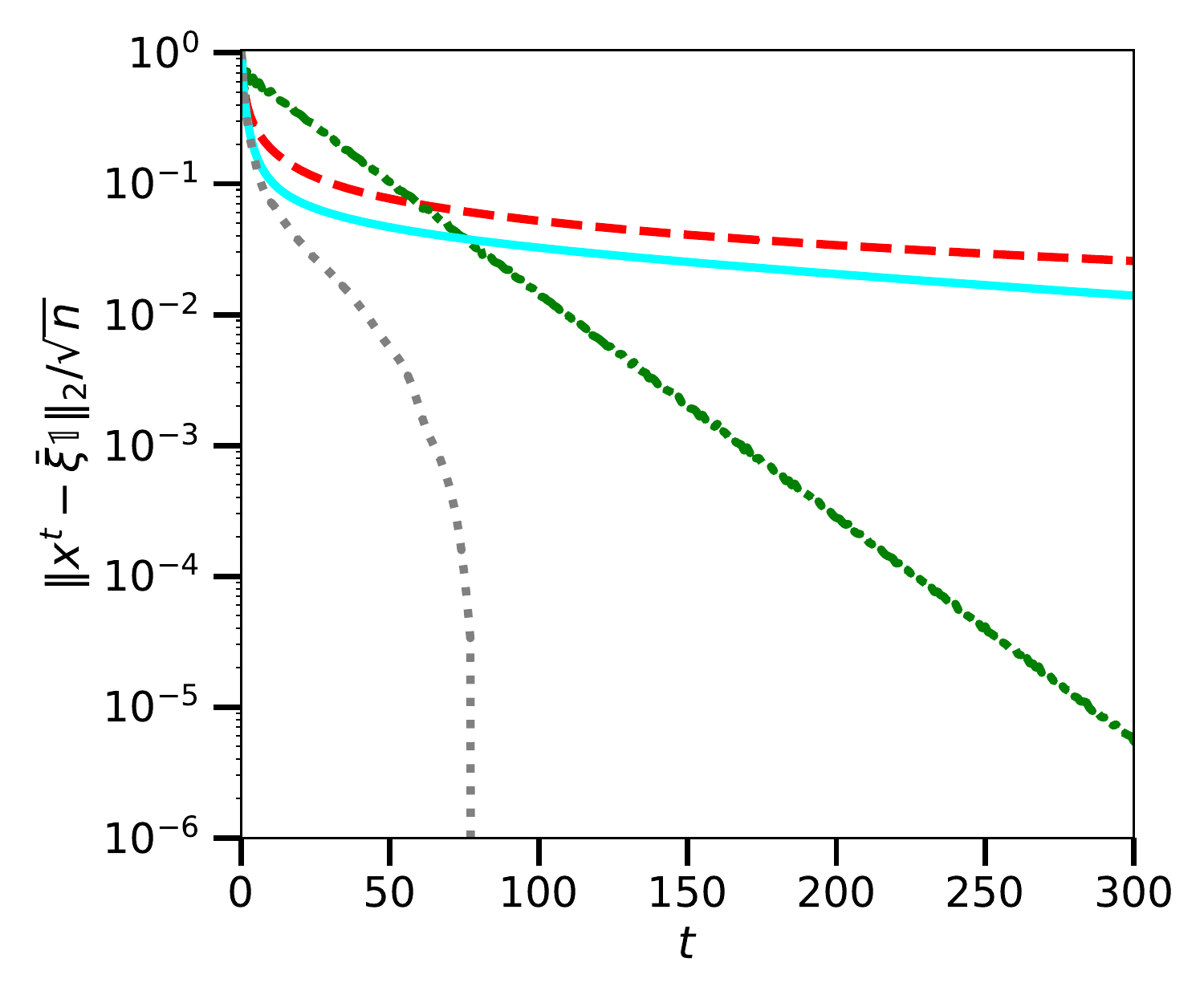}
		\caption{in log-scale}
	\end{subfigure}
	\caption{Performance of different gossip algorithms running on the 2D grid.}
	\label{fig:sim_2D_grid_log_mp}
\vspace{1cm}
	\begin{subfigure}{0.49\linewidth}
		\includegraphics[width=\linewidth, trim = 0 0 0 0]{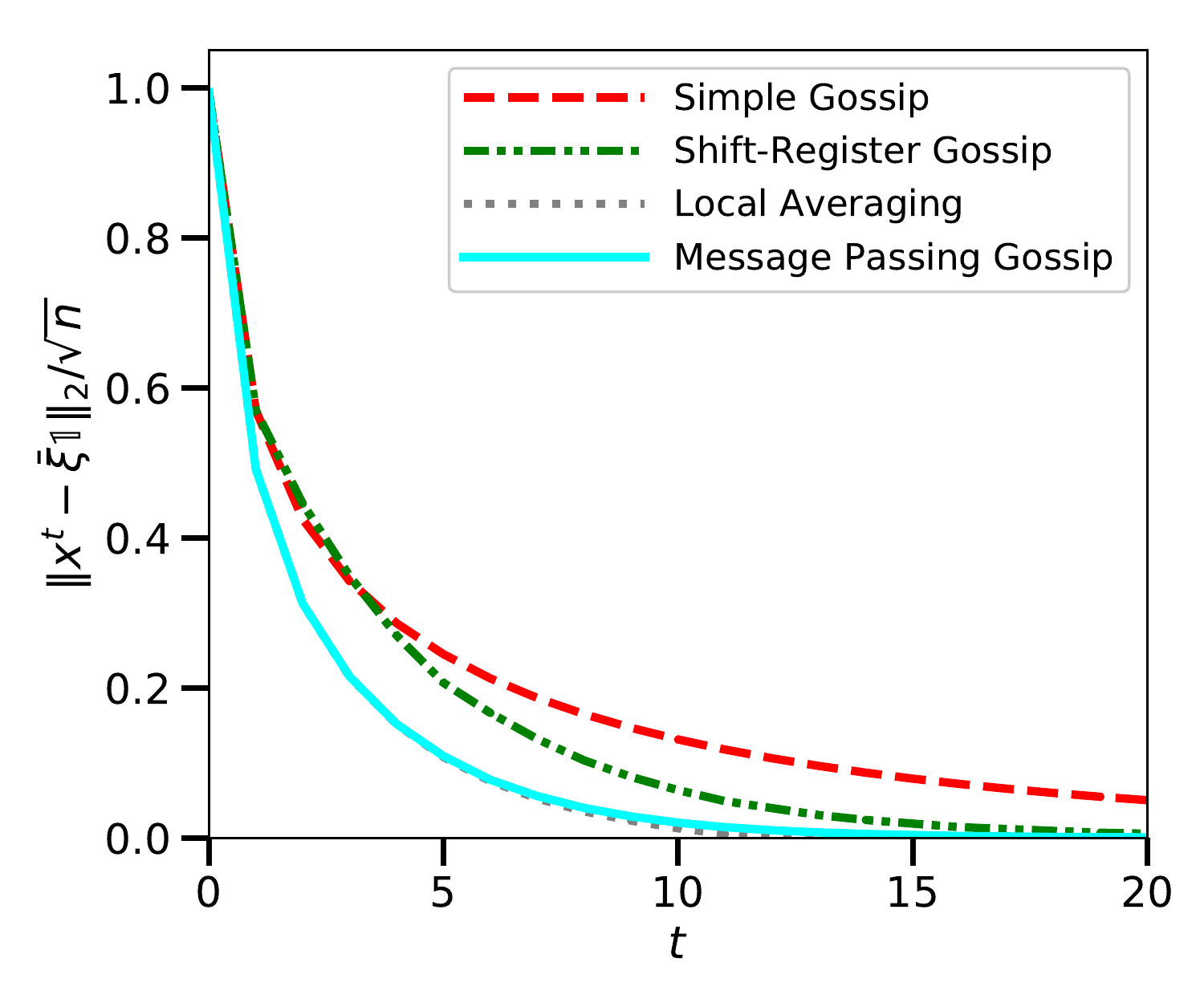}
		\caption{in linear scale}
	\end{subfigure}
	\begin{subfigure}{0.49\linewidth}
		\includegraphics[width=\linewidth, trim = 0 0 0 0]{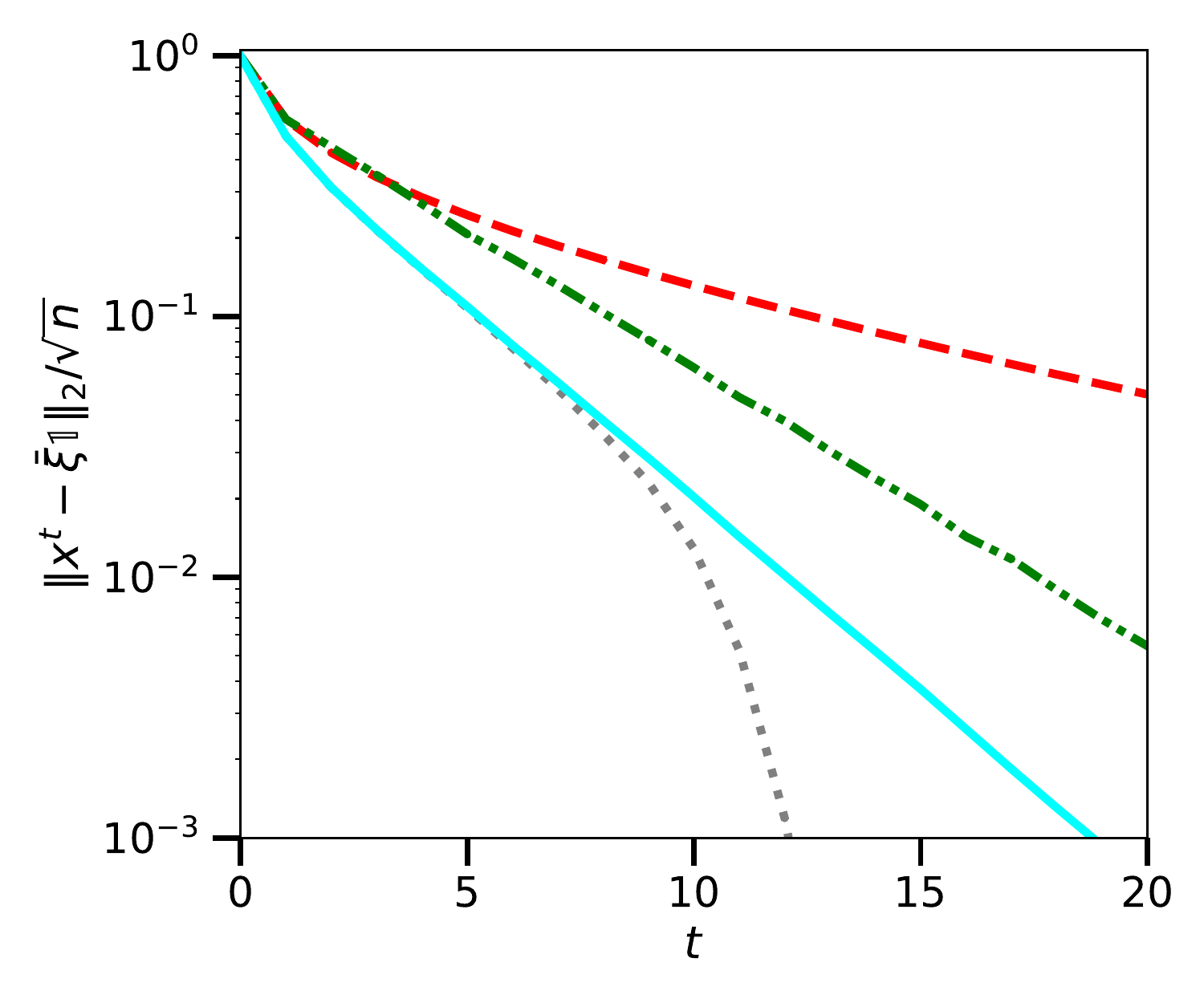}
		\caption{in log-scale}
	\end{subfigure}
	\caption{Performance of different gossip algorithms running a uniformly random $3$-regular graph of size $n=2000$.}
	\label{fig:rrg_curve}
\end{figure}
\medskip

However, the message passing algorithm can be competitive in situations with a large spectral gap. For instance, McKay's Theorem \cite[Theorem 1.1]{sodin2007random} states that the spectral measure of a uniformly random $d$-regular graph on $n$ vertices converges to the spectral measure $\sigma(\T_d)$ of the $d$-regular tree (in law, for the weak-convergence topology). This suggests that the message passing algorithm is well-suited for uniformly sampled large regular graphs. We give simulations in Figure~\ref{fig:rrg_curve} on uniformly sampled $3$-regular graphs of size $n=2000$. The results were averaged over $10$ graphs. We observe that in this case, message passing matches closely the lower-bound. Note that in this case, we do not have a diffusive rate of convergence because the spectral gap $\gamma$ does not converge to $0$ as $n \to \infty$ (see \cite{friedman2003proof} for a proof that $\gamma \to 1-2\sqrt{d-1}/d$ in probability).

\section{Conclusion}

Gossip methods based on the spectral gap were designed to improve the slow convergence rate of simple gossip. However, these methods are paradoxically bad at averaging locally in the intermediate regime before consensus is reached. In this paper, we propose another acceleration of simple gossip based on (i) the polynomial-based point of view, which designs iterations that are efficient at all times, and (ii) the Jacobi approximation, which uses prior information on the spectral dimension of the graph, a more natural property than the spectral gap. 

This paper advocates for the use of the polynomial point of view to design gossip algorithm, as it allows to use different types of prior information about the graph (spectral gap, spectral dimension, tree-like structure, etc.) and gives tools to prove the convergence of the designed algorithms.  

\section*{Acknowledgements}

We acknowledge support from the European Research Council (grant SEQUOIA 724063) and from the DGA. 

\bibliographystyle{amsalpha}
\bibliography{paper_gossip}

\appendix

\section{Details of the simulations of Section \ref{sec:simulations}}
\label{ap:details-simulations}

\textbf{2D grid.}
We run simulations on a $40 \times 40$ square lattice ($n \!=\! 1600$ vertices) endowed with the gossip matrix defined in \eqref{eq:natural_gossip_matrix} with $d_{\rm max} = 4$. The results are plotted in Figure {\scshape\ref{fig:sim_2D_grid}} and a $20 \times 20$ grid is plotted in Figure {\scshape\ref{fig:grid_drawing}} for visualization.

\textbf{3D grid.} We run simulations on a $12 \times 12 \times 12$ cubic lattice ($n = 1728$ vertices) endowed with the gossip matrix defined in \eqref{eq:natural_gossip_matrix} with $d_{\rm max} = 6$. The results are plotted in Figure {\scshape \ref{fig:sim_3D_grid}}.  

\textbf{2D percolation bond.} We build a 2D percolation bond by taking a $40 \times 40$ 2D grid, and keep each edge independently with probability $p= 0.6$. To avoid connectivity issues, we consider $G$ the largest connected component of the resulting graph, endowed with the gossip matrix defined in \eqref{eq:natural_gossip_matrix} with $d_{\rm max} = 4$. The results are plotted in Figure {\scshape\ref{fig:sim_2D_percolation}} and a $20 \times 20$ percolation bond is plotted in Figure {\scshape\ref{fig:percolation_drawing}} for visualization.

\textbf{3D percolation bond.} We build a 3D percolation bond by taking a $12 \times 12 \times 12$ 3D grid and keep each edge independently with probability $p= 0.4$. To avoid connectivity issues, we consider $G$ the largest connected component of the resulting graph, endowed with the gossip matrix defined in \eqref{eq:natural_gossip_matrix} with $d_{\rm max} = 6$. The results are plotted in Figure {\scshape\ref{fig:sim_3D_percolation}}.

\textbf{2D random geometric graph.}
We build a 2D random geometric graph $G$ by sampling $n = 1600$ points uniformly in the unit square $[0,1]^2$ and linking pairs closer than $1.5/\sqrt{n} = 0.0375$. To avoid connectivity issues, we consider $G$ the largest connected component of the resulting graph. We build a gossip matrix $W$ on $G$ with the formulas: $W_{vw}  = \max(\deg v, \deg w)^{-1} \quad \textrm{if }v \in \cN(w)$ and $W_{vv} = 1 - \sum_{w \in \cN(v)} \max(\deg v, \deg w)^{-1}$. The results are shown in Figure~{\scshape \ref{fig:sim_2D_rgg}}. 

\textbf{3D random geometric graph.} We build a 3D random geometric graph $G$ by sampling $n = 1728$ points in the unit cube $[0,1]^3$ and linking pairs closer than $1.5/n^{1/3} = 0.125$. To avoid connectivity issues, we consider $G$ the largest connected component of the resulting graph. We build a gossip matrix $W$ on $G$ with the formulas: $W_{vw}  = \max(\deg v, \deg w)^{-1} \quad \textrm{if }v \in \cN(w)$ and $W_{vv} = 1 - \sum_{w \in \cN(v)} \max(\deg v, \deg w)^{-1}$. The results are shown in Figure {\scshape \ref{fig:sim_3D_rgg}}.

\section{Toolbox from orthogonal polynomials}
\label{ap:toolbox-orthogonal-polynomials}

In this appendix, we describe the tools from the theory of orthogonal polynomials that we use in this paper. The definition of the orthogonal polynomials $\pi_t$ w.r.t.~some measure $\tau$ is given in Definition \ref{def:orthogonal_polynomials}.  We start by giving some general properties of orthogonal polynomials in Section \ref{ap:general-properties}. We then describe two parametrized measures with respect to which orthogonal polynomials can be explicitly described: the Jacobi polynomials, in Section \ref{ap:jacobi-polynomials}, and the polynomials orthogonal to some measure of the form $(1-\lambda^2)^{1/2}/\rho(\lambda)$, where $\rho$ is some polynomial, in Section \ref{ap:polynomials-orthogonal-szego}. We finally give in Section \ref{ap:asymptotics-jacobi} some asymptotic properties of the Jacobi polynomials as $t \to \infty$.

\subsection{General properties}
\label{ap:general-properties}

\begin{proposition}[{from \cite[Theorem 3.3.1]{szeg1939orthogonal}}]
	\label{prop:zeros}
	Let $\pi_t$ be a family of orthogonal polynomials w.r.t.~some measure $\tau$ on some interval $[a,b]$. Then the zeros of $\pi_t$ are real, distinct and located in the interior of $[a,b]$.
\end{proposition}

In Proposition \ref{prop:recurrence_relation}, it is stated that the orthogonal polynomials satisfy a three-term recurrence relation. We write here the short proof as it is used in Appendix \ref{ap:parameter_free_polynomial_iteration}. 

\begin{proof}[Proof of Proposition \ref{prop:recurrence_relation}]
The polynomial $\lambda \pi_t(\lambda)$ of the variable $\lambda$ is of degree $t+1$, thus it can be decomposed over the orthogonal basis $\pi_0(\lambda),\pi_1(\lambda),\dots,\pi_{t+1}(\lambda)$:
\begin{equation*}
\lambda \pi_t(\lambda) = \sum_{s=0}^{t+1} \frac{\langle \lambda \pi_t, \pi_s \rangle_\tau}{\langle \pi_s, \pi_s \rangle_\tau} \pi_s(\lambda)\, .
\end{equation*}
Note that $ \langle \lambda \pi_t, \pi_s \rangle_\tau = \int \lambda \pi_t(\lambda) \pi_s(\lambda) \diff\tau(\lambda) = \langle  \pi_t, \lambda \pi_s \rangle_\tau = 0$ when $s \leq t-2$ because in this case $\lambda \pi_s(\lambda) \in \R_{t-1}[X]$ and $\pi_t$ is orthogonal to $\R_{t-1}[X]$. Thus 
\begin{equation*}
\lambda \pi_t(\lambda) = \frac{\langle \lambda \pi_t, \pi_{t+1} \rangle_\tau}{\langle \pi_{t+1}, \pi_{t+1} \rangle_\tau} \pi_{t+1}(\lambda) + \frac{\langle \lambda \pi_t, \pi_t \rangle_\tau}{\langle \pi_t, \pi_t \rangle_\tau} \pi_t(\lambda) + \frac{\langle \lambda \pi_t, \pi_{t-1} \rangle_\tau}{\langle \pi_{t-1}, \pi_{t-1} \rangle_\tau} \pi_{t-1}(\lambda) \, ,
\end{equation*}
with the convention $\pi_{-1} = 0$. Note that $\langle \lambda \pi_t, \pi_{t+1} \rangle_\tau$ is non-zero as otherwise it would imply that $\lambda\pi_t$ is a polynomial of degree smaller or equal to $t$, which is absurd. We get the recursion formula by denoting
\begin{align}	
\label{eq:recurrence-coeffs-general}
&a_t = \frac{\langle\pi_{t+1},\pi_{t+1}\rangle_\tau}{\langle \lambda \pi_t, \pi_{t+1} \rangle_\tau} \, , 
&&b_t = -\frac{\langle \pi_{t+1}, \pi_{t+1} \rangle_\tau\langle \lambda\pi_t, \pi_t\rangle_\tau}{\langle \lambda \pi_t, \pi_{t+1} \rangle_\tau\langle\pi_t,\pi_t\rangle_\tau} \, , 
&&c_t = \frac{\langle \pi_{t+1}, \pi_{t+1} \rangle_\tau\langle\lambda\pi_t,\pi_{t-1}\rangle_\tau}{\langle \lambda \pi_t, \pi_{t+1} \rangle_\tau\langle\pi_{t-1},\pi_{t-1}\rangle_\tau} \, .
\end{align}
\end{proof}

\subsection{Jacobi polynomials}
\label{ap:jacobi-polynomials}

\begin{definition}[{from \cite[Chapter IV]{szeg1939orthogonal}}]
	\label{def:Jacobi-polynomials}
Let $\alpha, \beta>-1$. The Jacobi polynomials $P_t^{(\alpha,\beta)}$ are the orthogonal polynomials w.r.t.~the Jacobi measure
\begin{equation*}
\label{eq:jacobi-measure}
\sigma^{(\alpha,\beta)}(\diff\lambda) = (1-\lambda)^\alpha(1+\lambda)^\beta\bfone_{\{\lambda \in (-1,1)\}}\diff\lambda\, ,
\end{equation*}
normalized such that $P_t^{(\alpha,\beta)}(1) = {t + \alpha \choose t}$. 
\end{definition}

\begin{example}[{from \cite[Section 2.4]{szeg1939orthogonal}}]
	\label{ex:chebyshev}
\begin{enumerate}
	\item The Chebyshev polynomials $T_t$ of the first kind are the orthogonal polynomials w.r.t.~$\sigma^{(-1/2,-1/2)}(\diff\lambda) = (1-\lambda^2)^{-1/2}$ and normalized such that $T_t(1)=1$. They are, up to some rescaling, a family of Jacobi polynomials. They satisfy the trigonometric formula
	\begin{equation*}
	T_t(\cos \theta) = \cos(t\theta) \, .
	\end{equation*} 
	\item The Chebyshev polynomials $U_t$ of the second kind are the orthogonal polynomials w.r.t. $\sigma^{(1/2,1/2)}(\diff\lambda) = (1-\lambda^2)^{1/2}$ and normalized such that $U_t(1)=t+1$. They are, up to some rescaling, a family of Jacobi polynomials. They satisfy the trigonometric formula
	\begin{equation*}
	U_t(\cos \theta) = \frac{\sin (t+1)\theta}{\sin \theta} \, .
	\end{equation*} 
\end{enumerate}
\end{example}

A remarkable property of the Jacobi polynomials is that their recurrence relation can be computed explicitly.

\begin{proposition}[{from \cite[Section 4.4.5]{szeg1939orthogonal}}]
	\label{prop:recurrence-jacobi}
Let $\alpha, \beta > -1$. The Jacobi polynomials $P_t^{\alpha,\beta}$ satisfy the three recurrence formula
\begin{align*}
&P_0^{(\alpha,\beta)}(\lambda) = 1 \, , \quad P_1^{(\alpha,\beta)}(\lambda) = \frac{1}{2}(\alpha+\beta+2)\lambda + \frac{1}{2}(\alpha-\beta) \, , \\
&2(t+1)(t+1+\alpha+\beta)(2t+\alpha+\beta)P_{t+1}^{(\alpha,\beta)}(\lambda) \\
&\qquad= (2t+\alpha+\beta+1)[(2t+\alpha+\beta+2)(2t+\alpha+\beta)\lambda+\alpha^2-\beta^2]P_t^{(\alpha,\beta)}(\lambda)  \\
&\qquad\quad -2(t+\alpha)(t+\beta)(2t+\alpha+\beta+2)P_{t-1}^{(\alpha,\beta)}(\lambda) \, .
\end{align*}
\end{proposition}

\begin{example}
The Chebyshev polynomials of the first and the second kind satisfy the same recurrence formula, but with different initializations:
\begin{align*}
&T_0(\lambda) = 1 \, , &&T_1(\lambda) = \lambda \, , &&T_{t+1}(\lambda) = 2\lambda T_t(\lambda) - T_{t-1}(\lambda) \, , \\
&U_0(\lambda) = 1 \, , &&U_1(\lambda) = 2\lambda \, , &&U_{t+1}(\lambda) = 2\lambda U_t(\lambda) - U_{t-1}(\lambda) \, . 
\end{align*}
\end{example}

\begin{proposition}[{from \cite[Theorem 7.32.1]{szeg1939orthogonal}}]
	\label{prop:max-jacobi}
Let $\alpha, \beta \geq 1/2$. Then 
\begin{equation*}
\max_{\lambda \in [-1,1]} \left\vert P_t^{(\alpha,\beta)}(\lambda)\right\vert = {t+\max(\alpha,\beta) \choose t} \, .
\end{equation*}
\end{proposition}

\subsection{Polynomials orthogonal w.r.t.~$(1-\lambda^2)^{1/2}/\rho(\lambda)$, $\rho$ polynomial}
\label{ap:polynomials-orthogonal-szego}

In this section, we present how one can compute the recurrence relation for some orthogonal polynomials w.r.t.~a weight of the form $(1-\lambda^2)^{1/2}/\rho(\lambda)$, $\rho$ polynomial.

\begin{proposition}[{from \cite[Theorem 1.2.1]{szeg1939orthogonal}}]
	\label{prop:normalized-representation}
Let $\rho$ be a real polynomial of degree $l$ which is non-negative for $\lambda \in [-1,1]$. Then there exists a polynomial $h$ of degree $l$ such that for all real $\theta$, $\rho(\cos \theta) = \vert h(e^{i\theta})\vert^2$.
\end{proposition}

\begin{proposition}[{from \cite[Theorem 2.6]{szeg1939orthogonal}}]
	\label{prop:bernstein-szego}
Let $\rho$ be a real polynomial of degree $l$ taking positive values on the interval $[-1,1]$, and 
\begin{equation*}
\tau(\diff \lambda) = \frac{(1-\lambda^2)^{1/2}}{\rho(\lambda)}\diff\lambda \, .
\end{equation*}
 Let $h$ be a polynomial of degree $l$ such that $\rho(\cos \theta) = \vert h(e^{i\theta})\vert^2$ (see Proposition \ref{prop:normalized-representation}), and decompose $h(e^{i\theta}) = c(\theta) + is(\theta)$, $c(\theta)$ and $s(\theta)$ real. Then the polynomials 
 \begin{equation*}
 \pi_t(\cos \theta) = c(\theta) U_t(\cos\theta) - \frac{s(\theta)}{\sin\theta} T_{t+1}(\cos\theta) 
 \end{equation*}
 are orthogonal w.r.t.~$\tau$.
\end{proposition}

\subsection{Asymptotics for the Jacobi polynomials}
\label{ap:asymptotics-jacobi}

To prove the asymptotic performance guarantees of the polynomial iterations we build in this paper, we need the following asymptotic properties of the Jacobi polynomials. 

\begin{proposition}[{from \cite[Theorem 8.21.7]{szeg1939orthogonal}}]
	\label{prop:asymptotic-jacobi-out}
Let $\alpha, \beta > -1$, and $\lambda > 1$ a real number. Then there exists a positive constant $c = c(\alpha,\beta,\lambda)$ such that 
\begin{equation*}
P_t^{(\alpha,\beta)}(\lambda) \underset{t\to\infty}{\sim} \frac{c}{t^{1/2}} \left(\lambda+\sqrt{\lambda^2-1}\right)^t \, .
\end{equation*}
\end{proposition}

In the special case of the Chebyshev polynomials, we also have similar non-asymptotic bounds.

\begin{lemma}
	\label{lem:bounds-chebyshev}
For all $\lambda>1$, for all $t \geq 0$,
\begin{align}
&\frac{1}{2}\left(\lambda+\sqrt{\lambda^2-1}\right)^t \leq T_t(\lambda) \leq \left(\lambda+\sqrt{\lambda^2-1}\right)^t \, , \label{eq:aux-9}\\
&\left(\lambda+\sqrt{\lambda^2-1}\right)^t \leq U_t(\lambda) \leq (t+1)\left(\lambda+\sqrt{\lambda^2-1}\right)^t \, . \label{eq:aux-10}
\end{align}
\end{lemma}

\begin{proof}
We start by deriving a classic expression for the Chebyshev polynomials. The identities 
\begin{align*}
&T_t(\cos\theta) = \cos(t\theta) \, , &&U_t(\cos\theta) = \frac{\sin((t+1)\theta)}{\sin \theta} \, ,
\end{align*}
can be interpreted as 
\begin{align*}
&T_t\left(\frac{z+z^{-1}}{2}\right) = \frac{z^t+z^{-t}}{2} \, , &&U_t\left(\frac{z+z^{-1}}{2}\right) = \frac{z^{t+1}-z^{-(t+1)}}{z-z^{-1}} \, , &&\text{for $z=e^{i\theta}$.}
\end{align*}
The above equations are equalities of holomorphic functions on the unit circle, it implies that the identities must be true for all complex numbers $z \neq 0$; we use it here for real numbers $z$.

For $\lambda > 1$, write $\lambda = (z+z^{-1})/2$, $z > 1$. This is equivalent to $z = \lambda + \sqrt{\lambda^2-1}$. Then
\begin{equation*}
T_t(\lambda) = \frac{z^t+z^{-t}}{2} = \frac{1+z^{-2t}}{2}z^t = \frac{1+z^{-2t}}{2}\left( \lambda + \sqrt{\lambda^2-1}\right)^t \, .
\end{equation*}
As $z > 1$, 
\begin{equation*}
\frac{1}{2} \leq \frac{1+z^{-2t}}{2} \leq 1 \, .
\end{equation*}
This proves the inequalities \eqref{eq:aux-9}. Further,
\begin{equation*}
U_t(\lambda) = \frac{z^{t+1}-z^{-(t+1)}}{z-z^{-1}} = \frac{1-z^{-2t-2}}{1-z^{-2}}z^t = \frac{1-z^{-2t-2}}{1-z^{-2}}\left( \lambda + \sqrt{\lambda^2-1}\right)^t \, .
\end{equation*}
As $z > 1$, 
\begin{equation*}
1 \leq \frac{1-z^{-2t-2}}{1-z^{-2}} \leq t+1 \, .
\end{equation*}
This proves the inequalities \eqref{eq:aux-10}.
\end{proof}

\begin{proposition}[{from \cite[Theorem 7.32.2]{szeg1939orthogonal}}]
	\label{prop:asymptotic-jacobi-in}
Let $\alpha,\beta>-1$. There exists two constants $C_1, C_2 > 0$ such that,
\begin{equation*}
\left\vert P_t^{(\alpha,\beta)}(\cos\theta)\right\vert \leq \begin{cases}
C_1\theta^{-\alpha-1/2}t^{-1/2} &\text{if }1/t\leq\theta\leq\pi/2 \, ,  \\
C_2t^\alpha &\text{if }0\leq\theta\leq 1/t \, .
\end{cases}
\end{equation*}
\end{proposition}

\section{Some basic tools for the proofs}

\subsection{Comparing integrals using a domination of the cumulative distribution function}

\begin{lemma}
	\label{lem:domination-measure-increasing-functions}
	Let $\sigma, \tau$ be two positive measures on some interval $[a,b]$ such that for all $\lambda \in [a,b]$,
	\begin{equation}
	\label{eq:aux-11}
	\sigma([\lambda,b]) \leq \tau([\lambda,b]) \, .
	\end{equation}
	Then for all continuous non-decreasing functions $f:[a,b] \to \R_{\geq 0}$, 
	\begin{equation*}
	\int_{[a,b]} f(\lambda)\diff\sigma(\lambda) \leq \int_{[a,b]} f(\lambda)\diff\tau(\lambda) \, .
	\end{equation*}
\end{lemma}

\begin{proof}
	For any $u \in \R_{\geq 0}$, denote $\lambda^*(u) = \min\{\lambda\vert u \leq f(\lambda)\}$. 
	\begin{align*}
	\int_{[a,b]} f(\lambda)\diff\sigma(\lambda) &= \int_{[a,b]}\int_{\R_{\geq 0}} \bfone_{\{u \leq f(\lambda)\}} f(\lambda)\diff u \diff\sigma(\lambda) = \int_{\R_{\geq 0}} \left( \int_{[a,b]} \bfone_{\{u \leq f(\lambda)\}} \diff\sigma(\lambda) \right)\diff u \\
	&= \int_{\R_{\geq 0}} \sigma([\lambda^*(u),b])\diff u \, .
	\end{align*}
	The proof is finished using \eqref{eq:aux-11} and similar equalities for $\tau$.
\end{proof}.

\subsection{The gamma and beta function}

The gamma function $\Gamma$ and the beta function $B$ are defined as \cite[Section 5.2, Section 5.12]{olver2010nist}
\begin{align*}
&\Gamma(z) = \int_{0}^{\infty} e^{-u}u^{z-1}\diff u \, , \quad z \geq 0 \, ; &&B(a,b) =  \int_{0}^{1} s^{a-1}(1-s)^{b-1}\diff s = \frac{\Gamma(a)\Gamma(b)}{\Gamma(a+b)} \, , \quad a,b>0 \, .
\end{align*}
The asymptotic ratios of the gamma functions are given in \cite[Eq.~5.11.12]{olver2010nist}: for $c, d \in \R$,
\begin{equation*}
\frac{\Gamma(z+c)}{\Gamma(z+d)} \sim z^{c-d} \quad \text{as }z \to +\infty \, . 
\end{equation*}
This gives the asymptotic of the beta function
\begin{equation}
\label{eq:asymptotic-beta}
B(a,b) \sim \frac{\Gamma(b)}{a^b} \quad \text{as }a \to +\infty \, .
\end{equation}

\section{Proof of Proposition \ref{prop:opt_polynomial}}
\label{ap:opt_polynomial_proof}

Note first that as $\diff\tau(\lambda) = (1-\lambda)\diff\sigma(\lambda)$ is a measure on $[-1,1]$, if $\tilde{\pi}_0,\dots,\tilde{\pi}_{T-1}$ is a sequence of orthogonal polynomials w.r.t.~$\tau$, then the zeros of the polynomials $\tilde{\pi}_0,\dots,\tilde{\pi}_{T-1}$ are located in the interior of $[-1,1]$ (see Proposition \ref{prop:zeros}). In particular, $\tilde{\pi}_t(1) \neq 0$, $t <T$. Thus it is possible to build a family $\pi_0 = \tilde{\pi}_0/\tilde{\pi}_0(1),\dots,\pi_{T-1} = \tilde{\pi}_{T-1}/\tilde{\pi}_{T-1}(1)$ of orthogonal polynomials normalized to take value $1$ at $1$, as it is done in Proposition \ref{prop:opt_polynomial}.

The polynomial $\pi_t$ satisfies $\pi_t(1)=1$ and $\deg \pi_t = t$. We now consider some polynomial $Q_t$ also satisfying $Q_t(1)=1$ and $\deg Q_t =t$, and show that 
\begin{align*}
&\int \pi_t(\lambda)^2\diff\sigma(\lambda) \leq \int Q_t(\lambda)^2\diff\sigma(\lambda) \, , &&\text{i.e. }&&\langle\pi_t,\pi_t\rangle_\sigma \leq \langle Q_t, Q_t \rangle_\sigma \, . 
\end{align*}
The polynomial $Q_t-\pi_t$ vanishes at $1$ thus there exists a polynomial $R_{t-1}$ of degree at most $t-1$ such that $Q_t(\lambda) = \pi_t(\lambda) + (1-\lambda)R_{t-1}(\lambda)$. Then
\begin{equation*}
\langle Q_t, Q_t \rangle_\sigma = \langle\pi_t,\pi_t\rangle_\sigma + 2\langle\pi_t,(1-\lambda)R_{t-1}\rangle_\sigma + \langle (1-\lambda)R_{t-1},(1-\lambda)R_{t-1}\rangle_\sigma \, .
\end{equation*}
Note that $\langle\pi_t,(1-\lambda)R_{t-1}\rangle_\sigma = \langle\pi_t,R_{t-1}\rangle_\tau= 0$ because $\pi_t$ is orthogonal to all polynomials of degree smaller or equal to $t-1$ w.r.t.~$\langle.,.\rangle_\tau$. Moreover, 
\begin{equation*}
\langle (1-\lambda)R_{t-1},(1-\lambda)R_{t-1}\rangle_\sigma = \int (1-\lambda)^2R_{t-1}(\lambda)^2\diff\sigma(\lambda) \geq 0 \, .
\end{equation*}
Thus $\langle Q_t, Q_t \rangle_\sigma \geq \langle\pi_t,\pi_t\rangle_\sigma$. This shows that $\pi_t$ is a minimizer.

We now show that the minimizer $\pi_t$ is unique. There is equality $\langle Q_t, Q_t \rangle_\sigma = \langle\pi_t,\pi_t\rangle_\sigma$ if and only if $\int (1-\lambda)^2R_{t-1}(\lambda)^2\diff\sigma(\lambda) = 0$, i.e. $(1-\lambda)R_{t-1}$ vanishes on ${\rm Supp\,}\sigma$. But the cardinal of ${\rm Supp\,}\sigma$ is at least $T$ while $(1-\lambda)R_{t-1}$ is a polynomial of degree at most $t \leq T-1$. Thus the equality case is reached if and only if $R_{t-1} = 0$, i.e. $Q_t = \pi_t$.

\section{The parameter-free polynomial iteration}
\label{ap:parameter_free_polynomial_iteration}

In this section, we give the details of the implementation of the parameter-free polynomial iteration in a centralized setting. We explicit the computation of the optimal coefficients $a_t$, $b_t$ and $c_t$. It is used in the simulation of Figure \ref{fig:sim_2D_grid_log_all}. 

The parameter free polynomial iteration is Eq.~\eqref{eq:second-order-polynomial-gossip}, where the coefficients $a_t, b_t, c_t$, $t\geq 1$, are determined in Eq.~\eqref{eq:recurrence-coeffs-general}. The results are repeated here for convenience:
\begin{align*}
&x^0 = \xi \, , &&x^1 = a_0 W\xi + b_0 \xi \, , &&x^{t+1} = a_t W x^t + b_t x^t - c_t x^{t-1} \, , \\
&a_t = \frac{\langle\pi_{t+1},\pi_{t+1}\rangle_\tau}{\langle \lambda \pi_t, \pi_{t+1} \rangle_\tau} \, , 
&&b_t = -\frac{\langle \pi_{t+1}, \pi_{t+1} \rangle_\tau\langle \lambda\pi_t, \pi_t\rangle_\tau}{\langle \lambda \pi_t, \pi_{t+1} \rangle_\tau\langle\pi_t,\pi_t\rangle_\tau} \, , 
&&c_t = \frac{\langle \pi_{t+1}, \pi_{t+1} \rangle_\tau\langle\lambda\pi_t,\pi_{t-1}\rangle_\tau}{\langle \lambda \pi_t, \pi_{t+1} \rangle_\tau\langle\pi_{t-1},\pi_{t-1}\rangle_\tau} \, .
\end{align*} 
where $\tau = (1-\lambda)\sigma$, $\sigma$ is defined in \eqref{eq:error-as-spectral-finite}. Note that the scalar products that appear in the formulas for $a_t, b_t, c_t$ can be computed from the iterates $x^t = \pi_t(W)\xi$, $t\geq 0$. For instance,
\begin{align*}
\langle \lambda \pi_t, \pi_{t-1} \rangle_\tau &= \int \lambda\pi_t(\lambda)\pi_{t-1}(\lambda)(1-\lambda)\diff\sigma(\lambda) \\
&= \sum_{i=1}^n \langle\xi,u^i\rangle^2 \lambda_i \pi_t(\lambda_i) \pi_{t-1}(\lambda_i)(1-\lambda_i) \\
&= \langle W\pi_t(W)\xi , (I-W)\pi_{t-1}(W)\xi \rangle \\
&= \langle Wx^t, x^{t-1}-Wx^{t-1} \rangle \, .
\end{align*}
Note that the last line requires the computation of a scalar product $\langle ., . \rangle$ over $\R^V$, which means summing over $v \in V$. This is possible in simulations where we can centralize the information of the nodes $v \in V$. However in practical situation where the coordinates of $x^t$ are distributed among the nodes, such a computation requires many additional communication steps. This makes the parameter free polynomial iteration impractical. 

The computation of the other scalar products give
\begin{align*}
&b_t = -a_t\frac{\langle x^t-Wx^t, Wx^t\rangle}{\langle x^t, x^t-Wx^t\rangle} \, , &&c_t = a_t \frac{\langle x^t, x^{t-1}-Wx^{t-1}\rangle}{\langle x^{t-1}, x^{t-1}-Wx^{t-1}\rangle} \, ,
\end{align*}
and as $a_t + b_t - c_t = 1$ (that follows from $\pi_{t}(1)=1$ for all $t$), we get for $t \geq 1$,
\begin{align*}
&\tilde{b}_t = -\frac{\langle x^t-Wx^t, Wx^t\rangle}{\langle x^t, x^t-Wx^t\rangle} \, , \qquad \tilde{c}_t = \frac{\langle x^t, x^{t-1}-Wx^{t-1}\rangle}{\langle x^{t-1}, x^{t-1}-Wx^{t-1}\rangle} \, , \\
&x^{t+1} = \frac{1}{1+\tilde{b}_t-\tilde{c}_t}\left(Wx^t+\tilde{b}_tx^t-\tilde{c}_tx^{t-1}\right) \, .
\end{align*}
Similarly, one can compute that 
\begin{align*}
&x^1 = \frac{1}{1+\tilde{b}^0}\left(W\xi + \tilde{b}^0\xi\right) \, , &&\tilde{b}^0 = -\frac{\langle\xi-W\xi,W\xi\rangle}{\langle\xi,\xi-W\xi\rangle} \, ,
\end{align*}
which gives the initialization of the parameter-free polynomial iteration. 

\section{Proofs of Propositions \ref{prop:spectral-dim-Z^d} and \ref{prop:spectral-dimension-percolation}}
\label{ap:spectral-dimensions}

\subsection{Proof of Proposition \ref{prop:spectral-dim-Z^d}}
The return probability $p_t$ of the lazy random walk on $\Z^d$ is equivalent to $C/t^{d/2}$ for some constant $C$. It is, for instance, a consequence of the local central limit theorem for random walks on $\Z^d$ \cite[Theorem 2.1.1]{lawler2010random}. Thus the spectral dimension of $\Z^d$ is $d$.

\subsection{Proof of Proposition \ref{prop:spectral-dimension-percolation}}
The return probabilities of the random walk on the supercritical percolation cluster have rather been studied in continuous time. The continuous-time random walk is defined as follows: the random walk at $w$ waits at an exponential time of parameter $1$ before picking a site $w'$ out of the $2d$ neighboring sites uniformly randomly. If there is an edge in the percolation configuration between $w$ and $w'$, the random walk jumps to $w'$, otherwise it stays in $w$ and starts again. Denote $X_t$ the continuous-time random walk, and $\P_w$ the probability w.r.t.~this random walk when it is started from some vertex $w$.  

\begin{lemma}
	\label{lem:aux-1}
	There exists two constants $c = c(d,p), C = c(d,p)>0$ such that, a.s.~on the set $\{v\in G\}$, there exists a random time $t_0$ such that for $t\geq t_0$, 
	\begin{equation*}
	\frac{c}{t^{d/2}} \leq \P_v(X_t=v) \leq \frac{C}{t^{d/2}} \, .
	\end{equation*}
\end{lemma}

\begin{proof}
The upper bound is proved in \cite[Theorem 1.2]{mathieu2004isoperimetry}. As noted in \cite[Lemma 5.1]{biskup2011recent}, the lower bound can be proved using a central limit theorem on $X_t$; we repeat the argument here as our random walk differs slightly from theirs. As $X_t$ is reversible w.r.t.~the uniform measure on $G$,
\begin{equation*}
\P_v(X_{2t} = v) = \sum_{w \in G} \P_v(X_t = w)\P_w(X_t = v) = \sum_{w \in G} \P_v(X_t = w)^2 \, .
\end{equation*}
By the Cauchy-Schwarz inequality, 
\begin{align*}
\P_v(\Vert X_t-v \Vert_2 \leq \sqrt{t})^2 &= \left(\sum_{x \in G} \bfone_{\{\Vert x-v \Vert_2 \leq \sqrt{t}\}} \P_v( X_t =x)\right)^2 \\
&\leq \left\vert \{\Vert x \in G:  x-v \Vert_2 \leq \sqrt{t}\} \right\vert \left(\sum_{w\in G}\P_v(X_t = w)^2\right) \\
&\leq C_1t^{d/2}\P_v(X_{2t} = v) \, ,
\end{align*}
for some constant $C_1$. Now using \cite[Theorem 1.1(a)]{andres2013invariance}, there exists a deterministic variance $\sigma^2$ such that the law of $(X_t-v)/\sqrt{t}$ converges a.s.~on the event $\{v \in G\}$ to a centered Gaussian with variance $\sigma^2$. Thus there exists a deterministic constant $c_1 > 0$ and a random time $t_1$ such that for $t \geq t_1$, $\P_t(\Vert X_t-v \Vert_2 \leq \sqrt{t})^2 \geq c_1$. This finishes the proof of the lower bound. 
\end{proof}

We now finish the proof of the proposition using Lemma \ref{lem:aux-1}. If $\mu^t$ denotes the law of $X_t$,
\begin{align*}
&\frac{\diff}{\diff t} \E\left[\mu^t\right] = (W-I)\mu^t\, , &&\text{thus}\quad \mu^t = e^{t(W-I)}\mu^0 \, .
\end{align*}
Thus 
\begin{equation*}
\P_v(X_t =v) = \langle \delta_v, \mu^t \rangle = \langle \delta_v, e^{t(W-I)} \delta_v \rangle \stackrel{\text{(Definition \ref{def:spectral_measure_infinite_graph})}}{=} \int e^{t(\lambda -1)} \diff\sigma(\lambda) \, .
\end{equation*}

As a consequence, Lemma \ref{lem:aux-1} translates into bounds on the Laplace transform of $\sigma$: a.s.~on $\{v \in G\}$, for $t$ large enough,
	\begin{equation*}
	\frac{c}{t^{d/2}} \leq \int e^{t(\lambda -1)} \diff\sigma(\lambda) \leq \frac{C}{t^{d/2}} \, .
	\end{equation*}
Some bounds on the spectral density of $\sigma$ near $1$ easily follow (see \cite[Lemma 4.5]{muller2007spectral}): there exists constants $c',C'>0$ such that a.s.~on $\{v\in G \}$, for $E$ small enough,
\begin{equation*}
c'E^{d/2} \leq \sigma([1-E,1]) \leq C'E^{d/2} \, .
\end{equation*}
The proof is finished using Proposition \ref{prop:spectral-dimension-spectral-decay}. 

\section{Proof of Proposition \ref{prop:spectral-dimension-spectral-decay}}
\label{ap:proof-spectral-dimension-spectral-decay}

We start by assuming that $l = \lim_{E\to 0} \ln \sigma([1-E,1])/\ln E$
exists and is finite. We show that $d_s$ exists and that $l = d_s/2$. To this end, we define
\begin{align*}
&\underline{d}_s = -2\limsup_{t\to\infty} \frac{\ln p_t}{\ln t} \, , &&\bar{d}_s = -2\liminf_{t\to\infty} \frac{\ln p_t}{\ln t} \, ,
\end{align*}
where $p_t$ is defined as in Definition \ref{def:spectral-dimension}. Note that 
\begin{equation}
\label{eq:aux-8}
p_t = \left\langle e_v, \left(\frac{I+W}{2}\right)^te_v\right\rangle \stackrel{\text{(Definition \ref{def:spectral_measure_infinite_graph})}}{=} \int \left(\frac{1+\lambda}{2}\right)^t\diff\sigma(\lambda) \, .
\end{equation}
\smallskip

{\bfseries Proof that $\bar{d_s}/2\leq l$.} Consider $l_+ > l$. Then there exists constants $c_1, c_2>0$ such that for all $E \in [0,2]$, 
\begin{equation*}
\sigma([1-E,1])\geq c_1 E^{l_+} = c_2 \sigma^{(l_+-1,0)}([1-E,1]) \, ,
\end{equation*}
where $\sigma^{(l_+-1,0)}(\diff\lambda) = (1-\lambda)^{l_+-1}\diff\lambda$. Then
\begin{align*}
p_t &\stackrel{\eqref{eq:aux-8}}{=} \int_{[-1,1]} \left(\frac{1+\lambda}{2}\right)^{t}\diff\sigma(\lambda) \stackrel{\text{(Lemma \ref{lem:domination-measure-increasing-functions})}}{\geq} c_2\int_{-1}^1 \left(\frac{1+\lambda}{2}\right)^{t-1}(1-\lambda)^{l_+-1}\diff\lambda \\
&\hspace{-0.5cm}\stackrel{(u = (1+\lambda)/2)}{=} c_3 \int_{0}^{1} u^{t}(1-u)^{l_+-1}\diff u = c_3B\left(t+1,l_+\right) \stackrel{ \eqref{eq:asymptotic-beta}}{\underset{t\to\infty}{\sim}} \frac{c_4}{t^{l_+}} \, ,
\end{align*}
for some constant $c_3, c_4 > 0$. Thus 
\begin{align*}
&\liminf_{t\to\infty} \frac{\ln p_t}{\ln t} \geq -l_+ \, ,  &&\text{i.e.}\quad \frac{\bar{d}_s}{2} \leq l_+ \, .
\end{align*}
This being true for all $l_+ > l$, this proves $\bar{d}_s/2 \leq l$.
\smallskip

{\bfseries Proof that $\underline{d}_s/2 \geq l$.} Consider $l_- < l$. Then there exists constants $C_1, C_2$ such that for all $E \in [0,2]$,
\begin{equation*}
\sigma([1-E,1])\leq C_1E^{l_-} = C_2\sigma^{(l_--1,0)}([1-E,1]) \, ,
\end{equation*}
where $\sigma^{(l_--1,0)}(\diff\lambda) = (1-\lambda)^{l_--1}\diff\lambda$. Then 
\begin{align*}
p_t &\stackrel{\eqref{eq:aux-8}}{=} \int_{[-1,1]} \left(\frac{1+\lambda}{2}\right)^{t}\diff\sigma(\lambda) \stackrel{\text{(Lemma \ref{lem:domination-measure-increasing-functions})}}{\leq} C_2\int_{-1}^1 \left(\frac{1+\lambda}{2}\right)^{t}(1-\lambda)^{l_--1}\diff\lambda \\
&\hspace{-0.5cm}\stackrel{(u = (1+\lambda)/2)}{=} C_3 \int_{0}^{1} u^{t}(1-u)^{l_--1}\diff u = C_3B\left(t+1,l_-\right) \underset{t\to\infty}{\sim} \frac{C_4}{t^{l_-}} \, ,
\end{align*}
for some constants $C_3, C_4$. Thus 
\begin{align*}
&\limsup_{t\to\infty} \frac{\ln p_t}{\ln t} \leq -l_- \, ,  &&\text{i.e.}\quad \frac{\underline{d}_s}{2} \geq l_- \, .
\end{align*}
This being true for all $l_- < l$, this proves $\underline{d}_s/2 \geq l$.
\smallskip

Finally, we have proven $l \leq \underline{d}_s/2 \leq \bar{d}_s/2 \leq l$. Thus the limit $d_s = -2\lim_{t\to\infty} \ln p_t/\ln t$ exists and is equal to $2l$. 

\medskip
Conversely, we assume now that $d_s$ exists and is finite. We show that $l = \lim_{E\to 0} \ln \sigma([1-E,1])/\ln E$ exists and that $l = d_s/2$. To this end, we define
\begin{align*}
&\underline{l} = \liminf_{E\to 0} \frac{\ln \sigma([1-E,1])}{\ln E} \, , &&\bar{l} = \limsup_{E\to 0} \frac{\ln \sigma([1-E,1])}{\ln E} \, .
\end{align*}
\smallskip

{\bfseries Proof that $\underline{l} \geq d_s/2$.}
For any $t \in \N$, we have $\bfone_{\{\lambda \geq 1-E\}} \leq (1-E/2)^{-t}\left((1+\lambda)/2\right)^{t} $, thus by integrating against $\diff\sigma(\lambda)$, 
\begin{align*}
&\sigma([1-E,1]) \leq \left(1-\frac{E}{2}\right)^{-t} \int \left(\frac{1+\lambda}{2}\right)^{t} \sigma(\diff\lambda) \, ,\\
&\frac{\ln \sigma([1-E,1])}{\ln E} \geq \frac{\ln \int \left(\frac{1+\lambda}{2}\right)^{t} \sigma(\diff\lambda)}{\ln t}\frac{\ln t}{\ln E} - \frac{t \ln\left(1-\frac{E}{2}\right)}{\ln E} \, .
\end{align*}
We choose $t(E) = \lfloor E^{-1}\rfloor$. Then we get
\begin{equation*}
\underline{l} = \liminf_{E\to 0} \frac{\ln \sigma([1-E,1])}{\ln E} \geq -\frac{d_s}{2} (-1) - 0 = \frac{d_s}{2}
\end{equation*}

{\bfseries Proof that $\bar{l} \leq d_s/2$.}
For any $t \in \N$, we have $((1+\lambda)/2)^t - (1-E/2)^t \leq \bfone_{\{\lambda\geq 1-E\}}$, thus by integrating against $\diff\sigma(\lambda)$, 
\begin{align*}
&\int\left(\frac{1+\lambda}{2}\right)^t\diff\sigma(\lambda) - \left(1-\frac{E}{2}\right)^t \leq \sigma([1-E,1]) \, .
\end{align*}
Let $d>d_s$. There exists a constant $c>0$ such that $\int((1+\lambda)/2)^t\diff\sigma(\lambda)\geq c/t^{d/2}$. Then 
\begin{equation*}
\ln\left(\frac{c}{t^{d/2}}-\left(1-\frac{E}{2}\right)^t\right) \leq \ln\sigma([1-E,1]) \, .
\end{equation*}
Let $\alpha > 1$. We choose $t(E) = \lceil E^{-\alpha}\rceil$. Then
\begin{equation*}
\left(1-\frac{E}{2}\right)^{t(E)} = \exp\left(t(E)\ln\left(1-\frac{E}{2}\right)\right) \leq \exp\left(-\frac{t(E)E}{2}\right) \leq \exp\left(-\frac{1}{2}E^{1-\alpha}\right)
\end{equation*}
decreases super-polynomially fast as $E\to 0$. Moreover
\begin{equation*}
\frac{c}{t(E)^{d/2}} \underset{E\to 0}{\sim} cE^{\alpha d/2} \, .
\end{equation*}
Finally, 
\begin{equation*}
\bar{l} = \limsup_{E\to 0} \frac{\ln \sigma([1-E,1])}{\ln E} \leq \frac{\alpha d}{2} \, .
\end{equation*}
As this is true for all $\alpha > 1, d > d_s$, we have $\bar{l} \leq d_s/2$.
\smallskip 

Finally, we have proven that $d_s/2 \leq \underline{l} \leq \bar{l}\leq d_s/2$. Then the limit $l = \lim_{E\to 0} \ln \sigma([1-E,1])/\ln E$ exists and $l = d_s/2$. 

\section{Computation of the recursion coefficients of some orthogonal polynomials}
\label{ap:computation_recursion_coeffs}

\subsection{A rescaling lemma for orthogonal polynomials}

We start with a lemma giving the change in the recursion coefficients of orthogonal polynomials when the underlying measure undergoes an affine transformation. It is used in the next subsections. 

\begin{lemma}
	\label{lem:rescaled_orthogonal_pol}
	Let $\sigma$ be a measure on $\R$, $\pi_0, \dots, \pi_{T-1}$ a sequence of orthogonal polynomials w.r.t.~$\sigma$ and 
	\begin{equation}
	\label{eq:reccurence_relation_orth_pol}
	\pi_{t+1}(\lambda) = (a_t\lambda+b_t) \pi_t(\lambda) - c_t \pi_{t-1}(\lambda) \, , \qquad t \geq 1 \, , 
	\end{equation}
	their recurrence formula (see Definition \ref{def:orthogonal_polynomials} and Theorem \ref{prop:recurrence_relation}). 
	
	Let $\varphi:\lambda \mapsto \alpha \lambda + \beta$, $\alpha\neq 0$ be a linear function and $\tilde{\sigma}$ be the image measure of $\sigma$ by $\varphi$ (which means that for all measurable set $A$, $\tilde{\sigma}(A) = \sigma(\varphi^{-1}(A))$). Then a sequence of orthogonal polynomials w.r.t.~$\tilde{\sigma}$ is given by the formula
	\begin{equation*}
	\tilde{\pi}_t(\tilde{\lambda}) := \pi_t\left(\varphi^{-1}(\tilde{\lambda})\right) =  \pi_t\left(\frac{\tilde{\lambda}-\beta}{\alpha}\right) \, .
	\end{equation*}
	These polynomials follow the recursion formula
	\begin{align*}
	& \tilde{\pi}_{t+1}(\tilde{\lambda}) = (\tilde{a}_{t}\tilde{\lambda}+\tilde{b}_t) \tilde{\pi}_t(\tilde{\lambda}) - \tilde{c}_{t} \tilde{\pi}_{t-1}(\tilde{\lambda}) \, , \\
	&\tilde{a}_t = \frac{a_t}{\alpha} \, , \qquad \tilde{b}_t = b_t - \frac{a_t \beta}{\alpha} \, , \qquad \tilde{c}_t = c_t \, .
	\end{align*}
\end{lemma}

\begin{proof}
	By change of variable, 
	\begin{align*}
	\int\tilde{\pi}_t(\tilde{\lambda})\tilde{\pi}_s(\tilde{\lambda})\diff \tilde{\sigma}(\tilde{\lambda}) &= \int \pi_t\left(\varphi^{-1}(\tilde{\lambda})\right) \pi_s\left(\varphi^{-1}(\tilde{\lambda})\right) \diff\tilde{\sigma}(\tilde{\lambda}) \\
	&= \int \pi_t\left(\varphi^{-1}(\varphi(\lambda))\right) \pi_s\left(\varphi^{-1}(\varphi(\lambda))\right) \diff\sigma(\lambda) \\
	&= \int \pi_t(\lambda) \pi_s(\lambda) \diff\sigma(\lambda) = \bfone_{\{s=t\}} \, ,
	\end{align*}
	and $\deg \tilde{\pi}_t = t$ thus $\tilde{\pi}_0, \dots, \tilde{\pi}_{T-1}$ are orthogonal polynomials w.r.t.~$\tilde{\sigma}$. The recurrence relation for $\tilde{\pi}_t$ follows by evaluating the recurrence relation \eqref{eq:reccurence_relation_orth_pol} for $\pi_t$ in $(\tilde{\lambda}-\beta)/\alpha$. 
\end{proof}

\subsection{Jacobi polynomials}
\label{ap:Jacobi}

Let $\alpha, \beta>-1$. In this section, we derive, using the recurrence formula for the Jacobi polynomial $P_t^{(\alpha,\beta)}$ of Proposition \ref{prop:recurrence-jacobi}, a similar recurrence relation for the polynomials $\pi_t^{(\alpha,\beta)}$ orthogonal w.r.t.~the Jacobi measure $\sigma^{(\alpha,\beta)}$, but normalized such that $\pi_t^{(\alpha,\beta)}(1) = 1$.

Substituting $P_t^{(\alpha,\beta)} = {t + \alpha \choose t}\pi_t^{(\alpha,\beta)}$ in the recurrence relation of Proposition \ref{prop:recurrence-jacobi}, we get 
\begin{align*}
&2(t+1)(t+1+\alpha+\beta)(2t+\alpha+\beta){t+1 + \alpha \choose t+1}\pi_{t+1}^{(\alpha,\beta)}(\lambda) \\
&\qquad= (2t+\alpha+\beta+1)[(2t+\alpha+\beta+2)(2t+\alpha+\beta)\lambda+\alpha^2-\beta^2]{t + \alpha \choose t}\pi_t^{(\alpha,\beta)}(\lambda)  \\
&\qquad\quad -2(t+\alpha)(t+\beta)(2t+\alpha+\beta+2){t-1 + \alpha \choose t-1}\pi_{t-1}^{(\alpha,\beta)}(\lambda) \, .
\end{align*}
Using that $(t+1){t+1 + \alpha \choose t+1} = (t+1+\alpha){t + \alpha \choose t}$ and $t{t+ \alpha \choose t} = (t+\alpha){t-1 + \alpha \choose t-1}$, we can divide the above equation by ${t + \alpha \choose t}$. We get
\begin{align*}
&2(t+1+\alpha+\beta)(2t+\alpha+\beta)(t+1+\alpha)\pi_{t+1}^{(\alpha,\beta)}(\lambda) \\
&\qquad= (2t+\alpha+\beta+1)[(2t+\alpha+\beta+2)(2t+\alpha+\beta)\lambda+(\alpha+\beta)(\alpha-\beta)]\pi_t^{(\alpha,\beta)}(\lambda)  \\
&\qquad\quad -2t(t+\beta)(2t+\alpha+\beta+2)\pi_{t-1}^{(\alpha,\beta)}(\lambda) \, .
\end{align*}
Summing up, we obtain the recursion formula
\begin{align*}
&\pi_0(\lambda) = 1 \, , \qquad \pi_1(\lambda) = a_0^{(\alpha,\beta)}\lambda + b_0^{(\alpha,\beta)} \, , \\
&\pi_{t+1}^{(\alpha,\beta)}(\lambda) = \left(a_t^{(\alpha,\beta)}\lambda + b_t^{(\alpha,\beta)}\right) \pi_t^{(\alpha,\beta)}(\lambda) - c_t^{(\alpha,\beta)} \pi_{t-1}^{(\alpha,\beta)}(\lambda) \, , 
\end{align*}
with the recursion coefficients
\begin{equation}
\label{eq:Jacobi_weights}
\begin{aligned}
&a_0^{(\alpha,\beta)} = \frac{\alpha+\beta+2}{2(1+\alpha)} \, , \qquad b_0^{(\alpha,\beta)} = \frac{\alpha-\beta}{2(1+\alpha)} \, , \\
&a_t^{(\alpha,\beta)} = \frac{(2t+\alpha+\beta+1)(2t+\alpha+\beta+2)}{2(t+1+\alpha+\beta)(t+1+\alpha)} \, , \\
&b_t^{(\alpha,\beta)} = \frac{(2t+\alpha+\beta+1)(\alpha+\beta)(\alpha-\beta)}{2(t+1+\alpha+\beta)(t+1+\alpha)(2t+\alpha+\beta)} \, ,\\ 
&c_t^{(\alpha,\beta)} = \frac{t(t+\beta)(2t+\alpha+\beta+2)}{(t+1+\alpha+\beta)(2t+\alpha+\beta)(t+1+\alpha)} \, .
\end{aligned}
\end{equation}

\subsection{Rescaled Jacobi polynomials.}
\label{ap:recurrence-relation-rescaled-jacobi}

Let $\alpha, \beta > -1$. In this section, we determine a recursion formula for the orthogonal polynomials $\pi_t^{(\alpha,\beta,\gamma)}$ w.r.t.~ the rescaled Jacobi measure 
\begin{equation*}
\diff\sigma^{(\alpha,\beta,\gamma)}(\lambda) = ((1-\gamma)-\lambda)^\alpha(1+\lambda)^\beta\bfone_{\{\lambda \in (-1,1-\gamma)\}}\diff\lambda\, , 
\end{equation*}
The polynomials $\pi^{(\alpha,\beta,\gamma)}_t$ are normalized such that $\pi^{(\alpha,\beta,\gamma)}_t(1)=1$.

Note that, up to a rescaling, $\diff\sigma^{(\alpha,\beta,\gamma)}$ is the image measure of the Jacobi measure $\diff\sigma^{(\alpha,\beta)}$ (defined in \eqref{eq:jacobi-measure}) by the linear function $\varphi_\gamma(\lambda) = (1-\gamma/2)\lambda-\gamma/2$. Thus Lemma \ref{lem:rescaled_orthogonal_pol} gives a family of orthogonal polynomials $P^{(\alpha,\beta,\gamma)}_t$ w.r.t.~$\diff\sigma^{(\alpha,\beta,\gamma)}$ and their recursion formula:
\begin{align*}
&P^{(\alpha,\beta,\gamma)}_t(\tilde{\lambda}) = \pi^{(\alpha,\beta)}_t\left(\varphi_\gamma^{-1}(\tilde{\lambda})\right) \, , \\
&P^{(\alpha,\beta,\gamma)}_{t+1}(\tilde{\lambda}) = \left(a_t^{(\alpha,\beta,\gamma)}\tilde{\lambda} + b_t^{(\alpha,\beta,\gamma)}\right)P^{(\alpha,\beta,\gamma)}_{t}(\tilde{\lambda})-c_t^{(\alpha,\beta,\gamma)}P^{(\alpha,\beta,\gamma)}_{t-1}(\tilde{\lambda}) \, , \\
&a_t^{(\alpha,\beta,\gamma)} = a_t^{(\alpha,\beta)}\left(1-\frac{\gamma}{2}\right)^{-1} \, , \qquad b_t^{(\alpha,\beta,\gamma)} = b_t^{(\alpha,\beta)} + \frac{\gamma}{2}a_t^{(\alpha,\beta)}\left(1-\frac{\gamma}{2}\right)^{-1} \, , \qquad c_t^{(\alpha,\beta,\gamma)} = c_t^{(\alpha, \beta)} \, .
\end{align*}
However, the polynomials $P_t^{(\alpha,\beta,\gamma)}$ are not normalized such that $P_t^{(\alpha,\beta,\gamma)} = 1$. Indeed, $P_t^{(\alpha,\beta,\gamma)} = \pi_t^{(\alpha,\beta)}\left(\left(1-\gamma/2\right)^{-1}\left(1+\gamma/2\right)\right)$. It is difficult to deduce the recurrence relation for $\pi_t^{(\alpha,\beta,\gamma)} = P_t^{(\alpha,\beta,\gamma)}/P_t^{(\alpha,\beta,\gamma)}(1)$ from the recurrence relation for $P_t^{(\alpha,\beta,\gamma)}$. One can circumvent this difficulty by using that the normalization $P_t^{(\alpha,\beta,\gamma)}(1)$ also follows the recurrence relation 
\begin{equation*}
P^{(\alpha,\beta,\gamma)}_{t+1}(1) = \left(a_t^{(\alpha,\beta,\gamma)} + b_t^{(\alpha,\beta,\gamma)}\right)P^{(\alpha,\beta,\gamma)}_{t}(1)-c_t^{(\alpha,\beta,\gamma)}P^{(\alpha,\beta,\gamma)}_{t-1}(1) \, .
\end{equation*}
Summing things up, we get 
\begin{equation}
\label{eq:rescaled-Jacobi}
\begin{aligned}
&\pi_t^{(\alpha,\beta,\gamma)}(\lambda) = \frac{P_t^{(\alpha,\beta,\gamma)}(\lambda)}{P_t^{(\alpha,\beta,\gamma)}(1)} \, , \\
&P_0^{(\alpha,\beta,\gamma)}(\lambda) = 1 \, , \qquad P_0^{(\alpha,\beta,\gamma)}(1) = 1 \, , \\
&P_1^{(\alpha,\beta,\gamma)}(\lambda) = a_0^{(\alpha,\beta,\gamma)}\lambda + b_0^{(\alpha,\beta,\gamma)} \, , \qquad P_1^{(\alpha,\beta,\gamma)}(1) = a_0^{(\alpha,\beta,\gamma)} + b_0^{(\alpha,\beta,\gamma)} \, , \\
&P^{(\alpha,\beta,\gamma)}_{t+1}(\lambda) = \left(a_t^{(\alpha,\beta,\gamma)}\tilde{\lambda} + b_t^{(\alpha,\beta,\gamma)}\right)P^{(\alpha,\beta,\gamma)}_{t}(\lambda)-c_t^{(\alpha,\beta,\gamma)}P^{(\alpha,\beta,\gamma)}_{t-1}(\lambda) \, , \qquad t \geq 1 \, ,\\
&P^{(\alpha,\beta,\gamma)}_{t+1}(1) = \left(a_t^{(\alpha,\beta,\gamma)} + b_t^{(\alpha,\beta,\gamma)}\right)P^{(\alpha,\beta,\gamma)}_{t}(1)-c_t^{(\alpha,\beta,\gamma)}P^{(\alpha,\beta,\gamma)}_{t-1}(1) \, , \qquad t \geq 1 \, , \\
&a_t^{(\alpha,\beta,\gamma)} = a_t^{(\alpha,\beta)}\left(1-\frac{\gamma}{2}\right)^{-1} \, , \qquad b_t^{(\alpha,\beta,\gamma)} = b_t^{(\alpha,\beta)} + \frac{\gamma}{2}a_t^{(\alpha,\beta)}\left(1-\frac{\gamma}{2}\right)^{-1} \, , \qquad t \geq 0 \, , \\
&c_t^{(\alpha,\beta,\gamma)} = c_t^{(\alpha, \beta)} \, , \qquad t \geq 1 \, . 
\end{aligned}
\end{equation}
These equations give a practical way to compute the polynomials $\pi_t^{(\alpha,\beta,\gamma)}$ because all recursion coefficients can be computed explicitly using the formulas \eqref{eq:Jacobi_weights}.

\subsection{Polynomials orthogonal to the modified Kesten-McKay measure}
\label{ap:Kesten-McKay}
In this section, we determine the recurrence formula for the orthogonal polynomials $\pi_t$ w.r.t.~the modified Kesten-McKay measure
\begin{equation*}
(1-\lambda)\sigma(\T_d)(\diff\lambda) = \frac{d}{2\pi (1 + \lambda)}  {\left({\frac{4(d-1)}{d^2}-\lambda^2}\right)^{1/2}} \bfone_{\left[-2\sqrt{d-1}/d,2\sqrt{d-1}/d\right]}(\lambda)\diff\lambda \, .
\end{equation*}
The polynomials $\pi_t$ are normalized such that $\pi_t(1)=1$.

The measure $\diff\sigma(\T_d)$ is, up to a rescaling factor, the image measure of
\begin{equation}
\label{eq:Kesten-McKay-rescaled}
\diff\tilde{\sigma}(\lambda) = \frac{(1-\lambda^2)^{1/2}}{d+2\sqrt{d-1}\lambda}\bfone_{[-1,1]}(\lambda)\diff\lambda
\end{equation}
by the linear map $\varphi: \lambda \mapsto 2\sqrt{d-1}\lambda/d$. We thus compute a family of orthogonal polynomials w.r.t.~$\tilde{\sigma}$ and then use Lemma \ref{lem:rescaled_orthogonal_pol}. 

The orthogonal polynomials w.r.t.~$\tilde{\sigma}$ are given by Proposition \ref{prop:bernstein-szego}. Following the cited theorem, we define $\rho(\lambda) = d + 2\sqrt{d-1}\lambda$ and
\begin{align*}
&\rho(\cos \theta) = 2\sqrt{d-1}\cos \theta +d = |h(e^{i\theta})|^2 \, , \\
&h(e^{i\theta}) = \sqrt{d-1} + e^{i\theta} = \underbrace{\sqrt{d-1}+\cos\theta}_{:= c(\theta)} + i\underbrace{\sin\theta}_{:= s(\theta)} \, .
\end{align*}
Then we have the following family $\tilde{p}_t$ of orthogonal polynomials w.r.t.~$\tilde{\sigma}$. 
\begin{align*}
&\tilde{p}_t(\cos\theta) = c(\theta)U_t(\cos\theta) - \frac{s(\theta)}{\sin\theta}T_{t+1}(\cos\theta) \, , \\
&\tilde{p}_t(\lambda) = (\sqrt{d-1}+\lambda)U_t(\lambda)-T_{t+1}(\lambda) \, ,
\end{align*}
where $T_t$ and $U_t$ denote the $t$-th Chebyshev polynomial of the first kind and the second kind respectively. As the Chebyshev polynomials $T_t$ and $U_t$ both satisfy the same recurrence relation 
\begin{align*}
T_{t+1}(\lambda) &= 2\lambda T_t(\lambda) - T_{t-1}(\lambda) \, , \\
U_{t+1}(\lambda) &= 2\lambda U_t(\lambda) - U_{t-1}(\lambda) \, , \quad t\geq 1 \, ,
\end{align*}
the same relation follows for $\tilde{p}_t$:
\begin{equation*}
\tilde{p}_{t+1}(\lambda) = 2\lambda \tilde{p}_t(\lambda) - \tilde{p}_{t-1}(\lambda) \, , \quad t\geq 1 \, ,
\end{equation*}
with initial condition $\tilde{p}_0(\lambda) = \sqrt{d-1}$ and $\tilde{p}_1(\lambda) = 2\sqrt{d-1}\lambda+1$.

Lemma \ref{lem:rescaled_orthogonal_pol} gives the rescaled orthogonal polynomials $p_t(\lambda) = \tilde{p}_t\left(\varphi^{-1}(\lambda)\right)$ w.r.t.~$\diff\sigma(\T_d)$:
\begin{align}
\label{eq:rescaled-McKay}
&p_0(\lambda) = \sqrt{d-1} \, , 
&&p_1(\lambda) = d\lambda+1 \, , 
&&p_{t+1}(\lambda) = \frac{d}{\sqrt{d-1}}\lambda p_t(\lambda)-p_{t-1}(\lambda) \, , \quad t \geq 1 \, .
\end{align}
As $\pi_t(\lambda) = p_t(\lambda)/p_t(1)$, it now remains to compute $p_t(1)$. The sequence $p_t(1)$, $t\geq 1$ satisfies a second-order recurrence relation with fixed coefficients, it thus can be solved explicitly
\begin{equation*}
p_t(1) = \frac{1}{d-1}\left(d(d-1)^{(t+1)/2}-2(d-1)^{(1-t)/2}\right) \, .
\end{equation*}
By substituting $p_t(\lambda) = p_t(1)\pi_t(\lambda)$ in \eqref{eq:rescaled-McKay}, one obtains
\begin{align*}
&\pi_0(\lambda) = 1 \, , \qquad \pi_1(\lambda) = a_0\lambda+b_0 \, , \qquad \pi_{t+1}(\lambda) = a_t\lambda \pi_t(\lambda)-c_t\pi_{t-1}(\lambda) \, , \quad t \geq 1 \, , \\
&a_0 = \frac{d}{d+1} \, , \qquad b_0 =  \frac{1}{d+1} \, , \qquad a_t = \frac{\frac{d}{d-1}-2(d-1)^{-(t+1)}}{1-\frac{2}{d}(d-1)^{-(t+1)}} \, , \qquad c_t = \frac{\frac{1}{d-1}-\frac{2}{d}(d-1)^{-t}}{1-\frac{2}{d}(d-1)^{-(t+1)}} \, , \quad t \geq 1 \, .
\end{align*}

\section{Proof of Theorem \ref{thm:rate-decrease-Jacobi}}
\label{ap:proof-thm-rate-decrease-Jacobi}

The proof is divided in four subsections. Appendix \ref{ap:proof-thm-rate-decrease-Jacobi_preliminaries} develops tools that we use both in the proof of the theorem. We then prove the theorem in Sections \ref{ap:proof-thm-rate-decrease-Jacobi_simple}, \ref{ap:proof-thm-rate-decrease-Jacobi_shift-register}and \ref{ap:proof-thm-rate-decrease-Jacobi_jacobi}. Finally, in Appendix \ref{ap:tuning-Jacobi}, we discuss the choice of the parameters of the Jacobi polynomials in the Jacobi polynomial iteration. In all this appendix, we denote $\sigma = \sigma(G,W,v)$ the spectral measure of $G$.

\subsection{Preliminaries}
\label{ap:proof-thm-rate-decrease-Jacobi_preliminaries}

The first lemma relates the MSE of the estimator $x^t_v$ to the spectral measure.

\begin{lemma}
	\label{lem:proof-technical-1}
Write $x^t = P_t(W)\xi$ using the polynomial gossip point of view. Then 
\begin{equation*}
\E[(x^t_v-\mu)^2] = (\var \nu)\left\Vert P_t(W)e_v \right\Vert_{\ell^2(V)}^2 = (\var \nu)\int P_t(\lambda)^2\diff\sigma(\lambda) \, .
\end{equation*}
\end{lemma}

\begin{proof}
As $P_t(1) = 1$, we have 
\begin{equation*}
\E[x^t] = \E[P_t(W)\xi] = P_t(W)\E[\xi] = P_t(W)\mu\bfone = P_t(1)\mu\bfone = \mu\bfone \, .
\end{equation*}
In words, the estimator $x^t_v$ is unbiased. Thus 
\begin{equation*}
\E[(x^t_v-\mu)^2] = \var x^t_v = \var \left\langle P_t(W)\xi, e_v \right\rangle_{\ell^2(V)} = \var \left\langle \xi, P_t(W)e_v \right\rangle_{\ell^2(V)} = (\var \nu)\left\Vert P_t(W)e_v \right\Vert_{\ell^2(V)}^2  \, , 
\end{equation*}
using that $W$ is symmetric and that the $\xi_w$, $w \in V$ are i.i.d.~random variables. Then
\begin{equation*}
\E[(x^t_v-\mu)^2] = (\var \nu)\left\langle P_t(W)e_v , P_t(W)e_v \right\rangle_{\ell^2(V)} = (\var \nu)\left\langle e_v , P_t(W)^2e_v \right\rangle_{\ell^2(V)}  \, .
\end{equation*}
The proof is finished using the Definition \ref{def:spectral_measure_infinite_graph} of the spectral measure. 
\end{proof}

In the statement of Theorem \ref{thm:rate-decrease-Jacobi}, we have stated results in terms of the spectral dimension $d_s = 2\lim_{t\to\infty} \ln \sigma([1-E,1])/\ln E$. In the proof here, we will be more precise. We show how the results of Theorem \ref{thm:rate-decrease-Jacobi} actually depend on different definitions of the dimension. 

\begin{definition}
	Let $\tau$ be a probability measure on $[-1,1]$. We define
	\begin{enumerate}
		\item the right upper dimension $\overline{\dim}_\rightarrow\tau \in [0,\infty]$ of the measure $\tau$ as 
		\begin{equation*}
		\overline{\dim}_\rightarrow\tau = 2 \, \limsup_{E\to 0} \frac{\ln\tau([1-E,1])}{\ln E} \, ,
		\end{equation*}
		\item the right lower dimension $\underline{\dim}_\rightarrow\tau \in [0,\infty]$ of the measure $\tau$ as 
		\begin{equation*}
		\underline{\dim}_\rightarrow\tau = 2 \, \liminf_{E\to 0} \frac{\ln\tau([1-E,1])}{\ln E} \, ,
		\end{equation*}
		\item the left upper dimension $\overline{\dim}_\leftarrow\tau\in [0,\infty]$ of the measure $\tau$ as 
		\begin{equation*}
		\overline{\dim}_\leftarrow\tau = 2 \, \limsup_{E\to 0} \frac{\ln\tau([-1,-1+E])}{\ln E} \, .
		\end{equation*}
		\item the left lower dimension $\underline{\dim}_\leftarrow\tau\in [0,\infty]$ of the measure $\tau$ as 
		\begin{equation*}
		\underline{\dim}_\leftarrow\tau = 2 \, \liminf_{E\to 0} \frac{\ln\tau([-1,-1+E])}{\ln E} \, .
		\end{equation*}
	\end{enumerate}
\end{definition}

\subsection{Proof of Theorem \ref{thm:rate-decrease-Jacobi}: simple gossip}
\label{ap:proof-thm-rate-decrease-Jacobi_simple}

In the case of simple gossip, $P_t(\lambda) = \lambda^t$.

\begin{proposition}
	\label{prop:aux-1}
Let $\tau$ be a probability measure on $[-1,1]$. Then 
\begin{equation*}
\liminf_{t\to\infty} \frac{\int \lambda^{2t}\diff\tau(\lambda)}{\ln t} \geq -\frac{\min\left(\overline{\dim}_\rightarrow \tau, \overline{\dim}_\leftarrow \tau\right)}{2}
\end{equation*}
\end{proposition}

\begin{proof}
Let $d > \overline{\dim}_\rightarrow \tau$. As $\overline{\dim}_\rightarrow \tau = 2\limsup_{E\to0}\ln \sigma([1-E,1])/\ln E$, there exists constants $c_1, c_2 > 0$ such that for all $E \in [0,2]$, 
\begin{equation}
\label{eq:aux-12}
\tau([1-E,1]) \geq c_1E^{d/2} = c_2\sigma^{(d/2-1,0)}([1-E,1])
\end{equation}
where $\sigma^{(d/2-1,0)}(\diff\lambda) = (1 - \lambda)^{d/2-1}\diff\lambda$. Then using jointly Lemma \ref{lem:domination-measure-increasing-functions} and Eq.~\eqref{eq:aux-12},  
\begin{align*}
\int \lambda^{2t}\diff\tau(\lambda) &\geq \int_{[0,1]} \lambda^{2t}\diff\tau(\lambda) \geq c_1 \int_{0}^{1} \lambda^{2t} (1-\lambda)^{d/2-1}\diff\lambda = B(2t+1,d/2)   \stackrel{ \eqref{eq:asymptotic-beta}}{\underset{t\to\infty}{\sim}} \frac{c_3}{t^{d/2}} \, ,
\end{align*}
for some constant $c_3$. Thus
\begin{equation*}
\liminf_{t\to\infty} \frac{\int \lambda^{2t}\diff\tau(\lambda)}{\ln t} \geq -\frac{d}{2} \, .
\end{equation*}
This being true for all $d>\overline{\dim}_\rightarrow \tau$, this proves
\begin{equation*}
\liminf_{t\to\infty} \frac{\int \lambda^{2t}\diff\tau(\lambda)}{\ln t} \geq -\frac{\overline{\dim}_\rightarrow \tau}{2} \, .
\end{equation*}
The proof at the other edge of the spectrum is the same by symmetry.
\end{proof}

The proof of Theorem \ref{thm:rate-decrease-Jacobi} for simple gossip follows easily. Indeed, if $\tau = \sigma$ is the spectral measure of the graph, then  $\overline{\dim}_\rightarrow \sigma = d_s$. Thus
\begin{equation*}
\liminf_{t\to\infty} \frac{\ln \E[(x^t_v-\mu)^2]}{\ln t} \stackrel{(\text{Lemma \ref{lem:proof-technical-1}})}{=} \liminf_{t\to\infty} \frac{\ln \int \lambda^{2t}\diff\sigma(\lambda)}{\ln t} \stackrel{\text{(Proposition \ref{prop:aux-1})}}{\geq} -\frac{d_s}{2} \, .
\end{equation*} 

\subsection{Proof of Theorem \ref{thm:rate-decrease-Jacobi}: shift-register}
\label{ap:proof-thm-rate-decrease-Jacobi_shift-register}
In the case of the shift-register gossip iteration, $P_t(\lambda)$ satisfies the second-order recurrence relation 
\begin{align}
\label{eq:shift-register-polynomials}
&P_0(\lambda) = 1 \, , &&P_1(\lambda) = \lambda \, , &&P_{t+1}(\lambda) = \omega\lambda P_t(\lambda) + (1-\omega)P_{t-1}(\lambda) \, .
\end{align}
The case $\omega=1$ corresponds to simple gossip: it has been treated above. We now assume $\omega \in (1,2]$.

\begin{proposition}
	\label{prop:polynomials-shift-register-chebyshev}
Let $P_t$ be the polynomials defined in Eq.~\eqref{eq:shift-register-polynomials} with $\omega \in (1,2]$. Then 
\begin{equation*}
P_t(\lambda) = (\omega-1)^{t/2} \left[\left(2-\frac{2}{\omega}\right)T_t\left(\frac{\omega}{2\sqrt{\omega-1}}\lambda\right) + \left(\frac{2}{\omega}-1\right)U_t\left(\frac{\omega}{2\sqrt{\omega-1}}\lambda\right)\right]
\end{equation*}
where $T_t$ and $U_t$ are the Chebyshev polynomials of the first and second kind respectively (see Example~\ref{ex:chebyshev}).
\end{proposition}

\begin{proof}
Consider the rescaled version $Q_t$ of $P_t$ given by the formula 
\begin{equation}
\label{eq:aux-17}
P_t(\lambda) = (\omega-1)^{t/2}Q_t\left(\frac{\omega}{2\sqrt{\omega-1}}\lambda\right) \, .
\end{equation}
If follows from Eq.~\eqref{eq:shift-register-polynomials} that
\begin{align*}
&Q_0(\lambda) = 1 \, , &&Q_1(\lambda) = \frac{2}{\omega}\lambda \, , &&Q_{t+1}(\lambda) = 2\lambda Q_t(\lambda) - Q_{t-1}(\lambda) \, . 
\end{align*}
Thus the sequence $Q_t, t\geq 0$ satisfies the same recurrence relation as the Chebyshev polynomials, but with a different initialization. As a consequence, it must be a linear combination of the two sequences of Chebyshev polynomials: there exists $\mu, \nu \in \R$ such that for all $t$, 
\begin{equation*}
Q_t(\lambda) = \mu T_t(\lambda) + \nu U_t(\lambda)
\end{equation*}
The computation of the weights $\mu, \nu$ is straightforward from the initialization $Q_0, Q_1$. This proves the proposition. 
\end{proof}

\begin{proposition}
	\label{prop:shift-register-orthogonality-measure}
Let $\omega \in (1,2]$. The polynomials $P_t, t\geq 0$ defined in Eq.~\eqref{eq:shift-register-polynomials} are the orthogonal polynomials w.r.t.~the measure 
\begin{equation*}
\tau(\diff\lambda) = \frac{\left(\left(2\sqrt{\omega-1}/\omega\right)^2-\lambda^2\right)^{1/2}}{1-\lambda^2} \, . 
\end{equation*}
\end{proposition}

\begin{proof}
The orthogonal polynomials w.r.t.~the measure 
\begin{equation*}
\tilde{\tau}(\diff\lambda) = \frac{\left(1-\lambda^2\right)^{1/2}}{\left(\omega/(2\sqrt{\omega-1})\right)^2-\lambda^2}
\end{equation*}
are computed using Proposition \ref{prop:bernstein-szego} with 
\begin{equation*}
\rho(\cos \theta) = \frac{\omega^2}{4(\omega-1)}-\cos^2\theta \, . 
\end{equation*}
Simple computations give that $\rho(\cos \theta) = |h(e^{i\theta})|^2$ with
\begin{equation*}
h(e^{i\theta}) = \left\vert\frac{\omega}{2\sqrt{\omega-1}}\left[  \left(2-\frac{2}{\omega}\right)\frac{1-e^{2i\theta}}{2} + \left(\frac{2}{\omega}-1\right) \right]\right\vert^2 \, .
\end{equation*}
Proposition \ref{prop:bernstein-szego} then gives that the polynomials $(2-2/\omega)T_t+(2/\omega-1)U_t$ are orthogonal w.r.t.~$\tilde{\tau}$. But these polynomials are the polynomials $Q_t$ defined in Eq.~\eqref{eq:aux-17}. We then use Lemma \ref{lem:rescaled_orthogonal_pol} to prove that $P_t$ is orthogonal w.r.t.~$\tau$.
\end{proof}

\begin{lemma}
Let $P_t$ be the polynomials defined in Eq.~\eqref{eq:shift-register-polynomials} with $\omega \in (1,2]$ and $\tau$ a measure on $[-1,1]$. Then 
\begin{equation*}
\liminf_{t\to\infty} \frac{\int P_t(\lambda)^2\diff\tau(\lambda)}{\ln t} \geq -\frac{\min\left(\overline{\dim}_\rightarrow \tau, \overline{\dim}_\leftarrow \tau\right)}{2}
\end{equation*}
\end{lemma}

\begin{proof}
\begin{align*}
\int P_t(\lambda)^2\diff\tau(\lambda) &\geq \int_{[2\sqrt{\omega-1}/\omega,1]} P_t(\lambda)^2\diff\tau(\lambda) \\
&\hspace{-0.9cm}\stackrel{\text{(Proposition \ref{prop:polynomials-shift-register-chebyshev})}}{\geq} c_1(\omega-1)^t \int_{[2\sqrt{\omega-1}/\omega,1]}T_t\left(\frac{\omega}{2\sqrt{\omega-1}}\lambda\right)^2\diff\tau(\lambda)\\
&\hspace{-1mm}\stackrel{\eqref{eq:aux-9}}{\geq} c_2(\omega-1)^t \int_{[2\sqrt{\omega-1}/\omega,1]} \left(\frac{\omega}{2\sqrt{\omega-1}}\lambda + \sqrt{\frac{\omega^2}{4(\omega-1)}\lambda^2-1}\right)^{2t} \diff\tau(\lambda) \, . 
\end{align*}
where $c_i > 0$ is a constant independent of $t$. Let $d > \overline{\dim}_\rightarrow \tau$. As $\overline{\dim}_\rightarrow \tau = 2\limsup_{E\to0}\ln \sigma([1-E,1])/\ln E$, there exists constants $c_3, c_4 > 0$ such that for all $E \in [0,2]$, 
\begin{equation}
\label{eq:aux-18}
\tau([1-E,1]) \geq c_3E^{d/2} = c_4\sigma^{(d/2-1,0)}([1-E,1])
\end{equation}
where $\sigma^{(d/2-1,0)}(\diff\lambda) = (1 - \lambda)^{d/2-1}\diff\lambda$. Then using jointly Lemma \ref{lem:domination-measure-increasing-functions} and Eq.~\eqref{eq:aux-18}, 
\begin{align*}
\int P_t(\lambda)^2\diff\tau(\lambda) &\geq c_5(\omega-1)^t \int_{2\sqrt{\omega-1}/\omega}^1 \left(\frac{\omega}{2\sqrt{\omega-1}}\lambda + \sqrt{\frac{\omega^2}{4(\omega-1)}\lambda^2-1}\right)^{2t} (1-\lambda)^{d/2-1}\diff\lambda \\
&\geq c_6(\omega-1)^t \int_{0}^{\cosh^{-1}(\omega/(2\sqrt{\omega-1}))} e^{2tu}\left(1-\frac{2\sqrt{\omega-1}}{\omega}\cosh u\right)^{d/2-1}\sinh u \, \diff u \, .
\end{align*}
where in the last step we made the change of variable $\omega/(2\sqrt{\omega-1})\lambda + \sqrt{\omega^2/(4(\omega-1))\lambda^2-1} = e^u$, i.e. $\lambda = 2\sqrt{\omega-1}/\omega \, \cosh u$. Denote $u_{\max} = \cosh^{-1}(\omega/(2\sqrt{\omega-1}))$ to shorten notations. As $\cosh$ is a convex function, for $u \in [0,u_{\max}]$, 
\begin{align*}
&\cosh u_{\max} - \cosh u \leq \frac{\cosh u_{\max} - 1}{u_{\max}} (u_{\max} - u ) \, , \\
&\Leftrightarrow 1 - \frac{2\sqrt{\omega-1}}{\omega} \cosh u \leq \left(1-\frac{2\sqrt{\omega-1}}{\omega}\right)\left(1-\frac{u}{u_{\max}}\right)
\end{align*}
Moreover, choose some constant $u_{\min} \in (0,u_{\max})$ so that we can lower bound with a constant $c_7$: for all $u \in [u_{\min},u_{\max}]$, $\sinh u \geq c_7$. This finally gives: 
\begin{equation*}
\int P_t(\lambda)^2\diff\tau(\lambda) \geq c_8(\omega-1)^t \int_{u_{\min}}^{u_{\max}} e^{2tu}\left(1-\frac{u}{u_{\max}}\right)^{d/2-1} \, \diff u \, .
\end{equation*}
After the change of variable $w = 2t(u_{\max}-u)$, this gives
\begin{equation*}
\int P_t(\lambda)^2\diff\tau(\lambda) \geq c_8(\omega-1)^t \int_{0}^{2t(u_{\max}-u_{\min})} e^{2tu_{\max}}e^{-w}\left(\frac{w}{2tu_{\max}}\right)^{d/2-1} \frac{1}{2t}\, \diff w \, .
\end{equation*}
Note that $e^{2tu_{\max}} = (\omega-1)^{-t}$, thus there exists a constant $c_9>0$ such that 
\begin{equation*}
\int P_t(\lambda)^2\diff\tau(\lambda) \geq c_9 \frac{1}{t^{d/2}} \int_{0}^{2t(u_{\max}-u_{\min})} e^{-w}w^{d/2-1} \, \diff w \, .
\end{equation*}
This proves that 
\begin{equation*}
\liminf_{t\to\infty} \frac{\int P_t(\lambda)^2\diff\tau(\lambda)}{\ln t} \geq -\frac{d}{2} \, .
\end{equation*}
This being true for all $d > \overline{\dim}_\rightarrow \tau$, this proves
\begin{equation*}
\liminf_{t\to\infty} \frac{\int P_t(\lambda)^2\diff\tau(\lambda)}{\ln t} \geq -\frac{\overline{\dim}_\rightarrow \tau}{2} \, .
\end{equation*}
The proof at the other edge of the spectrum is the same by symmetry. 
\end{proof}

\subsection{Proof of Theorem \ref{thm:rate-decrease-Jacobi}: Jacobi polynomial iteration}
\label{ap:proof-thm-rate-decrease-Jacobi_jacobi}

In this section, we use again the notation of Appendix \ref{ap:Jacobi}: in the case of the Jacobi polynomial iteration \eqref{eq:Jabobi-polynomial-iteration}, we have $x^t = \pi_t^{(d_s/2,0)}(W)\xi$, where $\pi_t^{(\alpha,\beta)}$ is the rescaled Jacobi polynomial; $\pi_t^{(\alpha,\beta)} =  P_t^{(\alpha,\beta)}/{t+\alpha \choose t}$ where $P_t^{(\alpha,\beta)}$ is the traditional Jacobi polynomial. Lemma \ref{lem:proof-technical-1} suggests to study the quantity $\int \pi_t^{(d_s/2,0)}(\lambda)^2\diff\sigma(\lambda)$. However we study here the behavior of $\int \pi_t^{(\alpha,\beta)}(\lambda)^2\diff\sigma(\lambda)$ for any $(\alpha,\beta)$. This will be useful in Appendix \ref{ap:tuning-Jacobi} to give a motivation for the choice $\alpha = d/2, \beta = 0$ complementary to the intuition developed in Section \ref{sec:Jacobi-polynomial-gossip}, and will allow us to discuss the performance of other choices. 

\begin{proposition}
	\label{prop:jacobi-general-asymptotic}
Let $\tau$ be a probability measure on $[-1,1]$ and $\alpha,\beta\geq-1/2$. Then 
\begin{equation*}
\limsup_{t\to\infty} \frac{\ln\int \pi_t^{(\alpha,\beta)}(\lambda)^2\diff\tau(\lambda)}{\ln t} \leq - \min\left(2\alpha+1, \underline{\dim}_\rightarrow\tau, 2(\alpha-\beta)+\underline{\dim}_\leftarrow\tau\right) \, .
\end{equation*}
\end{proposition}

Before proving this proposition, we use it to finish the proof of the theorem. If $\tau = \sigma$ is the spectral measure of the graph, then $\underline{\dim}_\rightarrow\sigma = d_s$. Thus taking $\alpha = d_s/2$, $\beta=0$ in Proposition \ref{prop:jacobi-general-asymptotic}, we get
\begin{equation}
\label{eq:aux-5}
\limsup_{t\to\infty} \frac{\ln\int \pi_t^{(d_s/2,0)}(\lambda)^2\diff\sigma(\lambda)}{\ln t} \leq - \min(d_s + 1, d_s, d_s + \underline{\dim}_\leftarrow\sigma) = -d_s \, ,
\end{equation}
as $\underline{\dim}_\leftarrow\sigma \geq 0$. One can conclude the proof using Lemma \ref{lem:proof-technical-1}.

\smallskip
We now turn to the the proof of Proposition \ref{prop:jacobi-general-asymptotic}. 

\begin{lemma}
	\label{lem:proof-decrease-aux-2}
Let $\tau$ be a probability measure on $[-1,1]$ and $\alpha \geq -1/2$, $\beta >-1$. Then 
\begin{equation*}
\limsup_{t\to\infty} \frac{\ln\int_{[0,1]} P_t^{(\alpha,\beta)}(\lambda)^2\diff\tau(\lambda)}{\ln t} \leq -\min(1,\underline{\dim}_\rightarrow\tau-2\alpha) \, .
\end{equation*}
\end{lemma}

\begin{proof}
Let $d < \underline{\dim}_\rightarrow\tau$. As $\underline{\dim}_\rightarrow\tau = 2 \liminf_{E\to 0} \ln\tau([1-E,1])/\ln E$, there exists constants $C_1, C_2$ such that for all $E\in[0,2]$, 
\begin{equation}
\label{eq:aux-13}
\tau([1-E,1]) \leq C_1E^{d/2} = C_2\sigma^{(d/2-1,0)}([1-E,1])
\end{equation}
where $\sigma^{(d/2-1,0)}(\diff\lambda) = (1-\lambda)^{d/2-1}\diff\lambda$. 

For the proof of this result, we use the asymptotic bound on the Jacobi polynomials given by Proposition \ref{prop:asymptotic-jacobi-in}, thus we divide the integral 
\begin{equation*}
\int_{[0,1]} P_t^{(\alpha,\beta)}(\lambda)^2\diff\tau(\lambda) = \int_{[\cos 1/t,1]} P_t^{(\alpha,\beta)}(\lambda)^2\diff\tau(\lambda) + \int_{[0,\cos 1/t[} P_t^{(\alpha,\beta)}(\lambda)^2\diff\tau(\lambda) \, ,
\end{equation*}
and treat the two terms separately.
\begin{enumerate}[label=(\alph*)]
	\item \begin{align*}
		\int_{[\cos 1/t,1]} P_t^{(\alpha,\beta)}(\lambda)^2\diff\tau(\lambda) &\leq C_3 t^{2\alpha}\tau\left(\left[\cos\frac{1}{t},1\right]\right) \leq C_1C_3t^{2\alpha}\left(1-\cos\frac{1}{t}\right)^{d/2} \\&\leq C_4t^{2\alpha-d}
	\end{align*}
	for some constants $C_3, C_4$. Thus
	\begin{equation}
	\label{eq:aux-3}
	\limsup_{t\to\infty} \frac{\ln\int_{[\cos 1/t,1]} P_t^{(\alpha,\beta)}(\lambda)^2\diff\tau(\lambda)}{\ln t} \leq 2\alpha-d \, .
	\end{equation}
	\item \begin{equation}
	\label{eq:aux-14}
	\int_{[0,\cos 1/t[} P_t^{(\alpha,\beta)}(\lambda)^2\diff\tau(\lambda) \leq C_5t^{-1}\int_{[0,\cos(1/t)[} (\arccos \lambda)^{-2\alpha-1}\diff\tau(\lambda)
	\end{equation}
	We then use jointly Eq.~\eqref{eq:aux-13} and Lemma \ref{lem:domination-measure-increasing-functions} with the function 
	\begin{equation*}
	f(\lambda) = (\arccos\lambda)^{-2\alpha-1}\bfone_{\{\lambda < \cos1/t\}}+t^{2\alpha+1}\bfone_{\{\lambda \geq \cos1/t\}} \, .
	\end{equation*}
	Note that $f$ is non-decreasing as $\alpha \geq 1/2$. We get
	\begin{align*}
	\int_{[0,\cos(1/t)[} &(\arccos \lambda)^{-2\alpha-1}\diff\tau(\lambda) \\&\leq C_2\int_{[0,\cos(1/t)[} (\arccos \lambda)^{-2\alpha-1}(1-\lambda)^{d/2-1}\diff\lambda + C_1t^{2\alpha+1} \left(1-\cos \frac{1}{t}\right)^{d/2} 
	\end{align*}
	Now using the simple inequality $\arccos \lambda \geq \sqrt{2}\sqrt{1-\lambda}$, we get
	\begin{equation}
	\label{eq:aux-15}
	\begin{aligned}
	\int_{[0,\cos(1/t)[} &(\arccos \lambda)^{-2\alpha-1}\diff\tau(\lambda) \\&\leq C_5\int_{[0,\cos(1/t)[}(1-\lambda)^{-\alpha+d/2-3/2}\diff\lambda + C_6t^{2\alpha+1-d} 
	\end{aligned}
	\end{equation}
	for some constants $C_5, C_6$. 
	Now if $\beta$ is a real number, 
	\begin{equation}
	\label{eq:aux-1}
	\lim_{t\to\infty} \frac{\ln \int_{0}^{\cos 1/t}(1-\lambda)^\beta\diff\lambda}{\ln t} = \max(0,-2\beta-2) \, .
	\end{equation}
	Indeed, if $\beta \neq -1$, 
	\begin{align*}
	\int_{0}^{\cos 1/t}(1-\lambda)^\beta\diff\lambda &= \left[-\frac{(1-\lambda)^{\beta+1}}{\beta+1}\right]_0^{\cos 1/t} = \frac{1}{\beta+1}\left[1-\left(1-\cos\frac{1}{t}\right)^{\beta+1}\right] \\
	&\underset{t\to\infty}{=} \frac{1}{\beta+1}\left[1-t^{-2\beta-2}+o(t^{-2\beta-2})\right] \sim C(\beta) t^{\max(0,-2\beta-2)} \, .
	\end{align*}
	for some constant $C(\beta)$ depending on $\beta$. This proves the statement \eqref{eq:aux-1} for $\beta \neq -1$. The result for $\beta = -1$ follows easily by noting that both terms in \eqref{eq:aux-1} are decreasing in $\beta$.
	
	Merging finally Eqs. \eqref{eq:aux-14}, \eqref{eq:aux-15} and \eqref{eq:aux-1}, we get 
	\begin{equation}
	\label{eq:aux-4}
	\begin{aligned}
	\limsup_{t\to\infty} \frac{\ln\int_{[0,\cos 1/t[} P_t^{(\alpha,\beta)}(\lambda)^2\diff\tau(\lambda)}{\ln t}&\leq -1+\max(0,2\alpha+1-d) \\
	&= \max(-1,2\alpha-d) = -\min(1,d-2\alpha) \, .
	\end{aligned}
	\end{equation}
\end{enumerate}
Finally
\begin{align*}
\limsup_{t\to\infty}& \frac{\ln\int_{[0,1]} P_t^{(\alpha,\beta)}(\lambda)^2\diff\tau(\lambda)}{\ln t} \\
&\leq \limsup_{t\to\infty} \frac{2\max\left(\int_{[\cos 1/t,1]} P_t^{(\alpha,\beta)}(\lambda)^2\diff\tau(\lambda),\int_{[0,\cos 1/t[} P_t^{(\alpha,\beta)}(\lambda)^2\diff\tau(\lambda)\right)}{\ln t} \\
&\leq \max\left(\limsup_{t\to\infty} \frac{\ln\int_{[\cos 1/t,1]} P_t^{(\alpha,\beta)}(\lambda)^2\diff\tau(\lambda)}{\ln t},\limsup_{t\to\infty} \frac{\ln\int_{[0,\cos 1/t[} P_t^{(\alpha,\beta)}(\lambda)^2\diff\tau(\lambda)}{\ln t}\right) \\
&\stackrel{\eqref{eq:aux-3},\eqref{eq:aux-4}}{\leq} \max(2\alpha-d,-\min(1,d-2\alpha)) = -\min(1,d-2\alpha) \, .
\end{align*}
As this is true for all $d < \underline{\dim}_\rightarrow\tau$, the lemma is proved.
\end{proof}

\begin{proof}[Proof of Proposition \ref{prop:jacobi-general-asymptotic}]
If we denote $\tilde{\tau}$ the symmetric measure of $\tau$ w.r.t.~$0$ (i.e. the image measure of $\tau$ by the map $\lambda \mapsto -\lambda$), we have 
\begin{equation*}
\int_{[-1,0]} P_t^{(\alpha,\beta)}(\lambda)^2\diff\tau(\lambda) = \int_{[0,1]} P_t^{(\alpha,\beta)}(-\lambda)^2\diff\tilde{\tau}(\lambda) = \int_{[0,1]} P_t^{(\beta,\alpha)}(\lambda)^2\diff\tilde{\tau}(\lambda)
\end{equation*}
Thus according to Lemma \ref{lem:proof-decrease-aux-2},
\begin{equation}
\label{eq:proof-decrease-aux-3}
\limsup_{t\to\infty} \frac{\ln\int_{[-1,0]} P_t^{(\alpha,\beta)}(\lambda)^2\diff\tau(\lambda)}{\ln t} \leq -\min(1,\underline{\dim}_\rightarrow\tilde{\tau}-2\beta) = -\min(1,\underline{\dim}_\leftarrow\tau-2\beta) \, .
\end{equation}
Finally, using that  $\pi_t^{(\alpha,\beta)} =  P_t^{(\alpha,\beta)}/{t+\alpha \choose t}$, 
\begin{align*}
\limsup_{t\to\infty} \frac{\ln\int_{[-1,1]} \pi_t^{(\alpha,\beta)}(\lambda)^2\diff\tau(\lambda)}{\ln t} &\leq \limsup_{t\to\infty} \frac{\ln\int_{[-1,1]} P_t^{(\alpha,\beta)}(\lambda)^2\diff\tau(\lambda)}{\ln t} - 2\limsup_{t\to\infty} \frac{\ln{t+\alpha \choose t}}{\ln t} \\
&\hspace{-3cm}\leq \limsup_{t\to\infty} \frac{\ln\left( 2\max\left(\int_{[-1,0]} P_t^{(\alpha,\beta)}(\lambda)^2\diff\tau(\lambda),\int_{[0,1]} P_t^{(\alpha,\beta)}(\lambda)^2\diff\tau(\lambda)\right)\right)}{\ln t} - 2\alpha \\
&\hspace{-3cm}\leq \max\left(\limsup_{t\to\infty} \frac{\ln\int_{[-1,0]} P_t^{(\alpha,\beta)}(\lambda)^2\diff\tau(\lambda)}{\ln t},\limsup_{t\to\infty} \frac{\ln\int_{[0,1]} P_t^{(\alpha,\beta)}(\lambda)^2\diff\tau(\lambda)}{\ln t}\right) -2\alpha \\
&\hspace{-4cm}\stackrel{\rm(\eqref{eq:proof-decrease-aux-3}, Lemma\,\ref{lem:proof-decrease-aux-2})}{\leq} \max\left(-\min(1,\underline{\dim}_\leftarrow\tau-2\beta),-\min(1,\underline{\dim}_\rightarrow\tau-2\alpha)\right)-2\alpha \\
&\hspace{-3cm}\leq -\min(1,\underline{\dim}_\leftarrow\tau-2\beta,\underline{\dim}_\rightarrow\tau-2\alpha) - 2\alpha \\
&\hspace{-3cm}= -\min(2\alpha+1, 2(\alpha-\beta)+\underline{\dim}_\leftarrow\tau,\underline{\dim}_\rightarrow\tau) \, .
\end{align*}
\end{proof}

\subsection{Tuning of the parameters $\alpha$ and $\beta$}
\label{ap:tuning-Jacobi}

In this section, we discuss the performance of the polynomial gossip iteration $x^t = \pi_t^{(\alpha,\beta)}(W)\xi$ using the tools developed in the proof above. The Jacobi polynomial iteration introduced in Section \ref{sec:Jacobi-polynomial-iteration-subsection} corresponds to the specific choice $\alpha = d_s/2, \beta=0$, where $d_s$ is the spectral measure of the graph. Thanks to the tools developed in the proof above, we can explore the effect of changing $\alpha$ and $\beta$ analytically. Inspired by \eqref{eq:aux-5}, we define optimality as follows

\begin{definition}
	Let $\alpha, \beta > -1, d_\leftarrow, d_\rightarrow \geq 0$. We say that $(\alpha,\beta)$ is optimal for $(d_\leftarrow, d_\rightarrow)$ if for any spectral measure $\sigma$ such that $\underline{\dim}_\rightarrow \sigma = d_\rightarrow$ and $\underline{\dim}_\leftarrow \sigma = d_\leftarrow$,
\begin{equation*}
\limsup_{t\to\infty} \frac{\ln\int \pi_t^{(\alpha,\beta)}(\lambda)^2\diff\sigma(\lambda)}{\ln t} \leq -d_\rightarrow \, .
\end{equation*}
\end{definition}
The following theorem is an analogue of the optimality theorem \ref{thm:rate-decrease-Jacobi}\eqref{enu:jacobi} in the general case.

\begin{theorem}
	\label{thm:Jacobi-decrease-general}
Consider a graph $G$, a gossip matrix $W$ and a vertex $v$. Denote $\sigma = \sigma(G,W,v)$ the spectral measure of the graph. Let $\xi_v, v\in V$ be i.i.d.~samples from a distribution of mean $\mu$. 

Let $\alpha, \beta>-1$ and define the polynomial iteration $x^t = \pi_t^{(\alpha,\beta)}(W)\xi$. If $(\alpha,\beta)$ is optimal for $(\underline{\dim}_\leftarrow\sigma,\underline{\dim}_\rightarrow\sigma)$, then 
\begin{equation}
\label{eq:aux-7}
\limsup_{t\to\infty} \frac{\ln\E\left[(x^t_v-\mu)^2\right]}{\ln t} \leq -\underline{\dim}_\rightarrow\sigma \, .
\end{equation}
\end{theorem}
In the section above, we prove that $(d_\rightarrow/2,0)$ is optimal for $(d_\leftarrow,d_\rightarrow)$ (for any $d_\leftarrow, d_\rightarrow \geq -1/2$). We now explore other choices. According to Proposition \ref{prop:jacobi-general-asymptotic}, to prove that $(\alpha,\beta)$ is optimal for $(d_\leftarrow,d_\rightarrow)$, it is sufficient to prove that 
\begin{align}
\min(2\alpha+1, d_\rightarrow, 2(\alpha-\beta)+d_\leftarrow) = d_\rightarrow 
\qquad &\Leftrightarrow \qquad \begin{cases}
2\alpha+1 \geq d_\rightarrow \\
2(\alpha-\beta)+d_\leftarrow \geq d_\rightarrow
\end{cases} \nonumber \\
&\Leftrightarrow \qquad \begin{cases}
\alpha \geq \frac{1}{2}(d_\rightarrow-1)\\
\beta \leq \alpha + \frac{d_\leftarrow-d_\rightarrow}{2} 
\end{cases} \, . \label{eq:aux-6}
\end{align}
This gives a wide range of optimal parameters. For instance, the parameter $\alpha$ can be chosen arbitrarily large. In Figures {\scshape\ref{fig:tuning-1-parameters}} and {\scshape\ref{fig:tuning-2-parameters}}, the shaded regions corresponds to region for $(\alpha,\beta)$ defined by \eqref{eq:aux-6} with $(d_\leftarrow,d_\rightarrow)= (2,2)$. 

\begin{figure}
	\begin{subfigure}{0.49\linewidth}
		\includegraphics[width=\linewidth, trim = 0 0 0 0]{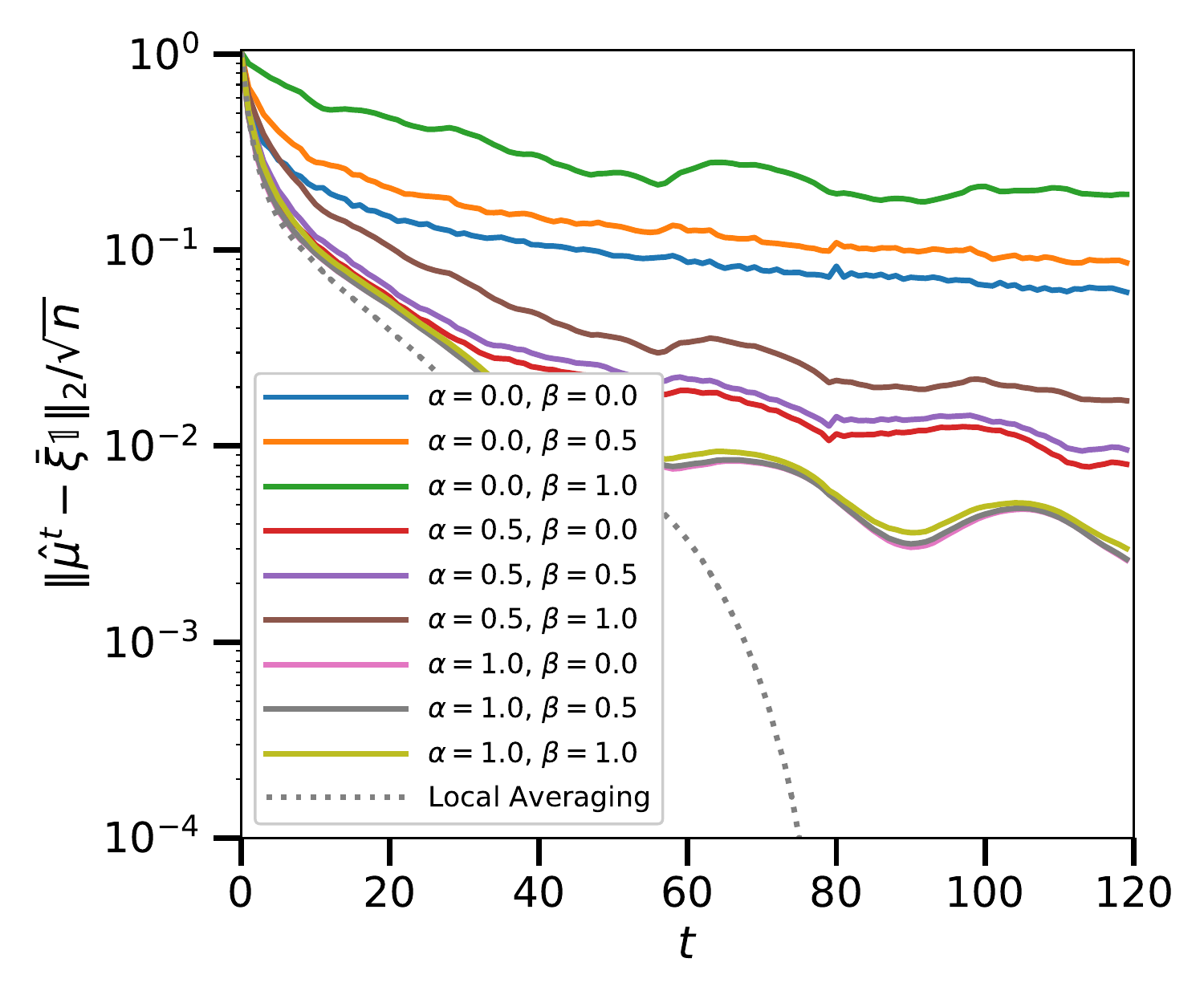}
		\caption{Performance curves}
		\label{fig:tuning-1-curves}
	\end{subfigure}
	\begin{subfigure}{0.49\linewidth}
		\includegraphics[width=\linewidth, trim = 0 0 0 0]{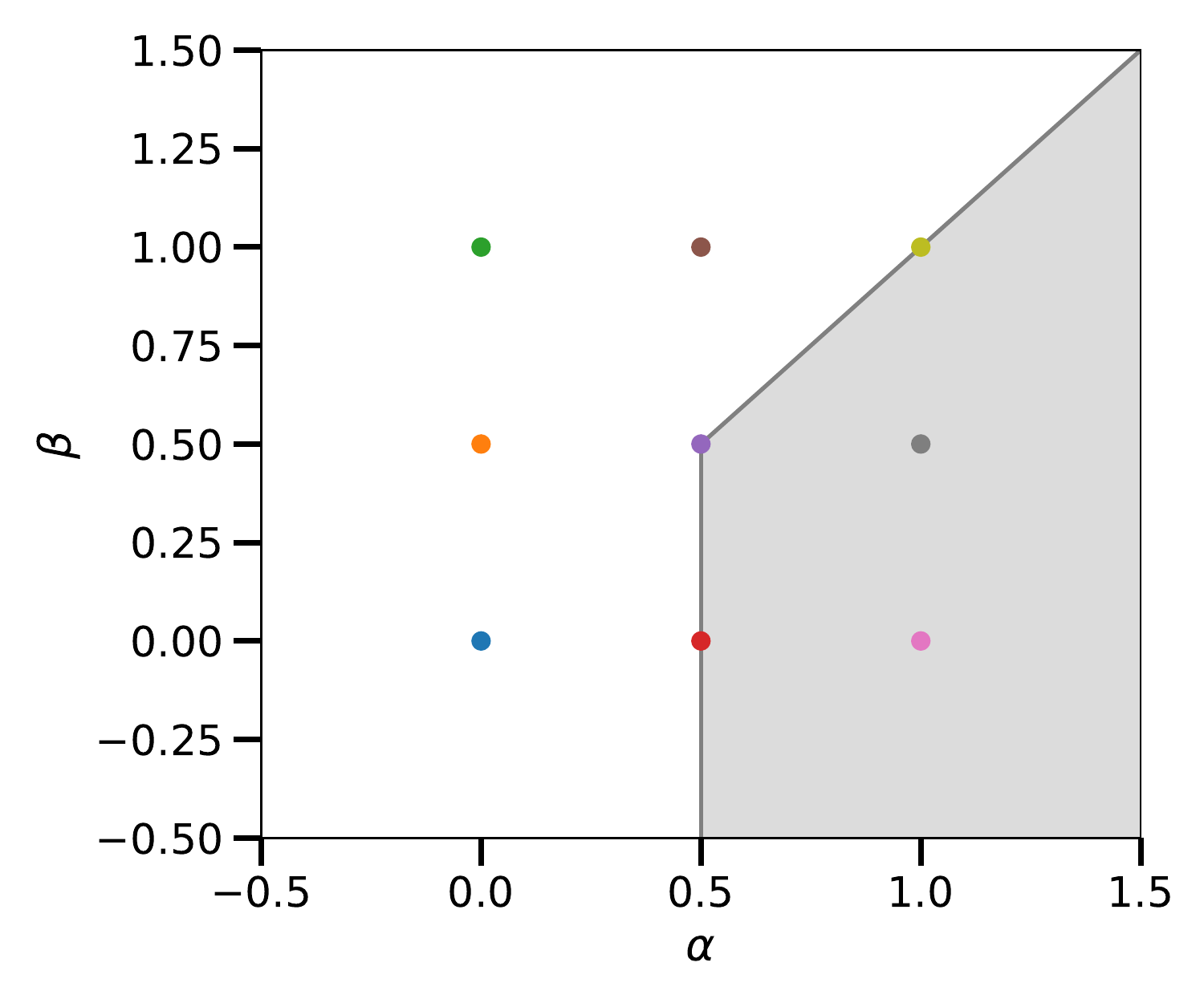}
		\caption{Parameter space}
		\label{fig:tuning-1-parameters}
	\end{subfigure}
	\caption{Simulations of polynomial iterations using Jacobi polynomials with different parameters $(\alpha,\beta)$: frontier tightness.}
	\label{fig:tuning-1} 
	\vspace{1cm}
	
	\begin{subfigure}{0.49\linewidth}
		\includegraphics[width=\linewidth, trim = 0 0 0 0]{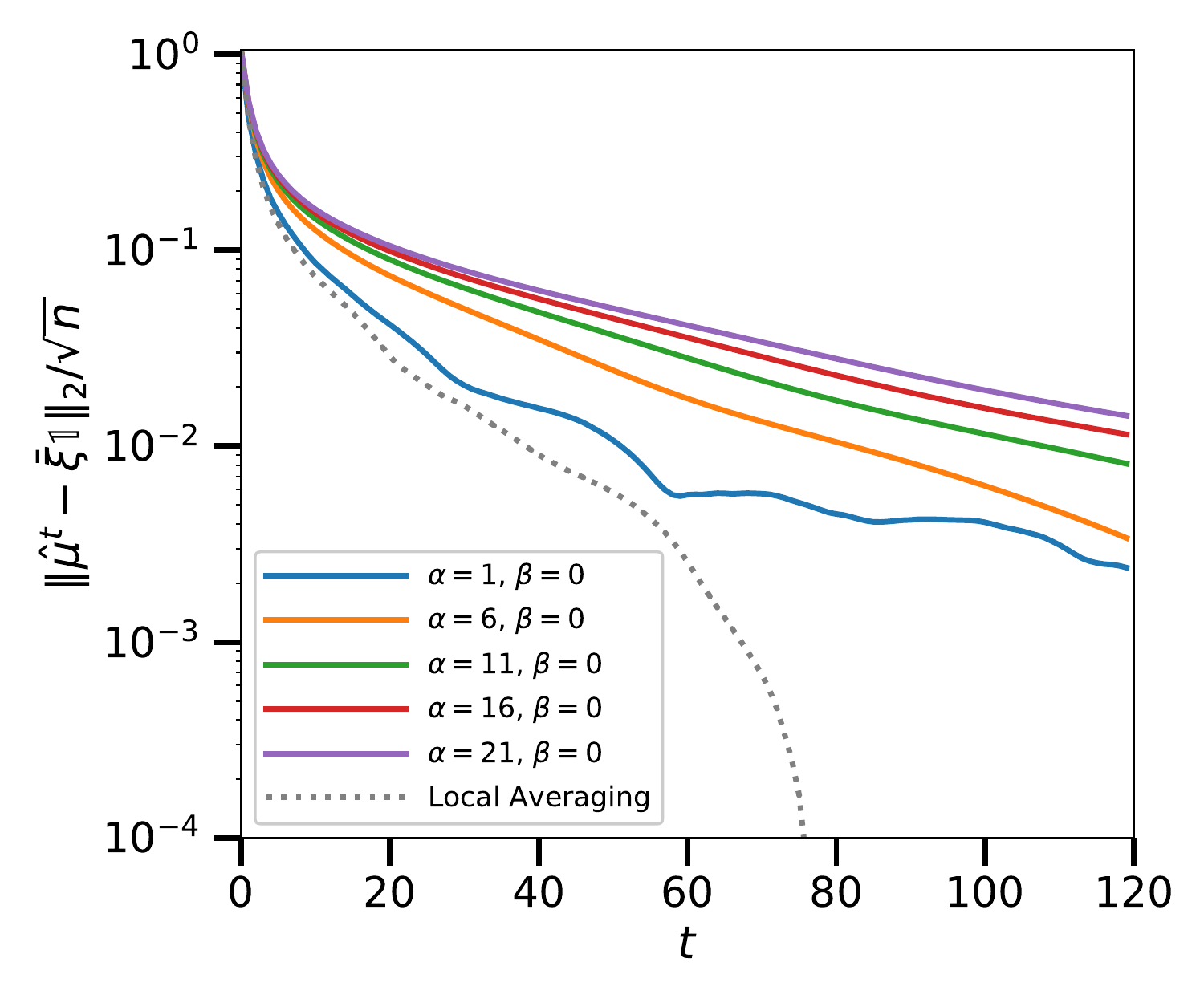}
		\caption{Performance curves}
		\label{fig:tuning-2-curves}
	\end{subfigure}
	\begin{subfigure}{0.49\linewidth}
		\includegraphics[width=\linewidth, trim = 0 0 0 0]{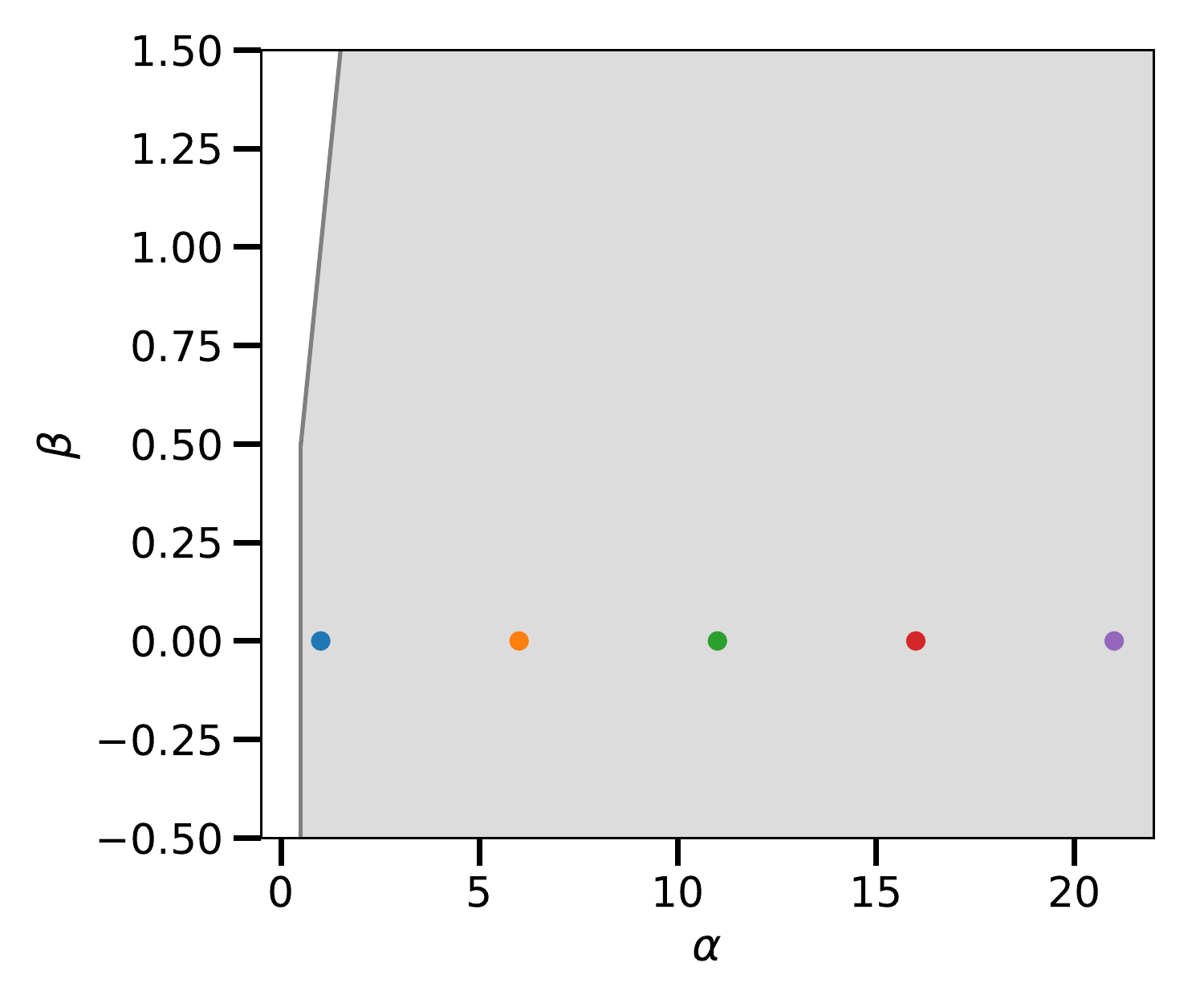}
		\caption{Parameter space}
		\label{fig:tuning-2-parameters}
	\end{subfigure}
	\caption{Simulations of polynomial iterations using Jacobi polynomials with different parameters $(\alpha,\beta)$: large $\alpha$ asymptotic.}
	\label{fig:tuning-2} 
\end{figure}

Note however that we have only proved that \eqref{eq:aux-6} are \emph{sufficient} conditions for the optimality Theorem \ref{thm:Jacobi-decrease-general} to hold. To explore the tightness of our condition, we present in Figure \ref{fig:tuning-1} the results of simulations on the 2D grid. The setting is the same as in Section \ref{sec:simulations} (see also Appendix \ref{ap:details-simulations} for details).  Note that for the 2D infinite grid, $d_\leftarrow = d_\rightarrow= 2$ (see Proposition \ref{prop:spectral-dim-Z^d} and the symmetry of the spectrum of $\Z^d$ that follows from \cite[Eq.(7.4)]{woess2000random}). The curves in Figure {\scshape\ref{fig:tuning-1-curves}} closest to the local averaging are those satisfying the condition \eqref{eq:aux-6}, thus our condition seems tight.

Finally, note that the result \eqref{eq:aux-7} of Theorem \ref{thm:Jacobi-decrease-general} gives the rate of the power decay of the MSE, but neglects constants and sub-polynomial factors. These factors depend on $(\alpha,\beta)$ and can be significant for extreme values of $(\alpha,\beta)$. For instance, in Figure \ref{fig:tuning-2}, we run simulations in the same setting as before, but for choices of parameters deeper in the optimality zone \eqref{eq:aux-6}. The performance worses as $\alpha$ gets bigger. So contrarily to what is suggested by \eqref{eq:aux-6} and Theorem \ref{thm:rate-decrease-Jacobi}, taking large values for $\alpha$ is a bad idea in practice. This can also be hinted at by the limit \cite[Eq.~(18.6.2)]{olver2010nist}
\begin{equation*}
\lim_{\alpha\to\infty} \pi_t^{(\alpha,\beta)}(\lambda) = \left(\frac{1+\lambda}{2}\right)^t \, .
\end{equation*}
This means that, as $\alpha \to \infty$, the polynomial gossip $x^t = \pi_t^{(\alpha,\beta)}(W)\xi$ converges to the simple gossip $x^t = \tilde{W}^t\xi$ with the gossip matrix $\tilde{W} = (I+W)/2$. We know from Theorem \ref{thm:rate-decrease-Jacobi}\eqref{enu:simple-gossip} that simple gossip is suboptimal. 

Overall, theory and practice suggest that the choice $\alpha = \underline{\dim}_\rightarrow\sigma/2$, $\beta=0$ that we make in Section \ref{sec:Jacobi-polynomial-iteration-subsection} is relevant. 

\section{Proof of Proposition \ref{prop:lower-bound}}
\label{ap:proof-lower-bound}

Note that the intuition lying behind the proposition is very simple: the unbiased estimator $x^t_v$ are linear combination of observations corresponding to vertices in the ball $B_v(t)$, thus it must have variance greater than $\var \nu/\vert B_v(t)\vert \approx \var \nu/t^{d_h}$. 

A more rigorous argument goes as follows: using that $W$ is a gossip matrix, it is easy to show by induction that for all $s \geq 0$ and $v, w \in V$, if $(W^s)_{vw} > 0$, then there exists a path of length $s$ linking $v$ to $w$ in $G$. As $\deg P_t \leq t$, this implies that $P_t(W)e_v$ has at most $\vert B_v(t)\vert $ non-zero entries. Furthermore, the entries of $P_t(W)e_v$ sum to $1$ because $W\bfone = \bfone$ and $P_t(1) = 1$. Thus, using the Cauchy-Schwarz inequality,
\begin{align*}
1 &= \left( \sum_{w \in V} (P_t(W)e_v)_w \right)^2 = \left(\sum_{w \in V} (P_t(W)e_v)_w \bfone_{\{(P_t(W)e_v)_w > 0\}}\right)^2 \\
&\leq \left\Vert P_t(W)e_v \right\Vert_{\ell^2(V)}^2 \sum_{w \in V} \bfone_{\{(P_t(W)e_v)_w > 0\}} \leq \left\Vert P_t(W)e_v\right\Vert_{\ell^2(V)}^2 \vert B_v(t)\vert  \\
&\hspace{-0.6cm}\stackrel{\text{(Lemma \ref{lem:proof-technical-1})}}{=} \E[(x^t_v-\mu)^2]|B_v(t)| \, .
\end{align*} 
Thus 
\begin{equation*}
\liminf_{t\to\infty} \frac{\ln \E[(x^t_v-\mu)^2]}{\ln t} \geq \liminf_{t\to\infty}- \frac{\ln |B_v(t)|}{\ln t}= -d_h \, .
\end{equation*}

\section{Proof of Theorem \ref{thm:asymptotic-rate-with-spectral-gap}}

In this section, we use the notation of Appendix \ref{ap:recurrence-relation-rescaled-jacobi}. As $x^t = \pi_t^{(d/2,0,\gamma)}(W)\xi$, we have 
\begin{equation}
\label{eq:aux-gap-0}
\Vert x^t-\bar{\xi}\bfone \Vert_2^2 = \sum_{i=2}^{n} \langle \xi, u^i \rangle^2 \pi_t^{(d/2,0,\gamma)}(\lambda_i)^2 \leq \Vert \xi - \bar{\xi}\bfone\Vert_2^2 \left(\sup_{\lambda \in [-1, 1-\gamma]} \vert\pi_t^{(d/2,0,\gamma)}(\lambda)\vert\right)^2 \, ,
\end{equation}
where $\lambda_2, \dots, \lambda_n$ are the eigenvalues of $W$ different from $1$, that lie in $[-1, 1-\gamma]$ by definition of~$\gamma$, and $u^2, \dots, u^n$ are the corresponding normalized eigenvectors.
\begin{align}
\sup_{\lambda \in [-1, 1-\gamma]} \vert\pi_t^{(d/2,0,\gamma)}(\lambda)\vert &\leq \frac{1}{\vert P_t^{(d/2,0,\gamma)}(1)\vert}\sup_{\lambda \in [-1, 1-\gamma]}\vert P_t^{(d/2,0,\gamma)}(\lambda)\vert \nonumber\\
&= \frac{1}{\vert\pi_t^{(d/2,0)}(\varphi_\gamma^{-1}(1))\vert}\sup_{\lambda \in \varphi_\gamma^{-1}([-1, 1-\gamma])}\vert\pi_t^{(d/2,0)}(\lambda)\vert \nonumber\\
&= \frac{1}{\left\vert\pi_t^{(d/2,0)}\left(\frac{1+\gamma/2}{1-\gamma/2}\right)\right\vert}\sup_{\lambda \in [-1, 1]}\vert\pi_t^{(d/2,0)}(\lambda)\vert \nonumber\\ 
&= \frac{1}{\left\vert P_t^{(d/2,0)}\left(\frac{1+\gamma/2}{1-\gamma/2}\right)\right\vert}\sup_{\lambda \in [-1, 1]}\vert P_t^{(d/2,0)}(\lambda)\vert \label{eq:aux-gap-1}
\end{align}
where $P_t^{(\alpha,\beta)}$ is the Jacobi polynomial, see Appendix \ref{ap:Jacobi}.

By Proposition \ref{prop:max-jacobi}, 
\begin{equation}
\label{eq:aux-gap-2}
\sup_{\lambda \in [-1, 1]}\vert P_t^{(d/2,0)}(\lambda)\vert = {t+d/2\choose t} \underset{t\to\infty}{\sim} t^{d/2} \, , 
\end{equation}
an by Proposition \ref{prop:asymptotic-jacobi-out} applied in $x = \frac{1+\gamma/2}{1-\gamma/2}$, there exists a positive constant $c$ such that,
\begin{align}
\label{eq:aux-gap-3}
P_t^{(d/2,0)}\left(\frac{1+\gamma/2}{1-\gamma/2}\right) \underset{t\to\infty}{\sim} ct^{-1/2}\left(\frac{(1+\sqrt{\gamma/2})^2}{1-\gamma/2}\right)^t \, .
\end{align}
Combining \eqref{eq:aux-gap-1}, \eqref{eq:aux-gap-2} and \eqref{eq:aux-gap-3}, we get that there exists a constant $C$ such that 
\begin{equation*}
\sup_{\lambda \in [-1, 1-\gamma]} \vert\pi_t^{(d/2,0,\gamma)}(\lambda)\vert \leq Ct^{(d+1)/2}\left(\frac{1-\gamma/2}{(1+\sqrt{\gamma/2})^2}\right)^t \, , 
\end{equation*}
and we conclude using \eqref{eq:aux-gap-0}.

\section{Proof of Proposition \ref{prop:message-passing_trees}}
\label{ap:proof-message-passing-trees}

Let $t \geq 0$ and $v, w \in V$ be two vertices linked by an edge in $G$. Define $B_{vw}(t)$ as the set of vertices $u$ in $B_w(t)$ such that all paths in the tree $G$ going from $u$ to $w$ pass though $v$.
\begin{center}
	\includegraphics[width=0.7\textwidth]{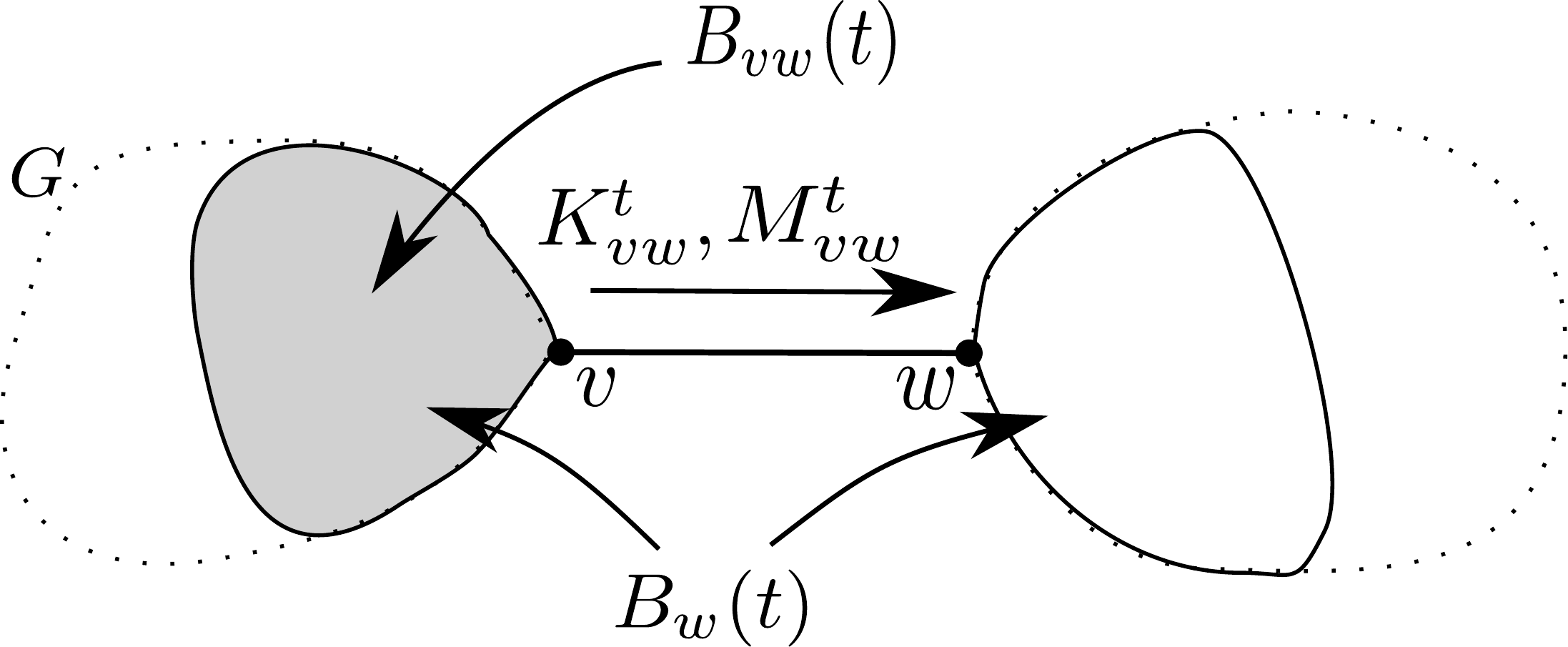}
\end{center}

\begin{lemma}
	\label{lem:MP_technical_lemma}
	For all $t \geq 0$, for all $v,w \in V$ linked by an edge in $G$, 
	\begin{equation*}
	K_{vw}^t = \vert  B_{vw}(t) \vert  \, , \qquad \textrm{and \quad if }t \geq 1,\quad M_{vw}^t = \frac{1}{\vert  B_{vw}(t) \vert } \sum_{u \in B_{vw}(t)} \xi_u \, . 
	\end{equation*}
\end{lemma}
\begin{proof}
	The proof goes by induction. The statement is trivial for $t=0,1$. For the induction, assume the result at time $t$ and note that 
	\begin{equation}
	\label{eq:ball_division}
	B_{vw}(t+1) = \{v\} \cup \left(\bigcup_{u\in\cN(v), \, u \neq w} B_{uv}(t) \right) \, ,
	\end{equation}
	where all unions are disjoint. This essentially comes from the fact that $G$ has no loops. 
	\begin{center}
		\includegraphics[width=0.7\textwidth]{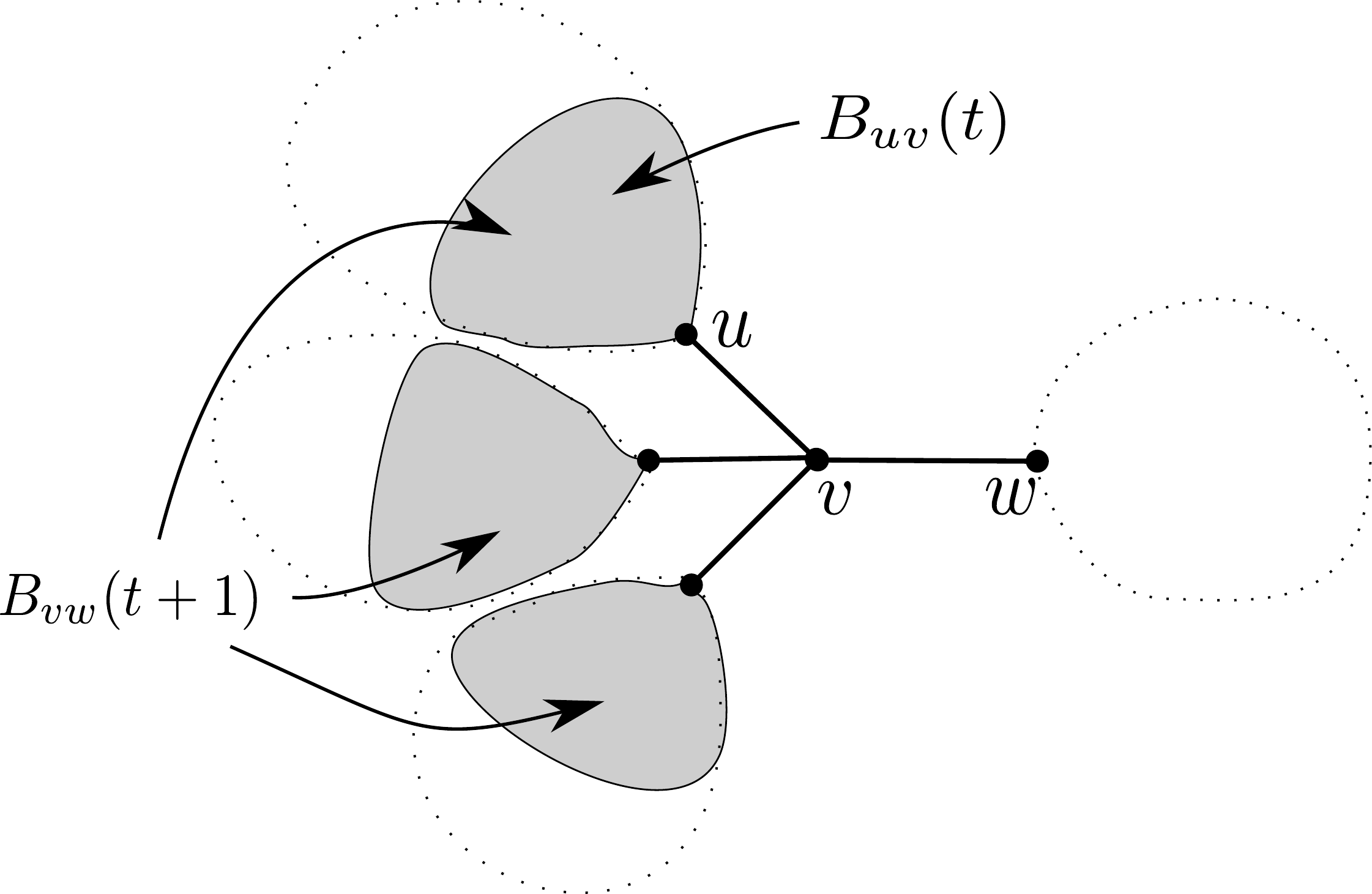}
	\end{center}
	Taking cardinal, we get that 
	\begin{equation*}
	\vert  B_{vw}(t+1) \vert  \stackrel{\eqref{eq:ball_division}}{=} 1 + \sum_{u\in\cN(v), \, u \neq w} \vert  B_{uv}(t) \vert \stackrel{\rm (induction)}{=} 1+ \sum_{u\in\cN(v), \, u \neq w} K_{uv}^t \stackrel{\eqref{eq:MP_1}}{=} K_{vw}^{t+1} \, .
	\end{equation*}
	This proves the induction for the first equality. The proof for the second equality is similar:
	\begin{align*}
	\frac{1}{\vert  B_{vw}(t+1) \vert } \sum_{u \in B_{vw}(t+1)} \xi_u &\stackrel{\eqref{eq:ball_division}}{=} \frac{1}{K_{vw}^{t+1}}\left(\xi_v +  \sum_{u\in\cN(v), \, u \neq w} \sum_{x \in B_{uv}(t)} \xi_x \right) \\
	&\stackrel{\rm (induction)}{=} \frac{\xi_v + \sum_{u\in\cN(v), \, u \neq w} \vert B_{uv}(t) \vert M_{uv}^t}{K_{vw}^{t+1}} \\
	&\stackrel{\rm (induction)}{=} \frac{\xi_v + \sum_{u\in\cN(v), \, u \neq w} K_{uv}^t M_{uv}^t}{K_{vw}^{t+1}} 
	\stackrel{\eqref{eq:MP_1}}{=} M_{vw}^{t+1} \, .
	\end{align*}
\end{proof}
We now end the proof of Proposition \ref{prop:message-passing_trees}. As $B_v(t) = \{v\}\cup\left(\bigcup_{u \in \cN(v)}  B_{uv}(t)\right)$ with disjoint unions, using Lemma \ref{lem:MP_technical_lemma}, we get 
\begin{equation*}
\frac{1}{\vert B_v(t) \vert}\sum_{w \in B_v(t)} \xi_w = \frac{\xi_v + \sum_{u \in \cN(v)} \sum_{w \in B_{uv}(t)} \xi_w}{1 + \sum_{u \in \cN(v)}\vert B_{uv}(t)\vert} 
= \frac{\xi_v + \sum_{u \in \cN(v)} K_{uv}^tM_{uv}^t}{1 + \sum_{u \in \cN(v)}K_{uv}^t} \stackrel{\eqref{eq:MP_2}}{=} x^t_v \, .
\end{equation*}

\section{Proof of Theorem \ref{thm:connection_mp_polynomial}}
\label{ap:proof-connection-mp-polynomial}

As noted by \cite{rebeschini2017accelerated}, the message passing iteration \eqref{eq:MP_1}-\eqref{eq:MP_2} indexed by the edges of the graph can be written as an iteration indexed by the vertices of the graph. We repeat here the elementary derivation of this statement in our particular case of $d$-regular graphs.

First, because $G$ is $d$-regular, it is an easy check from \eqref{eq:MP_1} that $K^t_{vw}$ does not depend on the edge $(v,w)$ (thus we denote it $K^t$) and it satisfies the recursion $K^0 = 0$, $K^{t+1} = 1+(d-1)K^t$.

Let us now denote $S^t_v = \xi_v + \sum_{u \in \cN(v)} K_{uv}^{t} M_{uv}^{t}$ and $L_t = 1 + dK^{t}$ so that $x^t_v = S^t_v / L_t$. We will now find recursions for $L_t$ and $S^t$:
\begin{equation*}
L_{t+1} = 1+dK^{t+1} \stackrel{\eqref{eq:MP_1}}{=} 1 + d(1+(d-1)K^t) = 2 + (d-1)(1+dK^t) = 2 + (d-1)L_t \, ,
\end{equation*}
and
\begin{align*}
S_v^{t+1} &= \xi_v + \sum_{u \in \cN(v)} K^{t+1}M_{uv}^{t+1} \stackrel{\eqref{eq:MP_1}}{=} \xi_v + \sum_{u \in \cN(v)}\left(\xi_u + \sum_{w \in \cN(u), w \neq v} K^tM^t_{wu}\right) \\
&= \xi_v + \sum_{u \in \cN(v)} \left(S^t_u - K^tM^t_{vu}\right) \, . 
\end{align*}
As 
\begin{align*}
\sum_{u \in \cN(v)}  K^tM^t_{vu} &\stackrel{\eqref{eq:MP_1}}{=} d \xi_v + \sum_{u \in \cN(v)} \sum_{w \in \cN(v),w \neq u}K^{t-1}M^{t-1}_{wv} \\
&= d\xi_v + (d-1)\sum_{w \in \cN(v)} K^{t-1}M_{wv}^{t-1} = \xi_v + (d-1)S_v^{t-1} \, ,
\end{align*}
we finally get 
\begin{equation*}
S^{t+1} = A(G)S^t - (d-1)S^{t-1} \, .
\end{equation*}
To sum up, we now have the simpler formulas for the message passing algorithm:
\begin{equation}
\label{eq:MP_vertex}
\begin{aligned}
&L_{t+1} = 2 + (d-1)L_t\, , && L_0 = 1\, , \\
&S^{t+1} = dW S^t - (d-1)S^{t-1}\, ,&&S^0 = \xi\, , && S^1 = \xi + dW\xi\, ,  \\ 
&x^t = S^t/L_t \, . && &&
\end{aligned}
\end{equation}
\smallskip
In Appendix \ref{ap:Kesten-McKay}, it is proved that $\pi_t(\lambda) = p_t(\lambda)/p_t(1)$ where $p_t$ satisfies the recursion formula 
\begin{align*}
&p_0(\lambda) = \sqrt{d-1} \, , 
&&p_1(\lambda) = d\lambda+1 \, , 
&&p_{t+1}(\lambda) = \frac{d}{\sqrt{d-1}}\lambda p_t(\lambda)-p_{t-1}(\lambda) \, , \quad t \geq 1 \, .
\end{align*}
Denote $q_t = (d-1)^{(t-1)/2}p_t$. It is an easy check that 
\begin{align*}
&q_0(\lambda) = 1 \, , 
&&q_1(\lambda) = d\lambda+1 \, , 
&&q_{t+1}(\lambda) = d\lambda q_t(\lambda)-(d-1)q_{t-1}(\lambda) \, , \quad t \geq 1 \, .
\end{align*}
Using \eqref{eq:MP_vertex}, one sees that for all $t$, $S^t = q_t(W)\xi$ and $L_t = q_t(1)$. Thus 
\begin{equation*}
x^t = \frac{S^t}{L_t} = \frac{q_t(W)\xi}{q_t(1)}= \frac{p_t(W)\xi}{p_t(1)} = \pi_t(W)\xi \, .
\end{equation*}

\section{Proof of Theorem \ref{thm:convergence_rate_message_passing}}

Theorem \ref{thm:connection_mp_polynomial} states that $x^t = \pi_t(W)\xi$ where the $\pi_t$ are the orthogonal polynomials w.r.t.~the modified Kesten-McKay measure $(1-\lambda)\sigma(\T_d)(\diff\lambda)$. Then 
\begin{equation}
\label{eq:aux-sup-0}
\Vert x^t-\bar{\xi}\bfone \Vert_2^2 = \sum_{i=2}^{n} \langle \xi, u^i \rangle^2 \pi_t(\lambda_i)^2 \leq \Vert \xi - \bar{\xi}\bfone\Vert_2^2 \left(\sup_{\lambda \in [-(1-\tilde{\gamma}), (1-\tilde{\gamma})]} \vert\pi_t(\lambda)\vert\right)^2 \, ,
\end{equation}
where $\lambda_2, \dots, \lambda_n$ are the eigenvalues of $W$ different from $1$, that lie in $[-(1-\tilde{\gamma}), (1-\tilde{\gamma})]$ by definition of the absolute spectral gap $\tilde{\gamma}$, and $u^2, \dots, u^n$ are the corresponding normalized eigenvectors. 

In Section \ref{ap:Kesten-McKay}, we show that 
\begin{align*}
&\pi_t(\lambda) = \frac{p_t(\lambda)}{p_t(1)} \, , \\
&p_t(\lambda) = \tilde{p}_t(\varphi^{-1}(\lambda)) \, , \qquad \varphi(\lambda) = \frac{2\sqrt{d-1}}{d}\lambda \, , \\
&\tilde{p}_t(\lambda) = \left(\sqrt{d-1}+\lambda\right)U_t(\lambda) - T_{t+1}(\lambda) \, , 
\end{align*}
where $T_t$ and $U_t$ are the Chebyshev polynomials of the first kind and of the second kind respectively. Denote $D = d/(2\sqrt{d-1})$. Then
\begin{equation}
\label{eq:aux-sup-1}
\sup_{\lambda \in [-(1-\tilde{\gamma}), (1-\tilde{\gamma})]} \vert\pi_t(\lambda)\vert = \frac{1}{\vert\tilde{p}_t(D)\vert} \sup_{\lambda \in [-(1-\tilde{\gamma})D, (1-\tilde{\gamma})D]} \vert\tilde{p}_t(\lambda)\vert \, .
\end{equation}
If $\lambda \in [-1,1]$, $|T_t(\lambda)|\leq 1$ and $|U_t(\lambda)| \leq t+1$. Thus 
\begin{equation}
\label{eq:aux-sup-2}
\sup_{\lambda \in [-1,1]} |\tilde{p}_t(\lambda)| \leq \left(\sqrt{d-1}+1\right)(t+1) +1 \, .
\end{equation}
We now discuss the different cases of the theorem.
\smallskip

\textbf{\emph{(1)}} We assume $\tilde{\gamma}<1-2\sqrt{d-1}/d$. As $\tilde{p}_t$ are orthogonal polynomials w.r.t.~some measure on $[-1,1]$, all zeros of $p_t$ are real, distinct and located in the interior of $[-1,1]$ (see Proposition \ref{prop:zeros}). It follows that 
\begin{align}
\label{eq:aux-sup-4}
&\sup_{\lambda \in (1,(1-\tilde{\gamma})D]} |\tilde{p}_t(\lambda)| = \left\vert\tilde{p}_t\left((1-\tilde{\gamma})D\right)\right\vert \, ,
&&\sup_{\lambda \in [-(1-\tilde{\gamma})D,-1)} |\tilde{p}_t(\lambda)| = \left\vert\tilde{p}_t\left(-(1-\tilde{\gamma})D\right)\right\vert \, .
\end{align}
Merging Eqs. \eqref{eq:aux-sup-1}-\eqref{eq:aux-sup-4}, we obtain
\begin{equation*}
\sup_{\lambda \in [-(1-\tilde{\gamma}), (1-\tilde{\gamma})]} \vert\pi_t(\lambda)\vert \leq \frac{1}{\vert\tilde{p}_t(D)\vert} \max\left(\vert\tilde{p}_t((1-\tilde{\gamma})D)\vert,\vert\tilde{p}_t(-(1-\tilde{\gamma})D)\vert,(\sqrt{d-1}+1)(t+1) +1\right) \, .
\end{equation*}
\begin{lemma}
	\label{lem:asymptotic_p}
	\begin{enumerate}
		\item 	If $x > 1$, then there exists a constant $C(d,x) \neq 0$ such that 
		\begin{equation}
		\tilde{p}_t(x) \underset{t\to\infty}{\sim} C(d,x) \left(x+\sqrt{x^2-1}\right)^t \, .
			\end{equation} 
		\item 	If $x < -1$, then there exists a constant $C(d,x) \neq 0$ such that 
		\begin{equation}
		\tilde{p}_t(x) \underset{t\to\infty}{\sim} C(d,x) \left(x-\sqrt{x^2-1}\right)^t \, .
		\end{equation} 
	\end{enumerate}
\end{lemma}
\begin{proof}
	In the proof of Lemma \ref{lem:bounds-chebyshev}, we developed the following formulas for the Chebyshev polynomials: 
	\begin{align*}
	&T_t\left(\frac{z+z^{-1}}{2}\right) = \frac{z^t+z^{-t}}{2} \, , &&U_t\left(\frac{z+z^{-1}}{2}\right) = \frac{z^{t+1}-z^{-(t+1)}}{z-z^{-1}} \, .
	\end{align*}
We write $x = (z+z^{-1})/2$ with $|z|> 1$. If $x > 1$, then $z = x + \sqrt{x^2-1}$, and if $x < -1$, then $z = x - \sqrt{x^2-1}$. Then 
	\begin{align*}
	\tilde{p}_t(x) &= \left(\sqrt{d-1}+x\right)U_t(x) - T_{t+1}(x) \\
	&= \left(\sqrt{d-1}+\frac{z+z^{-1}}{2}\right)\frac{z^{t+1}-z^{-(t+1)}}{z-z^{-1}} - \frac{z^{t+1}+z^{-(t+1)}}{2} \\
	&\underset{t\to\infty}{\sim} \left[\left(\sqrt{d-1}+\frac{z+z^{-1}}{2}\right)\frac{1}{z-z^{-1}}-\frac{1}{2}\right]z^{t+1} \, .
	\end{align*}
	The constant that appears is non-zero, thus the result is proved. 
\end{proof}
Using Lemma \ref{lem:asymptotic_p}, we get that there exists a constant $C(d)$ such that 
\begin{equation*}
\sup_{\lambda \in [-(1-\tilde{\gamma}), (1-\tilde{\gamma})]} \vert\pi_t(\lambda)\vert \leq C(d) \left(\frac{(1-\tilde{\gamma})D+\sqrt{((1-\tilde{\gamma})D)^2-1}}{D+\sqrt{D^2-1}}\right)^{t} \, .
\end{equation*}
Finally, using \eqref{eq:aux-sup-0}, this gives 
\begin{equation*}
\Vert x^t-\bar{\xi}\bfone \Vert_2 \leq \Vert\xi-\bar{\xi}\bfone\Vert_2 C(d) \left(\frac{(1-\tilde{\gamma})D+\sqrt{((1-\tilde{\gamma})D)^2-1}}{D+\sqrt{D^2-1}}\right)^{t} \, .
\end{equation*}
Dividing the numerator and the denominator of the fraction by $D$, we get the desired result.

\smallskip
We now turn to the second part of the statement. Let $u$ be an eigenvector of $W$ corresponding to an eigenvalue $\lambda$ of magnitude $1-\tilde{\gamma}$ such that $\langle \xi, u \rangle \neq 0$. Then 
\begin{equation*}
\Vert x^t - \bar{\xi}\bfone\Vert_2 \geq |\langle \xi, u \rangle\vert \vert\pi_t(\lambda)\vert = |\langle \xi, u \rangle\vert \frac{\left\vert\tilde{p}_t(\lambda)\right\vert}{\left\vert\tilde{p}_t(D)\right\vert}\, ,
\end{equation*} 
Using as before Lemma \ref{lem:asymptotic_p}, we get the desired lower bound.

\smallskip
\textbf{\emph{(2)}} We now assume $\tilde{\gamma}\geq 1-2\sqrt{d-1}/d$. This means that $(1-\tilde{\gamma})D\leq1$, and thus 
\begin{equation*}
\sup_{\lambda \in [-(1-\tilde{\gamma})D, (1-\tilde{\gamma})D]} \vert\tilde{p}_t(\lambda)\vert \leq \sup_{\lambda \in [-1, 1]} \vert\tilde{p}_t(\lambda)\vert \stackrel{\eqref{eq:aux-sup-2}}{\leq} \left(\sqrt{d-1}+1\right)(t+1) +1 \, .
\end{equation*}
Combining with \eqref{eq:aux-sup-0} and \eqref{eq:aux-sup-1}, we get 
\begin{equation*}
\Vert x^t - \xi\bfone \Vert_2 \leq \Vert\xi-\bar{\xi}\bfone\Vert_2 \frac{1}{\vert\tilde{p}_t(D)\vert}\left(\left(\sqrt{d-1}+1\right)(t+1) +1\right)\, ,
\end{equation*}
which gives the desired result using Lemma \ref{lem:asymptotic_p}.
\end{document}